\documentclass[letterpaper,UKenglish,cleveref, autoref, thm-restate]{article}
\pdfoutput=1
\usepackage[left=1in,right=1in,top=3cm,bottom=2.5cm]{geometry}
\usepackage[T1]{fontenc}
\usepackage{graphicx} 
\usepackage{lineno}
\usepackage{url}
\usepackage{float}
\usepackage{amsfonts}
\usepackage{amsmath}
\usepackage{bm}
\usepackage{amssymb}
\usepackage{amsthm}
\usepackage{orcidlink}
\usepackage{algorithm}
\usepackage[noend]{algpseudocode}
\usepackage{tikzlings}
\usepackage{orcidlink}
\usepackage{tikzducks}
\usepackage{tikzsymbols}
\usepackage{tikz}
\usepackage{tcolorbox}
\usepackage{mathtools}
\usepackage{thm-restate}
\usepackage{hyperref}
\usepackage[capitalise]{cleveref}
\usepackage{xspace}
\usepackage[new]{old-arrows}
\usepackage{float}
\usepackage{wrapfig}
\usepackage{complexity}
\usepackage{caption}
\usepackage{subcaption}
\usepackage{graphicx}
\usepackage{tcolorbox}
\usepackage{mathabx}

\newtheorem*{question}{Question}

\tcbuselibrary{theorems}
\newtcbtheorem
[]
  {mydef}
  {}
  {%
    colback=gray!5,
    colframe=gray!35!black,
 }
  {SR}

\usetikzlibrary{positioning}

\newcommand{\Func}[1]{{\mathsf{#1}}}

\newcommand{\col}{\Func{col}}
\newcommand{\Win}{\texttt{Win}}
\newcommand{\nextmove}{\Func{nextmove}}
\newcommand{\Distribution}{\Func{Distributions}}
\newcommand{\unitinterval}{\mathbb{U}}

\newcommand{\tpl}[1]{\left( #1 \right)}

\newcommand{\eve}{\pmb{\exists}}
\newcommand{\adam}{\pmb{\forall}}
\newcommand{\Gc}{\mathcal{G}} 
\newcommand{\Lc}{\mathcal{L}}
\newcommand{\Ac}{\mathcal{A}}
\newcommand{\Fc}{\mathcal{F}}
\newcommand{\Hc}{\mathcal{H}}
\newcommand{\Bc}{\mathcal{B}}
\newcommand{\Dc}{\mathcal{D}}
\newcommand{\Mc}{\mathcal{M}}
\newcommand{\Tc}{\mathcal{T}}
\newcommand{\Oh}{\mathcal{O}}

\newcommand{\Nc}{\mathcal{N}}
\newcommand{\Nb}{\mathbb{N}}
\newcommand{\Sc}{\mathcal{S}}
\newcommand{\Rc}{\mathcal{R}}
\newcommand{\Cc}{\mathcal{C}}
\newcommand{\Oc}{\mathcal{O}}
\newcommand{\Pc}{\mathcal{P}}

\newcommand{\Zc}{\mathcal{Z}}

\newcommand{\Act}{\Func{Act}}

\newcommand{\prob}{\mathtt{Prob}}

\renewcommand{\inf}{\Func{Inf}}
\newclass{\ptime}{P}
\newclass{\pspace}{PSPACE}
\newclass{\psps}{PSPACE}
\newclass{\logspace}{LOGSPACE}
\newcommand{\safety}{\texttt{safe}}

\newcommand{\CR}{{\mathsf{CR}}}
\newcommand{\WCR}{\mathsf{WCR}}

\newclass{\TOWER}{TOWER}
\newclass{\ACKERMANN}{ACKERMANN}
\newclass{\EXPTIME}{EXPTIME}
\newclass{\IIEXPTIME}{2-EXPTIME}
\newclass{\NEXPTIME}{NEXPTIME}

\usepackage{graphicx}
\usepackage{array}
\usepackage{tikz}
\usetikzlibrary{automata, arrows, positioning, decorations.markings, decorations.pathreplacing}
\usetikzlibrary {decorations.pathmorphing, decorations.shapes}
\usetikzlibrary{backgrounds,fit,shapes.geometric}
\usetikzlibrary {graphs, quotes}
\usetikzlibrary{calc}
\usetikzlibrary{shapes,shapes.geometric,fit}
\usepackage[new]{old-arrows}
\usetikzlibrary {graphs,shapes.geometric,quotes}
\usetikzlibrary{decorations.pathreplacing}
\usetikzlibrary{automata}
\usetikzlibrary{calc}
\usetikzlibrary{arrows,automata,decorations.pathmorphing}
\usetikzlibrary{shapes,shapes.geometric,fit,calc,automata}
\usepackage{xargs}

\newcommand{\swapletter}{
  \begin{tikzpicture}[baseline=(base)]
    \node (base) at (0, -0.04) {}; 
    \node (dot1) at (-0.08, 0) {\tikz\draw[fill] (0,0) circle [radius=1pt];};
    \node (dot2) at (-0.08, 0.17) {\tikz\draw[fill] (0,0) circle [radius=1pt];};
    \node (dot3) at (0.08, 0) {\tikz\draw[fill] (0,0) circle [radius=1pt];};
    \node (dot4) at (0.08, 0.17) {\tikz\draw[fill] (0,0) circle [radius=1pt];};
    \draw (-0.08, 0.17) -- (0.08, 0);
    \draw (-0.08, 0) -- (0.08, 0.17);
  \end{tikzpicture}
}

\newcommand{\sameletter}{
  \begin{tikzpicture}[baseline=(base)]
    \node (base) at (0, -0.04) {}; 
    \node (dot1) at (-0.08, 0) {\tikz\draw[fill] (0,0) circle [radius=1pt];};
    \node (dot2) at (-0.08, 0.17) {\tikz\draw[fill] (0,0) circle [radius=1pt];};
    \node (dot3) at (0.08, 0) {\tikz\draw[fill] (0,0) circle [radius=1pt];};
    \node (dot4) at (0.08, 0.17) {\tikz\draw[fill] (0,0) circle [radius=1pt];};
    \draw (-0.08, 0.17) -- (0.08, 0.17);
    \draw (-0.08, 0) -- (0.08, 0);
  \end{tikzpicture}
}
\newcommand{\upletter}{
  \begin{tikzpicture}[baseline=(base)]
    \node (base) at (0, -0.04) {}; 
    \node (dot1) at (-0.08, 0) {\tikz\draw[fill] (0,0) circle [radius=1pt];};
    \node (dot2) at (-0.08, 0.17) {\tikz\draw[fill] (0,0) circle [radius=1pt];};
    \node (dot3) at (0.08, 0) {\tikz\draw[fill] (0,0) circle [radius=1pt];};
    \node (dot4) at (0.08, 0.17) {\tikz\draw[fill] (0,0) circle [radius=1pt];};
    \draw (-0.08, 0.17) -- (0.08, 0.17);
  \end{tikzpicture}
}
\newcommand{\downletter}{
  \begin{tikzpicture}[baseline=(base)]
    \node (base) at (0, -0.04) {}; 
    \node (dot1) at (-0.08, 0) {\tikz\draw[fill] (0,0) circle [radius=1pt];};
    \node (dot2) at (-0.08, 0.17) {\tikz\draw[fill] (0,0) circle [radius=1pt];};
    \node (dot3) at (0.08, 0) {\tikz\draw[fill] (0,0) circle [radius=1pt];};
    \node (dot4) at (0.08, 0.17) {\tikz\draw[fill] (0,0) circle [radius=1pt];};
    \draw (-0.08, 0) -- (0.08, 0);
  \end{tikzpicture}
}

\newcommandx{\theju}[1]{{\color{red} #1 -Thejaswini}}
\newcommandx{\ap}[1]{{\color{red} #1 -AP}}

\newsavebox{\mycandle}
\savebox{\mycandle}{ 
\begin{tikzpicture}[scale=.1]
  \shade[top color=yellow,bottom color=red] (0,0) .. controls (1,.15)
  and (1,.3) .. (0,2.5) .. controls (-1,.3) and (-1,.15) .. (0,0);
  \fill[red!70!blue] (.4,0) rectangle (-.4,-5);
\end{tikzpicture} }

\nolinenumbers
\title{Resolving Nondeterminism with Randomness\footnote{This work is a part of project VAMOS that has received funding from the European Research Council (ERC), grant agreement No 101020093}}

\author{\hyperlink{https://pub.ista.ac.at/~tah/}{Thomas A. Henzinger}
\orcidlink{0000-0002-2985-7724}
\\ IST Austria \\ \texttt{tah@ist.ac.at} \and \hyperlink{https://apitya.github.io/}{Aditya Prakash}
\orcidlink{0000-0002-2404-0707}\thanks{Funded by the Chancellors International Scholarship at the University of Warwick.} \\ University of Warwick, UK \\ \texttt{aditya.prakash@warwick.ac.uk} 
\and \hyperlink{https://thejaswiniraghavan.github.io/}{K. S. Thejaswini} \orcidlink{0000-0001-6077-7514}
\\IST Austria \\ \texttt{thejaswini.k.s@ista.ac.at}}

\newtheorem{theorem}{Theorem}
\newtheorem{lemma}[theorem]{Lemma}
\newtheorem{observation}[theorem]{Observation}
\newtheorem{proposition}[theorem]{Proposition}
\newtheorem{corollary}[theorem]{Corollary}
\newtheorem{example}[theorem]{Example}
\newtheorem{definition}[theorem]{Definition}
\newtheorem{claim}[theorem]{Claim}

\date{}
\begin{document}

\maketitle
\begin{abstract}
In automata theory, nondeterminism is a fundamental paradigm that can offer succinctness and expressivity, often at the cost of computational complexity.  While determinisation provides a standard route to solving many common problems in automata theory, some weak forms of nondeterminism can be dealt with in some problems without costly determinisation.  For example, the handling of specifications given by nondeterministic automata over infinite words for the problems of reactive synthesis or runtime verification requires resolving nondeterministic choices without knowing the future of the input word.  We define and study classes of $\omega$-regular automata for which the nondeterminism can be resolved by a policy that uses a combination of memory and randomness on any input word, based solely on the prefix read so far.

We examine two settings for providing the input word to an automaton.  In the first setting, called \emph{adversarial resolvability}, the input word is constructed letter-by-letter by an adversary, dependent on the resolver’s previous decisions.  In the second setting, called \emph{stochastic resolvability}, the adversary pre-commits to an infinite word and reveals it letter-by-letter.  In each setting, we require the existence of an almost-sure resolver, i.e., a policy that ensures that as long as the adversary provides a word in the language of the underlying nondeterministic automaton, the run constructed by the policy is accepting with probability~1.

The class of automata that are adversarially resolvable is the well-studied class of history-deterministic automata.  The case of stochastically resolvable automata, on the other hand, defines a novel class.  Restricting the class of resolvers in both settings to stochastic policies without memory introduces two additional new classes of automata.  We show that the new automaton classes offer interesting trade-offs between succinctness, expressivity, and computational complexity, providing a fine gradation between deterministic automata and nondeterministic automata.
\end{abstract}
\newpage
\tableofcontents
\newpage
\section{Introduction}\label{sec:intro}
The trade-off between determinism and nondeterminism is a central theme in automata theory. For automata over infinite words, nondeterministic B\"uchi automata are more expressive and succinct than their deterministic counterpart and, in fact, are as expressive as deterministic parity automata and exponentially more succinct~\cite{McN66,Saf88}. Even though nondeterminism seems attractive because of these favourable qualities, it presents difficulties in contexts like reactive synthesis or runtime verification when the specifications are expressed as automata. The
fundamental challenge here arises because any algorithm operating with nondeterministic automata needs to account for any possible future inputs.

Some attempts have evolved that involve modifying algorithms to avoid or minimise the blow-up using determinisation procedures~\cite{KPV06,KV05,EKS16}, while other attempts instead focus on building classes of automata that bridge this gap between determinism and nondeterminism to obtain the best of both worlds, while avoiding exponential determinisation procedures. 
The latter of the attempts include defining several classes of automata like history-deterministic (HD) automata~\cite{HP06}, good-for-MDP automata~\cite{HPSS0W20}, semantically deterministic automata~\cite{KS15,AK20}, and explorable automata~\cite{HK23} to name a few. 

The notion of history-determinism, which has inspired this work, has gained significant attention in recent years, since it allows for on-the-fly resolution of nondeterminism, which make them relevant for the problems of model-checking, and reactive synthesis~\cite[Page 22]{BL23}. 

Notably, history-deterministic automata are exponentially more succinct than deterministic automata \cite[Theorem 1]{KS15}, and history-deterministic coB\"uchi automata give rise to a canonical representation of coB\"uchi as well as $\omega$-regular languages~\cite{AK22,ES22}. History-deterministic automata are nondeterministic automata whose nondeterminism can be resolved on-the-fly and only based on the prefix of a word read so far. 
More precisely, history-determinism of an automaton can be characterised by the following history-determinism game (HD game) on that automaton. The HD game is played between two players, Eve and Adam, on an arena of the automaton where Eve initially has a token at the start state. In each round of a play, Adam picks a letter, and Eve moves her token along a transition on the letter picked by Adam, thereby constructing a run on the word that Adam has picked so far. The game proceeds for an infinite duration, where Adam constructs an infinite word, and Eve produces a run on this word. Eve wins if she has a strategy that ensures that she constructs an accepting run whenever the infinite word given by Adam is in the language of the automaton. HD automata are defined as  automata in which Eve wins the HD game. For a sanity check, consider deterministic automata, where Eve trivially wins the HD game by choosing the unique deterministic transition on each letter to produce an accepting run whenever the word is in the language. 

Such a definition of resolving nondeterministic choices as in HD automata is useful to represent winning conditions in games against an adversary, thus also being dubbed good-for-games automata at the time of its introduction~\cite[Theorem 3.1]{HP06}. This property also makes history-deterministic automata relevant for the problems of reactive synthesis and model checking~\cite[Section 7]{BL23}.

Although history determinism is a useful concept, it is based on a highly adversarial setting for the player who resolves the nondeterminism. 
For example, consider the infinite word reachability automaton in~\cref{fig:ReachSRbutnotHD}, which accepts all infinite words over $\{a,b\}$.
In order to resolve the nondeterministic decision in state $q_0$ and eventually reach state $q_f$, the letter that is to be read in the next step must be guessed.  
If successive letters are chosen adversarially and based on the current state of the automaton, no resolver can successfully resolve the nondeterminism on all words, and therefore this automaton is not HD. 
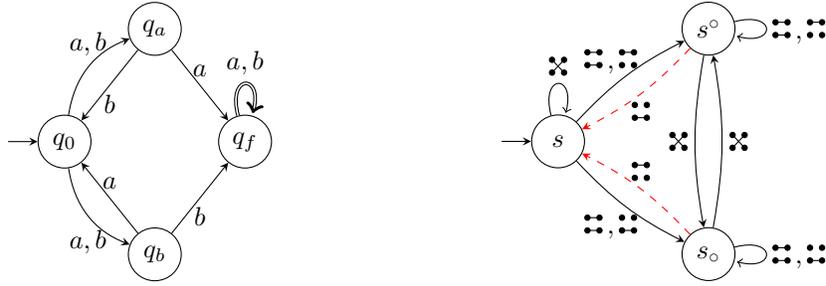
\begin{figure}[ht]
    \centering
    \begin{subfigure}[t]{0.4\textwidth}
    \centering
    \begin{tikzpicture}
        \tikzset{every state/.style = {inner sep=-3pt,minimum size =20}}

    \node[state] (q0) at (0,0) {$q_0$};
    \node[state] (qa) at (1.2,1.5) {$q_a$};
    \node[state] (qb) at (1.2,-1.5) {$q_b$};
    \node[state] (qf) at (2.4,0) {$q_f$};
    \path[-stealth]
    (-0.75,0) edge (q0)

    (q0) edge [bend left = 30] node [above, yshift = 1] {$a,b$} (qa)
    (q0) edge [bend right = 30] node [below, yshift = -2] {$a,b$} (qb)
    (qa) edge node [below] {$b$} (q0) 
    (qb) edge node [above] {$a$} (q0)
    (qa) edge node [above] {$a$} (qf)
    (qb) edge node [below] {$b$} (qf)
    (qf) edge [double, loop above] node [above] {$a,b$} (qf)
    ;
    \end{tikzpicture}
    \caption{A reachability automaton that is not history deterministic.}
    \label{fig:ReachSRbutnotHD}
    \end{subfigure}~~~
    \begin{subfigure}[t]{0.4\textwidth}
    \centering
        \begin{tikzpicture}
        \tikzset{every state/.style = {inner sep=-3pt,minimum size =20}}

    \node[state] (q0) at (0,0) {$s$};
    \node[state] (qa) at (2,1.5) {$s^\circ$};
    \node[state] (qb) at (2,-1.5) {$s_\circ$};
    \path[-stealth]
    (-0.75,0) edge (q0)
    (q0) edge [bend left = 10] (qa)
    (q0) edge [bend right = 10] (qb)
    (qa) edge[red, dashed,bend left =10] (q0) 
    (qb) edge[red, dashed,bend right =10]  (q0)
    (qa) edge[bend right =10]  (qb)
    (qb) edge[bend right =10]  (qa)
    (q0) edge[loop above] (q0)
    (qa) edge[loop right] (qa)
    (qb) edge[loop right] (qb)
    ;
    \node (l1) at (0,1) {$\swapletter$};
    \node (l2) at (2.4,0) {$\swapletter$};
    \node (l3) at (1.6,0) {$\swapletter$};

    \node (l5) at (3.2,1.5) {$\sameletter\!,\!\!\upletter$};
    \node (l6) at (3.2,-1.5) {$\sameletter\!,\!\!\downletter$};  
    
    \node (l4) at (0.7,1.1) {$\sameletter\!,\!\!\upletter$};
    \node (l7) at (0.7,-1.1) {$\sameletter\!,\!\!\downletter$};   
    
    \node (l8) at (1.1,0.4) {$\downletter$};
    \node (l9) at (1.1,-0.4) {$\upletter$};   
    \end{tikzpicture}
    \caption{A HD coB\"uchi automaton where nondeterminism can be resolved randomly}
    \label{fig:coBuchiEx}
    \end{subfigure}
    \caption{Examples of resolving nondeterminism with randomness}
\end{figure}
However, in a more lenient setting where the adversary pre-commits to an infinite word, a resolver that resolves the nondeterminism by choosing uniformly at random one of the two transitions from $q_0$ would produce an accepting run almost surely, that is, with probability 1. 

Allowing for stochastic resolvers simplifies the memory structure of the resolvers even for HD automata. 
Consider the coB\"uchi automaton in \cref{fig:coBuchiEx} inspired from the work of Kuperberg and Skrzypczak~\cite[Theorem~1]{KS15}. Accepting runs in a coB\"uchi automaton must visit rejecting (dashed) transitions only finitely often. The language of the auomaton is over the alphabet $\{\!\!\swapletter\!,\!\sameletter\!,\!\upletter\!,\!\downletter\!\!\}$ and therefore, each infinite word naturally constructs an infinite graph over  $\Nb\times \{1,2\}$. 
We demonstrate with the finite word $\!\!\swapletter\!\!\!\sameletter\!\!\!\sameletter\!\!\!\swapletter\!\!$ that represents a graph with two distinct finite paths of length 4, if the right-most dots of a letter are identified with the left-most dots of the next letter in the word.  
The infinite word 
$(\!\!\swapletter\!\!\!\sameletter\!\!\!\sameletter\!\!\!\swapletter\!\!)^\omega$ represents a graph with two infinite paths. 
However, the word
$(\!\!\swapletter\!\!\!\upletter\!\!\!\sameletter\!\!)^\omega$ does not have any infinite path, as the letter $\!\!\upletter\!\!$ breaks each path infinitely often. 
It can be verified that the nondeterministic coB\"uchi automaton in \cref{fig:coBuchiEx} accepts an infinite word if and only if the corresponding graph has at least one infinite path.  
Indeed, the two nondeterministic choices correspond to verifying if the path starting from the bottom vertex is infinite or the top vertex is infinite, respectively. A correct guess ensures that rejecting transitions never occur in the run and a wrong guess brings the run back to the start vertex from which a guess needs to be made again. 

This automaton is history-deterministic. Consider the resolver that chooses the transition from state $s$ that follows the longest unbroken path so far to verify that it extends to an infinite path. If there is an infinite path then eventually it would be the longest unbroken path, and such a resolver would correctly resolve the nondeterminism to produce an accepting run. Note that any resolver that selects a transition from $s$ uniformly at random also constructs an accepting run on any infinite word in the language with probability~$1$, even if letters are chosen adversarially and based on the current state. This is the case because the run of the resolver would, with probability~$1$, eventually coincide with one of the longest unbroken paths. 


Motivated by these examples, we introduce classes of automata inspired by history determinism, extending them in two ways. First, we consider resolvers that use not only memory but also stochasticity, and second, we also study settings where the letter-giving adversary commits to a word in advance.

\paragraph*{New classes of automata} Intuitively, stochastic resolvers propose an outgoing transition using both memory and randomness.
We study classes of automata for which there is an \emph{almost-sure} resolver, a policy for resolving the nondeterminism (which can use memory or randomness) such that for all words in the language, the run produced is almost-surely accepting, that is, with probability~1.

We consider the existence of almost-sure resolvers in the following two settings for the letter-giving adversary.  The first setting, called \emph{adversarial resolvability}, is where the input words to the automaton are from an adversary that generates the input letter-by-letter based on the resolution of the nondeterminism in the past. This definition is reminiscent to the definition of resolvers in history determinism, where now we allow for stochasticity. Indeed, the resolvers here represent the strategy of Eve in the HD game defined for history determinism, where we allow for a larger class of resolvers which allow for randomness. 
The second, and novel setting, called \emph{stochastic resolvability}, is where the adversary commits to an entire input-word, but only reveals the input letter-by-letter to their opponent who is resolving the nondeterminism. 

 We call the class of automata that have an almost-sure resolver in the adversarial resolvability setting as \emph{adversarially resolvable automata}, and that have an almost-sure resolver in the stochastic resolvability setting as \emph{stochastically resolvable automata}, or SR automata, for short.
We further study the classes of automata where ``weaker'' resolvers are used. If the resolvers of nondeterminism are restricted to policies where no memory and only stochasticity is used, then the resolving strategy is just a probability distribution among the outgoing transition for each state. We call such classes of resolvers \emph{memoryless resolvers}, and the class of automata that are adversarially resolvable using memoryless almost-sure resolvers as \emph{memoryless-adversarially resolvable}, or MA for short. Likewise, we call the class of automata that are stochastically resolvable using memoryless almost-sure resolvers \emph{memoryless-stochastically resolvable}, or MR for short. The class of adversarially resolvable automata, without any restrictions on resolvers, are equivalent to HD automata due to determinacy of $\omega$-regular games.

\paragraph*{Our results}
We introduce and then make comparisons between the three newly introduced classes of automata and also with existing notions such as HD automata and semantically deterministic (SD) automata (automata where every nondeterministic transition leads to language equivalent states)~\cite{AK23}, and we study the questions of succinctness, expressivity, and computational complexity of the class-membership problem for these automata classes. 


In \cref{sec:succandcomp}, we compare our novel automata classes with each other and also with history-deterministic and semantically deterministic automata. We show separating examples or prove the equivalences between the classes (\cref{theorem:comp}). A landscape of these automata can be found in \cref{fig:venndiagrammmmmm}.
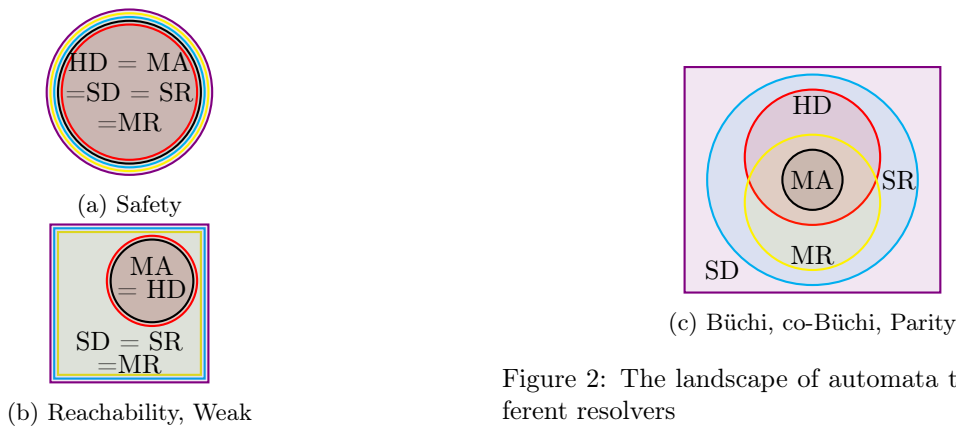
\begin{figure}[h]
    \centering
    \begin{minipage}[b]{0.4\linewidth}
            \begin{subfigure}[b]{\linewidth}
        \centering
        \begin{tikzpicture}
            \draw[thick,violet] (0,0) circle (1.1cm);
            \fill[violet, opacity=0.1] (0,0) circle (1.1cm);
            
            \draw[thick,yellow] (0,0) circle (1.05cm);
            \fill[yellow, opacity=0.1] (0,0) circle (1.05cm);
            
            \draw[thick,cyan] (0,0) circle (1cm);
            \fill[cyan, opacity=0.1] (0,0) circle (1cm); 

            \fill[black, opacity=0.1] (0,0) circle (0.95cm);;
            \draw[thick,black] (0,0) circle (0.95cm);
            
            \draw[thick,red] (0,0) circle (0.9cm);
            \fill[red, opacity=0.1] (0,0) circle (0.9cm);

            \node at (0,0.4) {HD = MA};
            \node at (0,0) {=SD = SR};
            \node at (0,-0.4) {=MR};
        \end{tikzpicture}
    \caption{Safety}
    \end{subfigure}%
    \vfill
    \begin{subfigure}[b]{\linewidth}
        \centering
        \begin{tikzpicture}
            \draw[thick,cyan] (-1,1) rectangle (1,-1);
            \draw[thick,yellow] (-0.95,0.95) rectangle (0.95,-0.95); 
            \fill[violet, opacity=0.1] (-1.05,1.05) rectangle (1.05,-1.05);
            \fill[cyan, opacity=0.1] (-1,1) rectangle (1,-1); 
            \fill[yellow, opacity=0.1] (-0.95,0.95) rectangle (0.95,-0.95);
            \draw[thick,violet] (-1.05,1.05) rectangle (1.05,-1.05);

            \draw[thick,red] (0.3,0.3) circle (0.6cm);
            \fill[red, opacity=0.1] (0.3,0.3) circle (0.6cm);
            \draw[thick,black] (0.3,0.3) circle (0.55cm);
            \fill[black, opacity=0.1] (0.3,0.3) circle (0.55cm);
            \node at (0,-0.5) {SD = SR};
            \node at (0,-0.8) {=MR};
            \node at (0.3,0.5) {MA};
            \node at (0.3,0.2) {= HD};
        \end{tikzpicture}
    \caption{Reachability, Weak}
    \end{subfigure}%
    \end{minipage}
    \hfill
    \begin{minipage}[b]{0.5\linewidth}
        \begin{subfigure}[b]{\linewidth}
        \centering
        \begin{tikzpicture}
            \draw[thick,violet] (-1.7,1.5) rectangle (1.7,-1.5);
            \fill[violet, opacity=0.1] (-1.7,1.5) rectangle (1.7,-1.5);
            \draw[thick,cyan] (0,0) circle (1.4cm);
            \fill[cyan, opacity=0.1] (0,0) circle (1.4cm); 
            \draw[thick,red] (0,0.3) circle (0.9cm); 
            \fill[red, opacity=0.1] (0,0.3) circle (0.9cm);
            \draw[thick,yellow] (0,-0.3) circle (0.9cm); 
            \fill[yellow, opacity=0.1] (0,-0.3) circle (0.9cm);
            \fill[black,opacity=0.1] (0,0) circle (0.4cm);
            \draw[black,thick] (0,0) circle (0.4cm); 
            \node at (-1.2,-1.2) {SD};
            \node at (1.15,0) {SR};
            \node at (0,1) {HD};
            \node at (0,-1) {MR};
            \node at (0,0) {MA};
        \end{tikzpicture}    \caption{B\"uchi, co-B\"uchi, Parity}
    \end{subfigure}
    \caption{The landscape of automata that admit different resolvers}\label{fig:venndiagrammmmmm}
    \end{minipage}
\end{figure}
In the same section, we show \emph{exponential succinctness of SR coB\"uchi to deterministic B\"uchi automata and SR B\"uchi automata to HD B\"uchi automata (\cref{theorem:succ})} 
For coB\"uchi automata on the other hand, we show a surprising result that \emph{any SR coB\"uchi automaton can be converted into an MA automaton with the same number of states} (\cref{theorem:coBuchiHDisSR}). This result, therefore, shows that SR coB\"uchi automata are no more succinct than HD coB\"uchi~automata. 

In \cref{sec:pih}, we turn our attention to the problem of expressivity.  History-deterministic (and therefore also MA) parity automata that use priorities from $\{i,i+1,\dots,j\}$ are as expressive as their deterministic counterparts~\cite{BKS17}. We show that SR and MR automata are similar to HD automata in this regard by showing that \emph{the parity index hierarchy for SR automata is strict} (\cref{theorem:pih}). 

Finally, in \cref{sec:complexity}, we tackle the problems of deciding if an automata is MA and SR. The exact complexity of the problem of deciding if an automaton is HD is an open problem since the introduction of the class in 2006~\cite{HP06}, with recent improvements giving the problem a $\NP$ lower bound~\cite{Pra24a} and a $\PSPACE$ upper bound~\cite[2-Token Theorem]{LP25, Pra25}.
We show that the situation for MA automata is better as we show that \emph{checking if an automaton is MA is $\NP$-complete} (\cref{theorem:manpcomplete}). The upper bound relies on showing that a slight modification of the 2-token game~\cite{BK18,LP25}, a game used to characterise history deterministic automata, also characterises MA automata. With such a characterisation, we show that to check if an automata is MA, one first needs to guess a strategy in the modified two-token game. Checking correctness of this strategy reduces to solving a Markov decision process (MDP) with Muller objective, where objective can be represented succinctly using a Zielonka DAG. We show that this can be computed in polynomial time (\cref{thm:ZlkDAGMDP}). This result on MDPs was only previously known for less succinct representations of Muller objectives~\cite{Cha07}.  We also consider problems related to checking whether an automaton is in the class MR and SR. We summarise the results of the decision procedures discussed in~\cref{table:TomwantsTables}.

\begin{table*}[t]
\centering
\begin{tabular}{| l || c | c | c | c | c | c | }
\hline
 & Safety & Reachability/Weak & Buchi & coB\"uchi & Parity\\
\hline\hline
\parbox{3cm}{Checking MA (\cref{theorem:manpcomplete})} & $\P$ & $\P$   & $\NP$ & $\NP$  & $\NP$-complete \\
\hline
Checking HD & $\P$~\cite{KS15,BL23quantitative} & $\P$~\cite{KS15,BL23quantitative}    & $\P$~\cite{KS15} & $\P$~\cite{BK18}  & \parbox{2.5cm}{$\PSPACE$~\cite{LP25}, $\NP$-hard~\cite{Pra24a}} \\
\hline
\parbox{3cm}{Checking SR/ MR (\cref{theorem:complexity})} & $\P$ & $\PSPACE$-comp  & Open & Open & Open\\ 			\hline
\parbox{3cm}{Resolver checking for SR/MR (\cref{theorem:complexity})} & $\P$ & $\PSPACE$-comp  & undec. & undec. & undec.\\
  \hline 
\end{tabular}
\caption{Complexity of checking membership of an automaton in a class.}
\label{table:TomwantsTables}
\end{table*}

\section{Automata and resolution of nondeterminism}
\subsection{Automata}
We use $\mathbb{N}$ to denote the set of natural numbers $\{0,1,2,\dots\}$. For two natural numbers $i,j$ with $i<j$, we shall use $[i,j]$ to denote the set $\{i,i+1,\dots,j\}$, consisting of natural numbers that are at least $i$ and at most $j$. For a natural number $i$,  we use $[i]$ to denote the set $[0,i]$. 
We use $\unitinterval$ to denote the unit interval, i.e., $\unitinterval = \{x \mid 0 \leq x \leq 1\}.$
For a finite set $X$, a probability distribution on $X$ is a function $f:X\xrightarrow{}\unitinterval$ that maps each element in $X$ to a number in $\unitinterval$, such that $\sum_{a\in X} f(X) = 1$. We denote the set of probability distributions on $X$ by $\Distribution(X)$. 

\paragraph*{Parity automata}
An $[i,j]$-\emph{nondeterministic parity automaton} $\Ac = (Q,\Sigma,\Delta,q_0)$, $[i,j]$-parity automaton for short, consists of a finite directed graph with edges labelled by letters in $\Sigma$ and \emph{priorities} in $[i,j]$ for some $i,j \in \mathbb{N}$ with $i\leq j$. 
The set of \emph{states} $Q$ constitutes the vertices of this graph, and the set of \emph{transitions} $\Delta \subseteq Q \times \Sigma \times [i,j] \times Q$ represents the labelled edges of the graph. 
Each automaton has a designated \emph{initial state} $q_0 \in Q$. 
For states $p,q \in Q$ and a letter $a \in \Sigma$, we use $p\xrightarrow{a:c}q$ to denote a transition from $p$ to $q$ on the letter $a$ that has the priority $c$. We assume our automata to be \emph{complete}, i.e., for each state and letter, there is a transition from that state on that letter.

A \emph{run} on an infinite word $w$ in $\Sigma^{\omega}$ is an infinite path in the automaton, starting at the initial state and following transitions that correspond to the letters of $w$ in sequence.  We say that such a run is \emph{accepting} if the highest priority occurring infinitely often amongst the transitions of that run is even, and a word $w$ in $\Sigma^{\omega}$ is accepting if the automaton has an accepting run on $w$. The \emph{language} of an automaton $\Ac$, denoted by $\Lc(\Ac)$, is the set of words that it accepts. We say that the automaton $\Ac$ \emph{recognises} a language $L$ if $\Lc(\Ac)=L$. A parity automaton $\Ac$ is said to be \emph{deterministic} if for any given state in $\Ac$ and any given letter in $\Sigma$, there is exactly one transition from the given state on the given letter.

We will say that $[i,j]$ with $i=0$ or $1$ is the parity index of $\Ac$. A \emph{B\"uchi} (resp. co-B\"uchi) automaton is a $[1,2]$ (resp. $[0,1]$) parity automaton. A \emph{safety automaton} is a B\"uchi automaton where all  transitions with parity $1$ occur as self-loops on a sink state. Dually, a \emph{reachability automaton} is a B\"uchi automaton $\Ac$ such that all transitions with parity $2$ in $\Ac$ occur as self-loops on a sink state. A \emph{weak automaton} is a B\"uchi automaton, in which there is no cycle that contains both an accepting and a rejecting transition.

We write $(\Ac,q)$ to denote the automaton $\Ac$ with $q$ instead as its initial state, and $\Lc(\Ac,q)$ to denote the language it recognises. We say two states $p$ and $q$ in $\Ac$ are language-equivalent if $\Lc(\Ac,p) = \Lc(\Ac,q)$.

\paragraph*{Probabilistic automata}
A probabilistic parity automaton $\Pc = (Q,\Sigma,\Delta,\rho,q_0)$---a natural extension of probabilistic B\"uchi automata defined by Baier, Gr\"o\ss er, and Bertrand~\cite{BGB12}---has the semantics of a parity automaton. Additionally, we assign a probability to each transition in $\Delta$ using the function $\rho: \Delta \xrightarrow{} \unitinterval$, such that for each state and each letter, the sum of $\rho(\delta)$s for outgoing transitions $\delta$ from that state on that letter add up to 1. We write $\Delta_{q,a}$ to denote the set of outgoing transitions from  the state $q$ on the letter~$a$.

Given a probabilistic automaton $\Pc$ as above, the behaviour of $\Pc$ on an input word $w$ is as follows: each transition $\delta$ from a given state on a given letter can be taken with probability $\mu(\delta)$. We formalise this by an infinite Markov chain that captures all the possible runs of $\Pc$ on $w$. For the word $w=a_0 a_1 a_2 \dots$, consider the Markov chain $M_w$ defined over the vertices $Q \times \mathbb{N}$. After `processing' the finite word $a_0 a_1 \dots a_{i-1}$, the Markov chain will be at some state $(q,i)$, where $q$ is a state that can be reached from $q_0$ on the word $a_0 a_1 \dots a_{i-1}$, and a run of the Markov chain moves from $(q,i)$ to the state $(p,i+1)$ with probability $\rho(\delta)$, where $\delta=q\xrightarrow{a_i:c_i}p$ is a transition in $\Ac$. The initial state of $M_w$ is $(q_0,a_0)$, and  we say that a run in $M_w$ is accepting if the corresponding run for $\rho$ in $\Ac$ is accepting.

For an input word $w$ and a probabilistic automaton $\Pc$, we define the probability $\prob_\Pc(w)$ to be the probability measure of accepting runs in $M_w$.  We mostly deal with almost-sure semantics of probabilistic automata, and therefore, we refer to language of a probabilistic automaton $\Pc$ as $\Lc(\Pc) = \{w\in \Sigma^\omega\mid \prob_\Pc(w) = 1\}$.

If the priorities occurring in a probabilistic parity automaton are from the set $[1,2]$ (resp.\ $[0,1]$), we call it a probabilistic B\"uchi (resp. coB\"uchi) automaton. 



\subsection{Resolution of nondeterminism}
We will deal with nondeterministic automata where the nondeterminism can be resolved using a combination of memory and randomness. Let us start by recalling the classical definition of history-determinism, which is characterised by the following history-determinism game.

\begin{definition}[History-determinism game]\label{defn:hd-game}
Given a nondeterministic parity automaton $\Ac = (Q,\Sigma, \Delta, q_0)$, the \emph{history-determinism (HD) game} on $\Ac$ is a two-player game between Eve and Adam that starts with Eve's token at $q_0$ and proceeds for infinitely many rounds. 
For each $i \in \mathbb{N}$, round $i$ starts with Eve's token at a state $q_i$ in $Q$, and proceeds as follows.
\begin{enumerate}
    \item Adam selects a letter $a_i \in \Sigma$;
    \item Eve selects a transition $q_i \xrightarrow{a_i:c_i} q_{i+1} \in \Delta$ along which she moves her token. Eve's token then is at $q_{i+1}$ from where the round $(i+1)$ is played. 
\end{enumerate} 
Thus, in the limit of a play of the HD game, Adam constructs a word letter-by-letter, and Eve constructs a run on her token transition-by-transition on that word. Eve wins such a play if the following condition holds: if Adam's word  is in $\Lc(\Ac)$, then the run on Eve's token is accepting.
\end{definition}
We say that an automaton is history deterministic (HD) if Eve has a winning strategy in the HD game on $\Ac$. History-determinism games are finite-memory determined~\cite[Theorem 3.12]{Pra25} and therefore, if an automaton $\Ac$ is HD then Eve has a finite memory winning strategy, that is, a mapping from the set of finite plays of the HD game to the next transition that Eve chooses, where the mapping is only based on the current position and a finite memory structure. We will formalise this concept by \emph{pure resolvers}. We first define the concept of resolvers where we allow for randomness as well.     

\paragraph*{Resolver}
For a nondeterministic parity automaton $\Ac=(Q,\Sigma,\Delta,q_0)$, a \emph{stochastic resolver}, or just \emph{resolver}, for $\Ac$ is given by $\Mc=(M,m_0,\mu,\nextmove)$, where $M$ is a finite set of \emph{memory states}, $m_0$ is the initial memory state. The function $\nextmove$ assigns to every three-tuple of memory state $m$, state $q$ of $\Ac$, and letter $a$ in $\Sigma$, a probability distribution $\nextmove(m,q,a)$ in $\Distribution(\Delta_{q,a})$, where $\Delta_{q,a}$ is the set of outgoing transitions from state $q$ on letter $a$. The function $\mu$ is the transition function and is given by $\mu:M\times \Delta \to M$.

Eve using a resolver $\Mc$, plays in the HD game on $\Ac$ as follows. At state $q$, when the memory state is $m$, suppose Adam chooses the letter $a$. Eve then selects an outgoing transition $\delta$ from $q$ on $a$ with the probability $(\nextmove(m,q,a)) \circ (\delta)$, and updates  to be $m'=\mu(m,\delta)$. 

We say that a resolver $\Mc$ for Eve constitutes a winning strategy in the HD game on $\Ac$ if she wins almost-surely when she plays using $\Mc$ as described above. We say that a resolver $\Mc$ is \emph{pure}, if for each memory state $m$, each state $q$, and letter $a$, the probability distribution $\nextmove(m,q,a)$ assigns probability $1$ to some outgoing transition in $\Delta_{q,a}$ and probability 0 to every other transition in $\Delta_{q,a}$. 

In $\omega$-regular games, which HD games are a subclass of, Eve has a winning strategy if and only if she has a pure winning strategy. Thus, if Eve has a stochastic resolver that is a winning strategy in the HD game on $\Ac$, then she also has a pure resolver. 
We say that a resolver is \emph{memoryless} if that resolver has one state. 

\begin{tcolorbox}
    An automaton $\Ac$ is \emph{memoryless-adversarial resolvable} (MA)  if there is a memoryless resolver for $\Ac$ using which Eve wins the HD game on~$\Ac$ almost-surely, i.e., with probability $1$.
\end{tcolorbox}

If there is a memoryless resolver $\Mc$ for Eve using which she wins the HD game with probability $1$, then any memoryless resolver $\Mc'$ that assigns nonzero probabilities to the same transitions as $\Mc$ can also be used by Eve to win the HD game on~$\Ac$ (\cref{lemma:indifferent2probabilities} in the appendix).

\paragraph*{Nonadversarial resolvability}
So far, we have discussed resolution of nondeterminism in automata where the letters are chosen adversarially. We next introduce the notions of nondeterminism where this is relaxed. For a parity automaton $\Ac$ and a resolver $\Mc$ for $\Ac$, we say that $\Mc$ is a \emph{almost-sure resolver} 
for $\Ac$ if for each word $w$, the run constructed on $w$ using $\Mc$ in the HD game on $\Ac$ is almost-surely accepting. 

More concretely, consider the probabilistic automaton $\Pc=\Mc\circ\Ac$ that is obtained by \emph{composing} the resolver $\Mc$ with~$\Ac$. That is, the states of $\Pc$ are $Q \times M$, the initial state of $\Pc$ is~$p_0=(q_0,m_0)$. The automaton $\Pc$ has the transition $(q,m)\xrightarrow{a:c}(q',m')$ of probability $p$ if $\delta=q\xrightarrow{a:c}q'$ is a transition in $\Ac$, $(\nextmove(m,q,a))\circ \delta = p$ and $m'=\mu(m,\delta)$. Then, we say that $\Mc$ is an \emph{almost-sure resolver} for $\Ac$  if $\Lc(\Ac)=\Lc(\Pc)$. For $\Mc$, $\Ac$, and $\Pc$ as above, we call $\Pc$ the \emph{resolver-product} of $\Mc$ and $\Ac$.   

\begin{tcolorbox}
   An automaton $\Ac$ is \emph{stochastically resolvable} (SR) if there is an almost-sure resolver for~$\Ac$.
\end{tcolorbox}
\begin{tcolorbox}
   An automaton $\Ac$ is \emph{memoryless stochastically resolvable} (MR) if there is an almost-sure resolver for~$\Ac$ that is memoryless.  
\end{tcolorbox}


\begin{example}\label{example:reachSRbutnotHD}
Consider the example of the reachability automaton $\Ac$ in \cref{fig:ReachSRbutnotHD}, where the accepting transitions are double-arrowed. We will show that $\Ac$ is MR but not HD. This automaton accepts all infinite words over the alphabet $\{a,b\}$. The only nondeterminism is at the initial state $q_0$.  At the state $q_0$ on reading $a$ or $b$,  if the resolver correctly guesses the next letter, then it will reach the accepting state. Consider a resolver that chooses the transitions to $q_a$ and $q_b$ with $\frac{1}{2}$ probability. Then for any fixed infinite word, this resolver produces an accepting run that will eventually, almost surely,  correctly guess the letter one step ahead.

However, in the HD game on $\Ac$, Adam can always adversarially choose the letter based on the resolution of the run. That is, whenever Eve in the HD game is at $q_b$, Adam chooses the letter $a$ and chooses the letter $b$ whenever Eve is at $q_a$. This adversarial choice of letters ensures that Eve's run is rejecting, and thus Adam wins the HD game on $\Ac$.
\end{example}
We give an equivalent definition for SR and MR automata based on the \emph{stochastic-resolvability} game, or the SR game for brevity. 
\paragraph*{Stochastic resolvability game} The SR game on an automaton $\Ac$ proceeds similarly to the HD game. In each round, Adam selects a letter and then Eve responds with a transition of $\Ac$ on that letter; thus, in the limit, Adam constructs an infinite word and Eve constructs a run on that word. 
Eve wins if her transitions form an accepting run whenever Adam's word is in the language. However, unlike the HD game, Adam does not observe Eve's run, and therefore his strategy must not depend on the position of Eve's token in the automaton. We define this game more formally and prove the following related result in the appendix.

\begin{restatable}{lemma}{randomispure}\label{lemma:random-is-pure}
For every parity automaton $\Ac$, a resolver $\Mc$ is an almost-sure resolver for $\Ac$ if and only if $\Mc$ is a finite-memory strategy that is almost-surely winning for Eve against all strategies of Adam in the SR game on $\Ac$.
\end{restatable}

\label{sec:prelims}
\section{Comparisons between the different notions}\label{sec:succandcomp}
In this section, we compare the notions of stochastic resolvability, memoryless adversarial resolvability, memoryless stochastic resolvability with each other and the existing notions of nondeterminism in the literature.

We start by making a connection to the notion of semantic determinism. These were introduced by Kuperberg and Skrzypczak as residual automata~\cite{KS15}, but we follow the more recent works of Abu Radi, Kupferman, and Leshkowitz in calling them semantically deterministic (SD) automata instead~\cite{AKL21,AK23}.
We call a transition $\delta$ from $p$ to $q$ on a letter $a$ in a parity automaton $\Ac$ as \emph{language-preserving} if $\Lc(\Ac,q) = a^{-1} \Lc(\Ac,p)$. We say that a parity automaton is \emph{semantically deterministic}, SD for short, if all transitions in that automaton are language-preserving. The following observation concerning semantically deterministic automata can be shown by a simple inductive argument on the length of words.

\begin{observation}\label{lemma:SDautomata}
For every semantically deterministic automaton $\Ac$, all states in $\Ac$ that can be reached from a state $q$ upon reading a finite word $u$ recognise the language $u^{-1}\Lc(\Ac,q)$.
\end{observation}

We observe that all SR automata are SD, up to removal of some transitions.
\begin{lemma}\label{lemma:SR-implies-SD}
Every stochastically resolvable parity automaton $\Ac$ has a language-equivalent subautomaton $\Bc$ that is semantically deterministic.
\end{lemma}
\begin{proof}
Let $\Mc$ be an almost-sure resolver for $\Ac$. Consider the set of transitions of $\Ac$ that are reachable from the initial state and taken in $\Mc \circ \Ac$ with nonzero probability: call these transitions \emph{feasible}. Let $\Bc$ be the subautomaton of $\Ac$ consisting of only feasible transitions. Since $\Mc\circ \Ac$ accepts the same language as $\Ac$, so does $ \Mc\circ \Bc$, and therefore, $\Lc(\Bc) = \Lc(\Ac)$. We will argue that $\Bc$ is semantically deterministic. Indeed, let $\delta=q\xrightarrow{a:c}q'$ be any transition in $\Bc$, and let $m,m'$ be states such that in $\Tc$, we have the transition $(q,m) \xrightarrow{a:c} (q',m')$ in $\Mc \circ \Ac$ that has non-zero probability. Let $u$ be a word such that $\Mc \circ \Ac$ has a run from from $(q_0,m_0)$ to $(q,m)$. For every word $aw \in \Lc(\Bc,q)$, we must have that the word $w$ is accepting in $(\Bc,q')$, since otherwise we have that the word $uaw$ is rejected with positive probability in $\Mc\circ \Ac$. It follows that $\Bc$ is semantically deterministic, as desired. 
\end{proof}

A \emph{pre-semantically deterministic} (pre-SD) automaton is an automaton that has a language-equivalent SD subautomaton. The following result gives a comprehensive comparison between the notions of nondeterminism we have discussed so far. The results of \cref{theorem:comp} are summarised by the Venn diagram in \cref{fig:venndiagrammmmmm} in \cref{sec:intro}.  

\begin{theorem}\label{theorem:comp}
\begin{enumerate}
    \item For safety automata, the notions of pre-semantic determinism, stochastic resolvability, memoryless-adversarial resolvability, memoryless-stochastic resolvability, and history-determinism are equivalent.
    \item For reachability and weak automata, the following statements hold.
    \begin{enumerate}
        \item Pre-semantic-determinism, stochastic resolvability, and memoryless stochastic resolvability are equivalent and are strictly larger classes than history-deterministic automata.
        \item History-determinism and memoryless adversarial resolvability are equivalent notions.
    \end{enumerate}
    \item For B\"uchi, coB\"uchi, and parity automata, the following statements hold.
    \begin{enumerate}
        \item Pre-semantically deterministic automata are a strictly larger class than stochastically resolvable automata.
        \item Stochastically resolvable automata are a strictly larger class than history-deterministic automata.
        \item Stochastically resolvable automata are a strictly larger class than memoryless stochastically resolvable automata.
        \item There are history-deterministic automata that are not memoryless stochastically resolvable, and there are memoryless stochastically resolvable automata that are not history-deterministic.
        \item Both history-deterministic and memoryless-stochastically resolvable automata are strictly larger classes than memoryless-adversarially resolvable automata.
    \end{enumerate}
\end{enumerate}
\end{theorem}

Even though the five notions of nondeterminism discussed for coB\"uchi automata are all different, we show that every stochastically resolvable coB\"uchi automaton can be efficiently converted to a language-equivalent memoryless-adversarially resolvable coB\"uchi automaton.

\begin{restatable}{theorem}{theoremcobuchisrtoma}\label{theorem:coBuchiHDisSR}
    There is a polynomial-time algorithm that converts stochastically resolvable coB\"uchi automata with $n$ states into language-equivalent memoryless-adversarially resolvable coB\"uchi automata with at most $n$ states. 
\end{restatable}
We note that the above result does not hold for SD coB\"uchi automata, since SD coB\"uchi are exponentially more succinct than HD coB\"uchi automata~\cite[Theorem 14]{AK23}. 
However, we show the succinctness of SR B\"uchi automata against HD B\"uchi automata and MA coB\"uchi automata against deterministic coB\"uchi~automata.
\begin{theorem}\label{theorem:succ}
\begin{enumerate}
    \item  There is a class of languages $L_n$ such that, there are memoryless-stochastically resolvable B\"uchi automata recognising $L_n$ with $\Oc(n)$ states, and any HD B\"uchi automaton recognising $L_n$ requires at least $2^n$ states. 
    \item  There is a class of languages $L_n'$, such that there are memoryless-adversarially resolvable coB\"uchi automata recognising $L_n'$ with $\Oc(n)$ states and any deterministic coB\"uchi automaton recognising $L_n'$ requires at least $\Omega(2^n/2n+1)$ states. 
\end{enumerate}
\end{theorem}

In the next subsections, we will prove \cref{theorem:comp,theorem:coBuchiHDisSR,theorem:succ}. We organise their proofs based on the acceptance~conditions.
\subsection{Safety automata}\label{subsec:sac-safe}
We showed in \cref{lemma:SR-implies-SD} that stochastically resolvable automata are semantically deterministic. The next result shows that every safety automaton $\Sc$ is semantically deterministic if and only if $\Sc$ is determinisable-by-pruning, that is, $\Sc$ contains a language-equivalent deterministic subautomaton, whose proof is presented in the appendix. 
\begin{restatable}[Folklore]{lemma}{SDsafetyisDBP}\label{lemma:sd-safety-is-dbp}
    Every semantically deterministic safety automaton is determinisable-by-pruning.
\end{restatable}

Observe that any determinisable-by-pruning automaton is trivially memoryless-adversarially resolvable, where Eve's strategy is to take transitions in a fixed, language-equivalent, deterministic subautomaton. Thus, \cref{lemma:sd-safety-is-dbp} implies the following result.

\begin{lemma}\label{lemma:comp-safety}
    For safety automata, the notions of pre-semantic determinism, stochastic resolvability, memoryless adversarial resolvability, memoryless stochastic resolvability, and history determinism are equivalent.
\end{lemma}

\subsection{Reachability and weak automata}\label{subsec:sac-rw}
We now compare our notions of nondeterminism on automata with reachability and weak acceptance conditions. To start with, recall that in \cref{example:reachSRbutnotHD}, we showed a MR reachability automaton that is not HD. This implies the following result.
\begin{lemma}\label{lemma:reachability-MR-not-HD}
    There is a memoryless stochastically resolvable reachability automaton that is not HD.
\end{lemma}

We next show that SD weak automata are MR. Since SR automata are pre-SD (\cref{lemma:SR-implies-SD}), we obtain that the notions of pre-semantic determinism and (memoryless) stochastic resolvability are equivalent on weak automata.  

\begin{restatable}{lemma}{lemmasdweakismr}\label{lemma:sdweak-is-mr}
    Every semantically deterministic weak automaton is memoryless-stochastically resolvable. 
\end{restatable}
\begin{proof}[Proof sketch]
Let $\Ac$ be a semantically deterministic weak automaton. We show that the resolver that selects transitions uniformly at random constructs, on any word in $\Lc(\Ac)$, a run that is almost-surely accepting. Let $w$ be a word in $\Lc(\Ac)$, and $\rho$ a run of $\Ac$ on $w$ where transitions are chosen uniformly at random. Then there is a finite prefix $u$ of $w$ and a state $p$ of $\Ac$, such that there is a run from $q_0$ to $p$ on $u$ and a run from $p$ on $u^{-1}w$ that only visits accepting states. Let $K=n2^n$, where $n$ is the number of states of $\Ac$. We show that there is a positive probability $\epsilon>0$, such that on any infix of $u^{-1}w$ that has length at least $K$, the segment of the run $\rho$ on that infix contains an accepting transition with probability at least $\epsilon$. Using this, we argue that $\rho$ contains infinitely many accepting transitions with probability $1$, and hence is almost-surely~accepting.
\end{proof}

Thus, stochastically resolvable weak automata are also memoryless-stochastically resolvable, and strictly encompass HD automata. 
History-deterministic weak automata are determinisable-by-pruning~\cite{BKS17}, and since MA automata are HD, and deterministic automata are trivially MA, we obtain the following result.
\begin{lemma}\label{lemma:reachability-hd=ma}
    The notions of determinisable-by-pruning, history-determinism, and memoryless adversarial resolvability coincide on reachability automata.
\end{lemma}

\subsection{CoB\"uchi automata}\label{subsec:sac-cobuchi}
We continue our comparison of the notions of nondeterminism, and focus on coB\"uchi automata. We will show that no two notions among SD, SR, HD, MR, and MA are equivalent for coB\"uchi automata. We will then give a polynomial-time algorithm that converts stochastically resolvable coB\"uchi automata with $n$ states into language-equivalent memoryless adversarially resolvable coB\"uchi automata with at most $n$ states. This shows that SR and MR coB\"uchi automata are no more succinct than MA (and also HD) coB\"uchi automata. 
However, this is not the case for B\"uchi automata (\cref{lemma:succinctBuchi}).

We start by showing an SD coB\"uchi automaton that is not stochastically resolvable, as shown in \cref{fig:cobuchisdbutnotsr}. 
\begin{figure}[h]
    \centering
    \begin{minipage}[b]{0.45\linewidth}
        \centering
            \begin{tikzpicture}
        \tikzset{every state/.style = {inner sep=-3pt,minimum size =20}}

    \node[state] (q0) at (0,0) {$q_0$};
    \node[state] (qa) at (1.5,0) {$q_a$};
    \node[state] (qb) at (-1.5,0) {$q_b$};
    \path[-stealth]
    (0,-0.75) edge (q0)

    (q0) edge node [above,yshift=-1mm] {$a,b$} (qa)
    (q0) edge node [above,yshift=-1mm] {$a,b$} (qb)
    (qa) edge [red, dashed, bend right = 45] node [above] {$b$} (q0) 
    (qa) edge [bend left = 45] node [below] {$a$} (q0)
    (qb) edge [red, dashed, bend left = 45] node [above] {$a$} (q0)
    (qb) edge [bend right = 45] node [below] {$b$} (q0)
    ;
    \end{tikzpicture}
    \caption{A coB\"uchi automaton that is SD but not SR. Rejecting transitions are represented by dashed arrows.}\label{fig:cobuchisdbutnotsr}
    \end{minipage}
    \hfill
    \begin{minipage}[b]{0.5\linewidth}
                \begin{tikzpicture}
        \tikzset{every state/.style = {inner sep=-3pt,minimum size =20}}

    \node[state] (q0) at (0,0) {$q_0$};
    \node[state] (q1) at (0,1.5) {$q_1$};
    \node[state] (d1) at (1.5,0) {$d_1$};
    \node[state] (d2) at (1.5,1.5) {$d_2$};
    \node[state] (d3) at (3,1.5) {$d_3$};
    \path[-stealth]
    (-0.5,-0.5) edge (q0)
    (q0) edge [bend left = 15] node [left] {$x$} (q1)
    (q1) edge [red, dashed, bend left = 15] node [right] {$a$} (q0)
    (q0) edge [in=-30,out=-60,loop] node [below, xshift=1mm] {$x$} (q0)
    (q0) edge [red, dashed, loop below] node [left] {$b$} (q0)
    (q1) edge [loop above] node [above] {$x$} (q1)
    (q1) edge node [above] {$b$} (d2)
    (d2) edge [red, dashed, loop above] node [left] {$b$} (d2)
    (d2) edge [red, dashed, bend left = 15] node [right] {$a$} (d1)
    (d1) edge [red, dashed, bend left = 15] node [left] {$b$} (d2)
    (q0) edge node [above] {$a$} (d1)
    (d1) edge [loop right] node [right] {$x,a$} (d1)
    (d2) edge [bend left = 15] node [above] {$x$} (d3)
    (d3) edge [bend left = 15] node [below] {$b$} (d2)
    (d3) edge [loop above] node [right] {$x$} (d3)
    (d3) edge [red, dashed] node [below] {$a$} (d1)
;
    \end{tikzpicture}
          \caption{A HD coB\"uchi automaton that is not MR. Rejecting transitions are represented by dashed arrows.}\label{fig:HDcoBuchinotMR}
    \end{minipage}
\end{figure}

\begin{lemma}
    There is a semantically deterministic coB\"uchi automaton that is not stochastically resolvable.
\end{lemma}
\begin{proof}
    Consider the coB\"uchi automaton $\Cc$ shown in \cref{fig:cobuchisdbutnotsr} that accepts all infinite words over $\{a,b\}$. The automaton $\Cc$ has nondeterminism on the letters $a$ and $b$ in the initial state.    
    Consider the following random strategy of Adam in the SR game on~$\Cc$, where in each round, he picks the letter $a$ or $b$ with equal probability. Then whenever Eve's token is at $q_{\alpha}$ for $\alpha\in\{a,b\}$, Adam's letter is $\beta \in \{a,b\}$ with $\beta\neq \alpha$ with probability~$\frac{1}{2}$. Thus, in round $2i$ of the SR game for each $i$ in $\mathbb{N}$, the run on Eve's token takes a rejecting transition with probability~$\frac{1}{2}$. Therefore, due to the second Borel-Cantelli lemma (\cref{lemma:secondborellcantelli}), the run on Eve's token contains infinitely many rejecting transitions with probability~1, and Eve loses almost-surely. Thus, Eve has no strategy to win the SR game on $\Cc$ almost-surely, proving $\Cc$ is not SR (\cref{lemma:random-is-pure}). 
\end{proof}

In \cref{subsec:sac-rw}, we showed an MR reachability automaton that is not HD (\cref{lemma:reachability-MR-not-HD}). Since reachability automata are a subclass of coB\"uchi automata, there is a coB\"uchi automaton that is not HD. We next show a HD coB\"uchi automaton that is not MR, and hence, also not MA. 
\begin{restatable}{lemma}{lemmaHDcoBuchinotMR}\label{lemma:HDcoBuchinotMR}
      There is a history-deterministic coB\"uchi automaton that is not memoryless stochastically resolvable.
\end{restatable}
\begin{proof}[Proof sketch]
    Consider the coB\"uchi automaton $\Cc$ shown in \cref{fig:HDcoBuchinotMR}. The automaton $\Cc$ has nondeterminism on the letter~$x$ in the initial state $q_0$. Informally, Eve, from the state $q_0$ in the HD game or the SR game, needs to `guess' whether the next sequence of letters till an $a$ or $b$ is seen form a word in $x^* a$ or in $x^+ b$. The automaton $\Cc$
    recognises the language $$L=(x+a+b)^{*} ((x)^{\omega} + (x^* a)^{\omega} + (x^+ b)^{\omega}).$$ 
    \paragraph*{$\Cc$ is HD}  If Eve's token in the HD game reaches the state $d_1,d_2$, or $d_3$, then she wins the HD game from here onwards since her transitions are deterministic. At the start of the HD game on $q_0$, or whenever she is at $q_0$ after reading $a$ or $b$ in the previous round she decides between staying at $q_0$ till an $a$ or $b$ is seen, or moving to $q_1$ on the first $x$ as follows.
    \begin{itemize}
        \item If the word read so far has a suffix in $x^{*}a$, then she stays in $q_0$ till the next $a$ or~$b$.
        \item If the word read so far has a suffix in $x^{+}b$, then she takes the transition to $q_1$ on~$x$.
        \item Otherwise, she stays in $q_0$ till the next $a$ or $b$.
    \end{itemize}
    Due to the language of $\Cc$ being the set of words which have a suffix in $x^{\omega},(x^{*}a)^{\omega},$ or $(x^+b)^{\omega}$, the above strategy guarantees that Eve's token moves on any word in $L$ in HD game to one of $d_1$ or $d_2$, from where she wins the HD game.
    \paragraph*{$\Cc$ is not MR} Note that the automaton $\Cc$ does not accept the same language if any of its transitions are deleted. Consider a memoryless resolver for $\Cc$ that takes the self-loop on $x$ on $q_0$ with probability $(1-p)$ and and the transition to $q_1$ on $x$ with probability $p$, for some $p$ satisfying $0<p<1$. We show that the runs that the resolver constructs on the word $x^2 a x^3 a x^4 a \dots$, do not visit the states $d_1,d_2,$ or $d_3$ and are rejecting with positive probability. This shows that $\Cc$ has no almost-sure memoryless resolver for $\Cc$, and is not MR.
\end{proof}

We have shown so far that each of the five classes in \cref{fig:venndiagrammmmmm} are different for coB\"uchi automata. We next show that every SR coB\"uchi automaton can be converted into a language-equivalent MA coB\"uchi automaton without any additional states.

\theoremcobuchisrtoma*
We start by fixing a SR coB\"uchi automaton $\Ac$ throughout the proof of \cref{theorem:coBuchiHDisSR}. 
We first relabel the priorities on $\Ac$ to obtain $\Cc$ as follows. Consider the graph consisting of all states of $\Ac$ and 0 priority transitions of $\Ac$. For any 0 priority transition of $\Ac$ that is not in any strongly connected component (SCC) in this graph, we change that transition to have priority~1 in $\Cc$. This relabelling of priority does not change the acceptance of any run (\cref{prop:priority-reduction}, \cref{app:succandcompCB}), and thus,~$\Cc$ is SR and language-equivalent to~$\Ac$.  

We start by introducing notions to describe a proof sketch of \cref{theorem:coBuchiHDisSR}. 
\paragraph*{Safe-approximation}
For the automaton $\Cc$, define the safe-approximation of $\Cc$, denoted $\Cc_{\safety}$ as the safety automaton constructed as follows. The automaton $\Cc_{\safety}$ has the same states as $\Cc$ and an additional rejecting sink state. The transitions of priority $0$ in $\Cc$ are preserved as the safe transitions of $\Cc_{\safety}$, and transitions of priority~$1$ in $\Cc$ are redirected to the rejecting sink state and have priority~1. 

\paragraph*{Weak-coreachability}
We call two states  $p$ and $q$ in $\Cc$ as coreachable, denoted by $p,q\in \CR(\Cc)$, if there is a finite word $u$ on which there are runs from the initial state of $\Cc$ to $p$ and $q$. We denote the transitive closure of this relation as \emph{weak-coreachability}, which we denote by $\WCR(\Cc)$. Note that weak-coreachability is an equivalence relation.

\paragraph*{SR self-coverage} 
For two  parity automata $\Bc$ and $\Bc'$, we say that $\Bc$ SR-covers $\Bc'$, denoted by $\Bc \succ_{SR} \Bc'$, if Eve has an almost-sure winning strategy in the modified SR game as follows. Eve, similar to the SR game on $\Bc$, constructs a run in $\Bc$, but 
Eve wins a play of the game if, in that play, Eve's constructed run in $\Bc$ is accepting whenever Adam's word is in $\Lc(\Bc')$.
We say that a coB\"uchi automaton $\Bc$ has \emph{SR self-coverage} if for every state $q$ there is another state $p$ that is coreachable to $q$ in $\Cc$, such that $(\Bc_{\safety},p)$ SR-covers~$(\Bc_{\safety},q)$.  

The crux of \cref{theorem:coBuchiHDisSR} is in proving the following result.
\begin{restatable}{lemma}{lemmaSRhassafeSRcoverage}\label{lemma:coBuchiSRhassafeSRcoverage}
The coB\"uchi automaton $\Cc$  has SR self-coverage.
\end{restatable}
\begin{proof}[Proof sketch] Fix an almost-sure resolver $\Mc$ for Eve in $\Cc$. Let~$\Pc$ be the probabilistic automaton that is the resolver-product of $\Mc$ and $\Cc$. We define $\Pc_{\safety}$ as a safety probabilistic automaton that is the safe-approximation of $\Pc$, similar to how we defined $\Cc_{\safety}$. Suppose, towards a contradiction, that there is a state $q$ in $\Cc$, such that for every state $p$ coreachable to $q$ in $\Cc$, $(\Cc_{\safety},p)$ does not SR-cover $(\Cc_{\safety},q)$. In particular, for every state $(p,m)$ in $\Pc$, where $p$ is coreachable to $q$ in $\Cc$, we have that $\Lc(\Pc_{\safety},(p,m)) \subsetneq \Lc(\Cc_{\safety},q)$. We use this to show that there is a finite word $\alpha_{(p,m)}$, on which there is a run consisting of only priority 0 transitions from $q$ to $q$ in $\Cc$, while a run $\rho$ of $(\Pc,(p,m))$ on $\alpha_{(p,m)}$ contains a priority $1$ transition with probability at least $\epsilon$ for some $\epsilon>0$.    

Adam then has a strategy in the SR game on $\Cc$ against Eve's strategy $\Mc$ as follows. Adam starts by giving a finite word~$u_q$, such that there is a run of $\Cc$ from its  initial state to~$q$. Then Adam, from this point and at each \emph{reset}, selects a state~$(p,m)$ of $\Pc$ uniformly at random, such that $p$ is coreachable to $q$ in $\Cc$ and $m$ is a memory-state in $\Mc$. He then plays the letters of the word $\alpha_{(p,m)}$ in sequence, after which he \emph{resets} to select another such state and play similarly. This results in Eve constructing a run in the SR game on $\Cc$ that contains infinitely many priority $1$ transitions almost-surely, while Adam's word is in $\Lc(\Cc)$. It follows that $\Mc$ is not an almost-sure resolver for Eve, which is a contradiction.
\end{proof}

 SR-covers is a transitive relation, i.e., if $\Ac_1,\Ac_2,\Ac_3$ are nondeterministic parity automata, such that $\Ac_1 \succ_{SR} \Ac_2$ and $\Ac_2 \succ_{SR} \Ac_3$, then $\Ac_1 \succ_{SR} \Ac_3$. The following result then follows from the definition of SR self-coverage and the fact that $\Cc$ has finitely many states.

\begin{restatable}{lemma}{lemmacobuchisometingdbp}\label{lemma:cobuchi-srselfcoverage-implies-somtingsdbp}
    For every state $q$ in $\Cc$, there is another state $p$ weakly coreachable to $q$ in $\Cc$, such that $(\Cc_{\safety},p)$ SR-covers $(\Cc_{\safety},q)$ and $(\Cc_{\safety},p)$ SR-covers $(\Cc_{\safety},p)$.  
\end{restatable}

Note that if $(\Cc_{\safety},p)$ SR-covers $(\Cc_{\safety},p)$ then $(\Cc_{\safety},p)$ is SR. Since SR automata are semantically deterministic (\cref{lemma:SR-implies-SD}) and SD safety automata are determinisable-by-pruning (\cref{lemma:sd-safety-is-dbp}), we call  such states $p$ as \emph{safe-deterministic}. 

We will build a memoryless adversarially resolvable automaton $\Hc$, whose states are the safe-deterministic states in $\Cc$. This construction is similar to the one used by Kuperberg and Skrzypczak, in 2015, for giving a polynomial time procedure to recognise HD coB\"uchi automata~\cite[Section E.7 in the full version]{KS15}.  We fix a uniform determinisation of transitions from each safe-deterministic state in $\Cc_{\safety}$ and we add these transitions in $\Hc$ with priority 0. If there are no outgoing transitions in $\Hc$ from the state $p$ on letter $a$ so far, then we add outgoing transitions from $p$ on $a$ as follows. Let $q$ be a state in $\Cc$ such that there is a transition from $p$ to $q$ on $a$ in $\Hc$. For each state $r$ that is weakly coreachable to $q$ in $\Cc$ and that is safe-deterministic, we add a transition from $q$ to $r$ in $\Hc$ with priority $1$. This concludes our construction of $\Hc$.

We show that the strategy of Eve that chooses transitions from $\Hc$ uniformly at random is an almost-surely winning strategy for Eve in the HD game, and thus, $\Hc$ is MA (\cref{lemma:cobuchi-h-is-ma},\cref{app:succandcompCB}). Both this fact and the language equivalence of $\Hc$ to $\Cc$ primarily relies on \cref{lemma:cobuchi-srselfcoverage-implies-somtingsdbp}. Since a safety automaton is DBP if and only if that automaton is HD and every HD safety automaton can be determinised in polynomial-time~\cite{BL23quantitative},  the safe-deterministic states of $\Cc$ can be identified, and outgoing safe transitions from these states can be determinised in polynomial-time. Thus, the construction of $\Hc$ takes polynomial-time overall. This completes our proof sketch for \cref{theorem:coBuchiHDisSR}.



We next show that MA coB\"uchi automata are exponentially more succinct when compared to deterministic coB\"uchi automata, thus proving the coBu\"chi part of \cref{theorem:succ}.
\begin{corollary}\label{lemma:succinctcoBuchi}
    There is a family $L_2,L_3,L_4,\dots$ of languages such that for every $n\geq 2$, there is a memoryless adversarially resolvable automaton recognising $L_n$ that has $2n+1$ states and any deterministic coB\"uchi automaton recognising $L_n$ needs at least $\Omega(2^n/2n+1)$ states. 
\end{corollary}
The language family constructed in the work of Kuperberg and Skrzypczak~\cite[Theorem~1]{KS15} are accepted by $2n+1$-state HD coB\"uchi automata and are not accepted by any $\Omega(2^n/2n+1)$-state deterministic coB\"uchi automata. From \cref{theorem:coBuchiHDisSR}, since any SR---and therefore any HD---automaton has a language-equivalent MA automaton with the same number of states, \cref{lemma:succinctcoBuchi} follows. 
\subsection{B\"uchi automata}\label{subsec:sac-buchi}
Similar to coB\"uchi automata, we show that for B\"uchi automata, no two notions amongst the notions of semantic determinism, stochastic resolvability, history determinism, memoryless stochastic resolvability, and memoryless adversarial resolvability coincide. We then later show that memoryless-stochastically resolvable automata are exponentially more succinct than HD B\"uchi automata: recall that this is not the case for coB\"uchi automata (\cref{theorem:coBuchiHDisSR}).

We start by giving a SD B\"uchi automaton that is not stochastically resolvable. Consider the B\"uchi automaton as shown in \cref{fig:BuchiSDbutnotSR} below, which we show is SD but not SR.
\begin{figure}[ht]
\centering
        \begin{tikzpicture}[auto]
        \tikzset{every state/.style = {inner sep=-3pt,minimum size =15}}
    \node[state] (s1)  at (0,0) {$q_0$};
    \node[state]  (s2)  at (2,0) {};
    \node[state] (s3) at (1,0.8) {$q_a$};
    \node[state] (s4) at (1,-0.8) {$q_b$};

    \node[state] (f1)  at (-1.2,0) {};
    \node[state]  (f2)  at (-2.2,0.8) {};
    \node[state] (f3) at (-2.2,-0.8) {};

    \path[->]
        (f3) edge node [left] {$x$}  (f2)
        (f1) edge node [yshift=1mm] {$y$} (f3)
        (f2) edge node [xshift=-2mm] {$a,b$} (f1)
        (s2) edge [double,bend left = 8] node {$z$}  (s1)
          (s1) edge  node [below,xshift=2mm,yshift=2mm] {$x$} (s3)
         (s1) edge  node [above,xshift=2mm,yshift=-2mm] {$x$} (s4)
         
         (s3) edge  node [above] {$a$} (s2)
         (s4) edge  node [below] {$b$} (s2)
         
         (s2) edge [double,bend left = 8] node [below] {$z$} (s1)
         (s2) edge [bend right = 8] node [above] {$y$} (s1)
         (s3) edge  node [above] {$b$} (f1)
         (s4) edge node [below] {$a$} (f1)
         (f1) edge node [above,xshift=1mm,yshift=-0.5mm] {$z$} (s1)
;
    \path[->,every node/.style={sloped,anchor=south}]
        ;
    \end{tikzpicture}
    \caption{A semantically deterministic B\"uchi automaton that is not stochastically resolvable. The accepting transitions are double-arrowed, and the initial state is $q_0$.}
    \label{fig:BuchiSDbutnotSR}
\end{figure}
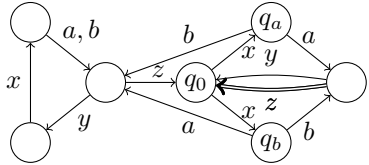 
\begin{restatable}{lemma}{buchisdnotsr}\label{lemma:buchisdnotsr}
    There is a semantically deterministic B\"uchi automaton that is not stochastically resolvable.
\end{restatable}
\begin{proof}[Proof sketch]
    Consider the B\"uchi automaton $\Bc$ shown in \cref{fig:BuchiSDbutnotSR}. This automaton $\Bc$ has nondeterminism on the initial state $q_0$, and it recognises the language $$((x \cdot (a+b)\cdot y)^{*}(x\cdot (a+b)\cdot z))^{\omega}.$$ 
    It is easy to verify that $\Bc$ is SD. We will describe a strategy for Adam in the SR game on $\Bc$ using which he wins almost surely. This would imply, due to \cref{lemma:random-is-pure}, that $\Bc$ is not SR. Note that when Eve's token is at $q_0$ in the SR game, then Eve needs to guess whether the next letter after $x$ is going to be $a$ or $b$. If she guesses incorrectly, then her token moves to the left states---states $l_1,l_2,$ and $l_3$, where she stays until a $z$ is seen. Adam's strategy in the SR game is as follows. Let $Y$ be the regular expression $xay+xby$ and $Z$ be the regular expression $xaz+xbz$. Note that both $Y$ and $Z$ consist of two words. Adam picks a word from the set $YZY^2ZY^3ZY^4Z \dots$ in the SR game, where from each occurrence of $Y$ or $Z$, he picks one of the two words in the regular expression with half probability. We show that the probability that Eve's token takes an accepting transition on reading a word chosen randomly from $Y^{n}Z$ is $\frac{1}{2^{n+1}}$. It then follows from the Borel-Cantelli lemma (\cref{lemma:borellcantelli}, \cref{app:defs}) that the probability that Eve's token takes infinitely many accepting transitions in the SR game is 0, as desired. 
\end{proof}

We showed in \cref{lemma:reachability-MR-not-HD} that MR reachability automaton are not HD. This also shows that there are MR B\"uchi automata which are not HD. We next show the other side by showing that there are HD B\"uchi automata that are not MR.
\begin{restatable}{lemma}{HDBuchinotMR}\label{lemma:HDBuchinotMR}
    There is a history-deterministic B\"uchi automaton that is not memoryless stochastically resolvable.
\end{restatable}
\begin{proof}[Proof sketch]
    Consider the automaton in \cref{fig:HDBuchinotMR}. We will only define the language here, but remark that the proof that this language is not MR is similar to that the proof of \cref{lemma:HDcoBuchinotMR}. 
    Let  $\Sigma_\diamond = \{a,b,c,\diamond\}$ and $\Sigma = \{a,b,c\}$. Let $L_1$ and $L_2$ be languages of finite words over the alphabet $\Sigma_\diamond$ where $L_1 =   {\Sigma_\diamond}^*  c^+\diamond $ and $L_2 = {\Sigma_\diamond}^* a \Sigma^* b^+\diamond$. 
    The automaton $\Ac$ accepts the language $\left[(L_1+L_2)^*(L_1L_1+L_2L_2)\right]^\omega$. Equivalently, it accepts words in $(L_1+L_2)^\omega$ that are not in  $(L_1+L_2)^*(L_1L_2)^\omega$.  
\end{proof}
\begin{figure}[ht]
\centering
        \begin{tikzpicture}
        \tikzset{every state/.style = {inner sep=-3pt,minimum size =15}}

    \node[state] (q0) at (-1,1.5) {};
    \node[state] (q1) at (-2.5,1.5) {};
    \node[state] (q2) at (2,1.5) {};
    \node[state] (q3) at (0.5,1.5) {};
    \path[->] (-0.7,1) edge (q0);
    
    \node[state, fill=blue!40, blue!40] (r0) at (2,0.2) {$q_B$};
    \node[state, blue] (r1) at (2,-1.3) {};
    \node[state, blue] (r2) at (3.5,-1.3) {};
    \node[state, blue] (r3) at (3.5,0.2) {};    
    \node[state, blue] (r4) at (5.25,-0.55) {};    

    \node[state, fill=red!40, red!40] (l0) at (-3,0.2) {$q_R$};
    \node[state, red] (l1) at (-3,-1.3) {};
    \node[state, red] (l2) at (-1.5,-1.3) {};
    \node[state, red] (l3) at (-1.5,0.2) {};    
    \node[state, red] (l4) at (0.25,-0.55) {};    

    \node[state, fill=red!40, red!40] (l5) at (-2.25,-2.1) {$q_R$};    
    \node[state, fill=blue!40, blue!40] (l6) at (0.25,-1.8) {$q_B$};  
    
    \node[state, fill=red!40, red!40] (r5) at (2.75,-2.1) {$q_R$};    
    \node[state, fill=blue!40, blue!40] (r6) at (5.25,-1.8) {$q_B$};    
    \node (l52) at (-2.25,-2.1) {$q_R$};    
    \node (l62) at (0.25,-1.8) {$q_B$};  
    
    \node (r51) at (2.75,-2.1) {$q_R$};    
    \node (r61) at (5.25,-1.8) {$q_B$};    
    \node (l01) at (-3,0.2) {$q_R$};
    \node (r01) at (2,0.2) {$q_B$};
    
    \path[-stealth]
    (q0) edge [loop above] node [right] {$a,b,\diamond$} (q0)
    (q0) edge [bend left = 8] node [above] {$a$} (q3)
    (q3) edge [bend left = 8] node [below] {$\diamond$} (q0)
    (q0) edge [bend left = 8] node [below] {$c$} (q1)
    (q1) edge [bend left = 8] node [above] {$b,a$} (q0)
    (q3) edge [bend right = 8] node [below] {$b$} (q2)
    (q2) edge [bend right = 8] node [above] {$c,a$} (q3)
    (q3) edge [loop above] node [right] {$a,c$} (q3)
    (q1) edge [loop above] node [left] {$c$} (q1)
    (q2) edge [loop right] node [right] {$b$} (q2)

    (l0) edge [loop above] node [above] {$b,\diamond$} (l0)
    (l0) edge [bend left = 8] node [above] {$a$} (l3)
    (l3) edge [bend left = 8] node [below] {$\diamond$} (l0)
    (l1) edge  node [below, left] {$a$} (l3)
    (l0) edge  node [left] {$c$} (l1)
    (l3) edge [bend right = 8] node [left] {$c$} (l2)
    (l2) edge [bend right = 8] node [right] {$a$} (l3)
    (l3) edge [bend right = 8] node [below=3pt,left] {$b$} (l4)
    (l4) edge [bend right = 8] node [above=3pt,right] {$a$} (l3)
    (l4) edge [bend right = 8] node [above=0.5pt] {$c$} (l2)
    (l2) edge [bend right = 8] node [below=0.5pt] {$b$} (l4)
    
    (l3) edge [in=30,out=60,loop] node [right] {$a$} (l3)
    (l1) edge [in=240,out=270,loop] node [below] {$c$} (l1)
    (l2) edge [in=-30,out=-60,loop]  node [right] {$c$} (l2)
    (l4) edge [loop above] node [above] {$b$} (l4)

    (r0) edge [in=165,out=195,loop] node [left] {$b,\diamond$} (r0)
    (r0) edge [bend left = 8] node [above] {$a$} (r3)
    (r3) edge [bend left = 8] node [below] {$\diamond$} (r0)
    (r1) edge  node [below, left] {$a$} (r3)
    (r0) edge  node [left] {$c$} (r1)
    (r3) edge [bend right = 8] node [left] {$c$} (r2)
    (r2) edge [bend right = 8] node [right] {$a$} (r3)
    (r3) edge [bend right = 8] node [below=3pt,left] {$b$} (r4)
    (r4) edge [bend right = 8] node [above=3pt, right] {$a$} (r3)
    (r2) edge [bend right = 8] node [below=0.5pt] {$b$} (r4)
    (r4) edge [bend right = 8] node [above=0.5pt] {$c$} (r2)

    (r3) edge [in=30,out=60,loop]  node [above] {$a$} (r3)
    (r1) edge [loop left] node [left] {$c$} (r1)
    (r2) edge [in=-30,out=-60,loop] node [right] {$c$} (r2)
    (r4) edge [loop above] node [above] {$b$} (r4)

    (q1) edge [in=50,out=-90] node [right] {$\diamond$} (l0)
    (q2) edge node [right] {$\diamond$} (r0)

    (l1) edge[ double] node [left] {$\diamond$} (l5)
    (l2) edge[ double] node [right] {$\diamond$} (l5)
    (l4) edge node [left] {$\diamond$} (l6)

    (r1) edge node [left] {$\diamond$} (r5)
    (r2) edge node [right] {$\diamond$} (r5)
    (r4) edge [double] node [left] {$\diamond$} (r6)
;
    \end{tikzpicture}
\caption{A HD B\"uchi automaton that is not MR. The accepting transitions are represented by double arrows. All red-filled states ($q_R$) are identified as the same state, and all blue-filled states ($q_B$) are identified as the same state.}\label{fig:HDBuchinotMR}
\end{figure}
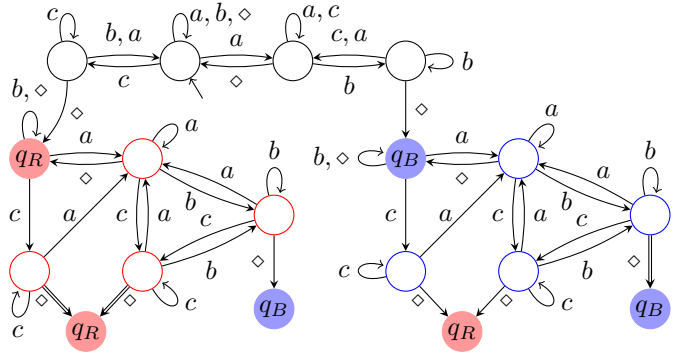
We now show that SR B\"uchi automata  are exponentially more succinct than HD B\"uchi automata. 
\begin{restatable}{lemma}{succinctBuchi}\label{lemma:succinctBuchi}
    There is a family $L_2,L_3,L_4,\dots$ of languages such that for every $n\geq 2$, there is a memoryless stochastically resolvable automaton recognising $L_n$  that has $3n+3$ states and any HD B\"uchi automaton recognising $L_n$ needs at least $2^n$ states. 
\end{restatable}
To prove this result, we will use the family of languages of Abu Radi and Kupferman to show the exponential succinctness of SD B\"uchi automata~\cite[Theorem 5]{AK23}. For each $n\geq 2$, consider the alphabet $\Sigma_n=\{1,2,\dots,n,\$,\#\}$, and let us denote the set $\{1,2,\dots,n\}$ by $[n]$. 
Define the language $L_n=\{\$w_0\#i_0\$w_1\#i_1\$w_2\#i_2\dots \mid \text{~there are}$ $\text{ infinitely many indices $j$ such that $i_j$ appears in $w_j \in [n]^{\omega}$}\}$. The automaton in \cref{fig:succinctSRBuchi}, which has $3n+3$ states (with a missing rejecting sink state), accepts this language. This automaton is MR, where the memoryless resolver chooses uniformly at random one of the outgoing transitions on $\$$ at state $q_0$ is an almost-sure resolver for it. For any word $w$ in $L_n$, there are infinitely many occurrences of $i$-good words for some $i \in [1,n]$, and therefore, by the second Borel-Cantelli Lemma (\cref{lemma:secondborellcantelli}), there are infinitely many positions at which the resolver chooses the transition to $s_i$ while reading a $\$$ right before an $i$-good word from $q_0$, and then visits an accepting transition. 
\begin{figure}[ht]
\centering
  
        \begin{tikzpicture}
        \tikzset{every state/.style = {inner sep=-3pt,minimum size =15}}

    \node[state,initial,initial text=] (q0)  at (-2.5,1.5) {$q_0$};
    \node[state] (r)  at (-2.5,3) {$r$};    

    \node[state] (s1)  at (-0.7,-0.7) {$s_n$};
    \node[state]  (s2)  at (-0.7,1.5) {$s_2$};
    \node[state] (s3) at (-0.7,3) {$s_1$};

    \node[state] (m1)  at (0.4,-0.7) {$m_n$};
    \node[state]  (m2)  at (0.4,1.5) {$m_2$};
    \node[state] (m3) at (0.4,3) {$m_1$};

    \node[state] (f1)  at (1.4,-0.7) {$f_n$};
    \node[state]  (f2)  at (1.4,1.5) {$f_2$};
    \node[state] (f3) at (1.4,3) {$f_1$};

    \node (d)  at (0,0.9) {$\vdots$};

    \node[state] (fin) at (3,1.5) {$q_0$};

    \path[->]
    (s3) edge node [above] {$\#$}  (r)
    (r) edge node [left] {$[n]$}  (q0);

        \path[->]
    (s1) edge[loop above] node [above] {$[n]\setminus\{i\}$}  (s1)
        (s1) edge node [above] {$i$}  (m1)
        (m1) edge [loop above] node [above] {$[n]$}  (m1)
        (m1) edge  node [above] {$\#$} (f1)
        ;
            \path[->]
    (s2) edge[loop above] node [above] {$[n]\setminus\{i\}$}  (s2)
        (s2) edge node [above] {$i$}  (m2)
        (m2) edge [loop above] node [above] {$[n]$}  (m2)
        (m2) edge  node [above] {$\#$} (f2)
        ;
            \path[->]
    (s3) edge[loop above] node [above] {$[n]\setminus\{i\}$}  (s3)
        (s3) edge node [above] {$i$}  (m3)
        (m3) edge [loop above] node [above] {$[n]$}  (m3)
        (m3) edge  node [above] {$\#$} (f3)
        ;
        \path[->]
        (q0) edge node [below] {$\$$} (s1)
        (q0) edge node [below] {$\$$} (s2)
        (q0) edge node [above] {$\$$} (s3)



        (f1) edge[bend right =15, double] node [left] {$n$} (fin)
        (f1) edge[bend right = 30] node [right] {$[n]\setminus \{n\}$} (fin)
        (f2) edge[double] node [above] {$2$} (fin)
        (f2) edge[bend right =30] node [below] {$[n]\setminus \{2\}$} (fin)
        (f3) edge[double, bend left = 8] node [below, xshift=-2mm,yshift=+2mm] {$1$} (fin)
        (f3) edge[bend left = 30] node [right] {$[n]\setminus \{1\}$} (fin);
    \end{tikzpicture}
\caption{An MR B\"uchi automaton that is exponentially more succinct than any HD automaton accepting the same language. Both $q_0$s are the identified as the same state. Additionally, transitions are added from all states $s_i$ on $\#$ to state $r$. Other missing transitions go to a sink state (not pictured).}\label{fig:succinctSRBuchi}
\end{figure}
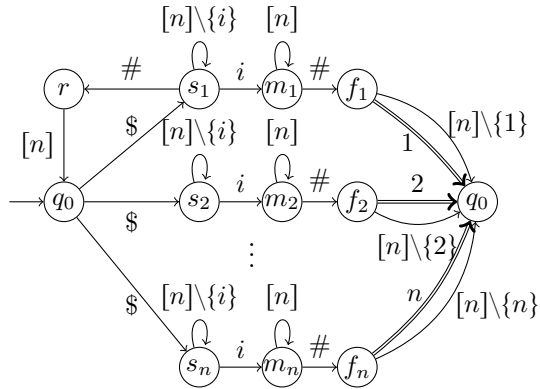 
\section{Expressivity}\label{sec:pih}
In \cref{sec:succandcomp}, we compared the novel classes of nondeterministic automata we defined with each other and the existing notions of semantic determinism and history determinism, focusing on the questions of succinctness and how they coincide. In this section, we compare these notions in terms of expressivity. We show that, similar to history-determinism, stochastically resolvable $[i,j]$-parity automata are as expressive as deterministic $[i,j]$-parity  automata.

\begin{restatable}{theorem}{pih}\label{theorem:pih}
Stochastically resolvable $[i,j]$-parity automata recognise the same  languages as deterministic $[i,j]$-parity automata.
\end{restatable}
Since deterministic automata are trivially SR, one direction is clear. For the other direction, we will show that any $\omega$-regular language that is not recognised by any deterministic $[i,i+d]$-parity automaton  cannot be recognised by any SR $[i,i+d]$-parity automaton. To show this, we will consider the language $L_{[i+1,i+d]}$ of the $[i+1,i+d+1]$-parity condition, which is the set of infinite words over the alphabet $[i+1,i+d+1]$ in which the highest number occurring infinitely often is even. We show that there is no SR $[i,i+d]$-parity automaton recognising language $L_{[i+1,i+d+1]}$. 

To prove that no SR $[i,i+d]$-parity automaton recognising language $L_{[i+1,i+d+1]}$, consider any SR automaton $\Ac$ that recognises the language $L_{[i+1,i+d+1]}$ and has an almost-sure resolver $\sigma$. We inductively construct words $u_0,u_1,\cdots,u_d$, such that from every state $q$, a run from $q$ on the word $u_k$ in automaton $\Ac$ constructed using $\sigma$ contains a transition with priority at least $(i+k+1)$ with positive probability. This part of the proof is nontrivial and requires careful analysis of probabilities. Once we have proved this result inductively, we obtain that $\Ac$ has at least as many priorities as in the interval $[i+1,i+d+1]$, and in particular, is not an $[i,i+d]$-parity automaton, as desired.

\section{Complexity of recognition}\label{sec:complexity}
We now turn our attention to the computational complexity for the problems of deciding if a given automaton is MR, SR, or MA, respectively. The exact complexity of deciding if a given automaton is HD is an open problem since 2006, with only recent results showing that the problem is $\NP$-hard~\cite{Pra24a} and in $\PSPACE$~\cite{LP25}. 

We discuss the key results of this section. Firstly, unlike for history-determinism where the problem of deciding if the automaton is HD has a  complexity gap for parity automata, we show that the problem for memoryless adversarial resolvability is $\NP$-complete.
\begin{restatable}{theorem}{theoremmanpcomplete}\label{theorem:manpcomplete}
    The problem of deciding if a given parity automaton is memoryless-adversarially resolvable is $\NP$-complete. 
\end{restatable}
Later, we turn our attention to problems of checking if an automaton is SR and the problem of checking if a resolver is an almost-sure resolver in the stochastically resolvable setting. 
\begin{restatable}{theorem}{complexity}\label{theorem:complexity}
\begin{enumerate}
    \item The problem of deciding if a given safety automaton is stochastically resolvable or memory\-less-stochastically  resolvable is in $\ptime$.
    \item The problem of deciding if a given reachability or weak automata is stochastically resolvable or memoryless stochastically resolvable is $\PSPACE$-complete.
    \item The problem of deciding, for a given B\"uchi or coB\"uchi automaton and a finite memory resolver for that automaton, if the resolver is an almost-sure resolver of that automaton is undecidable.
\end{enumerate}
\end{restatable}
\subsection{Memoryless adversarially resolvable automata}
We first prove \cref{theorem:manpcomplete} and show $\NP$-completeness for the problem of deciding if an automaton is MA. To prove the upper bound of $\NP$ (\cref{lemma:MAcheckNPEasy}), we show that an automaton is MA if and only if there is a specific kind of strategy in the so-called $2$-token game on that automaton. 
The $2$-token game, defined by Bagnol and Kuperberg, was introduced as a game that characterises history determinism for B\"uchi automata~\cite{BK18}. Lehtinen and Prakash have recently shown that these games also characterise history determinism for parity automata~\cite{LP25}. Using the characterisation of MA using $2$-token games, we provide an algorithm for verifying if such a specific kind of strategy is indeed a winning strategy. We then show the lower bound in \cref{lemma:MAcheckNPhard} by showing that the problem is $\NP$-hard, similar to the proof of $\NP$-hardness for deciding history-determinism~\cite{Pra24a}. 

\begin{lemma}\label{lemma:MAcheckNPEasy}
Checking if a given parity automaton is memoryless adversarially resolvable is in $\NP$.
\end{lemma}
Next, we describe a proof sketch for \cref{lemma:MAcheckNPEasy}, where we will use the following notions.

\paragraph*{$2$-token game} Informally, the 2-token game on an automaton is a game played between Adam and Eve with three tokens at the starting state of the automaton, one used by Eve and two used by Adam. The game proceeds in infinitely many rounds, where in each round, Adam selects a letter, Eve selects a transition on that letter along which she moves her token, and finally Adam moves each of his two tokens along transitions on the same letter. Therefore, in a play of this game, Eve builds a run and Adam two runs, all on the same word. A play is won by Adam if at least one of the runs on his tokens is accepting while Eve's run on her token is rejecting, and Eve wins otherwise. Given an automaton $\Ac$, we write $G2(\Ac)$ to represent the 2-token game on $\Ac$, and we use $G2(\Bc;\Ac)$ to describe a modified version of the 2-token game where Eve moves her token in the automaton $\Bc$, while Adam moves his two tokens in the automaton~$\Ac$.

 \paragraph*{Muller condition} Muller objectives on graphs are specified by a finite set of colour $C$ and a set of accepting subsets $\Fc\subseteq 2^C$, and each edge of the graph is labelled by a colour from $C$. An infinite path of this graph is accepting if the set of colours that appears infinitely often in this path is a member of~$\Fc$. Such subsets $\Fc$ are sometimes represented using \emph{Zielonka trees}~\cite{DJW97}  defined below.

\paragraph*{Zielonka tree $\Zc_{C,\Fc}$}
For a Muller objective defined by colours $C$ and accepting set $\Fc$, we construct its Zielonka tree, denoted $\Zc_{C,\Fc}$, as a labelled rooted tree. We will call the vertices of this tree as \emph{nodes}. Each node $\nu$ of the Zielonka tree is labelled by a nonempty subset of $C$. 
The root of $\Zc_{C,\Fc}$ is labelled by $C$. For a node $\nu_X$ labelled by $X$, its children are nodes $\nu_Y$ labelled by distinct maximal nonempty subsets $Y$ of $X$ such that $Y\in\Fc$ if and only if $X \notin \Fc$. If there are no such subsets $Y$, then $\nu_X$ has no children in $\Zc_{C,\Fc}$.

\paragraph*{Zielonka DAG} A Zielonka DAG~\cite{HD05} is a succinct representation of Zielonka tree, where nodes with the same labels are merged.

\noindent We show that we can characterise memoryless resolvability using the $2$-token game.
    \begin{lemma}\label{prop:rigid2Token}
        An automaton $\Ac$ is MA if and only if there is a subautomaton $\Bc$ of $\Ac$ such that Eve randomising between all available outgoing edges in $G2(\Bc;\Ac)$ is almost-surely winning for Eve.
    \end{lemma}
    \begin{proof}[Proof of \cref{prop:rigid2Token}]
        \emph{$\Ac$ is not HD.} Then $\Ac$ is clearly not MA.  Adam wins $G2(\Ac)$ and hence $G2(\Bc;\Ac)$, due to the recent result of Lehtinen and Prakash~\cite[2-token theorem]{LP25}.
        
        \emph{$\Ac$ is HD and MA.} If $\Ac$ is MA, then there is a strategy of Eve in the HD game on $\Ac$ where she only randomises between some fixed set of transitions $\Ac$. Let $\Bc$ be the sub-automaton consisting exactly of this set of transitions.  Since such a resolution always produces an accepting run almost surely on any word in $\Lc(\Ac)$, Eve also must win the 2-token game $G2(\Bc;\Ac)$ using this strategy. 
        
        \emph{$\Ac$ is HD and not MA.} Then for every Eve's memoryless strategy $\sigma$ in the HD game, Adam has a strategy $\tau_{\sigma}$ to ensure, with positive probability, that his word is accepting while Eve's run is rejecting. For every subautomaton $\Bc$ of Eve, in the $2$-token game, consider the strategy for Adam $\tau_{\Bc}$ in the $2$-token game, where he picks letters according to a strategy that ensures Adam's word is accepting while Eve's run on her token is rejecting with positive probability. He picks transitions on his token according to a winning strategy for Eve in the HD game on $\Ac$, which ensures that the runs on his tokens are accepting if his word is accepting. Thus, there is no subautomaton $\Bc$ of $\Ac$ such that Eve wins $G2(\Bc;\Ac)$ almost-surely by choosing transitions randomly in $\Bc$, as desired.  
    \end{proof}
    To check if an automaton is MA, we guess a subautomaton $\Bc$ of $\Ac$, construct the $2$-token game $G2(\Bc;\Ac)$, and verify if Eve playing randomly is an almost-sure winning strategy. Subsequently, we construct a game where Eve's vertices in are substituted with stochastic vertices in $G2(\Bc;\Ac)$, resulting in a Markov Decision Process (MDP): a single-player complete-observation stochastic game where Adam is the only player. 
    
    If $\Ac$ is MA, Adam can satisfy his objective in the 2-token game with probability $0$ in this resulting MDP. The winning condition of the 2-token game for Adam, which is that at least one of Adam's runs on his tokens is accepting while Eve's run on her token is rejecting, can be represented by a Muller objective~\cite[Page 70]{Pra25}. The Zielonka DAG of this Muller objective has size that is at most polynomial in the size of the automaton $\Ac$~\cite[Page 72]{Pra25}.

    In \cref{lemma:ZlkDAGMDP}, we show that it can be verified in polynomial time if an MDP  satisfies a Muller objective with positive probability or almost-surely where the objective is input as a Zielonka DAG. Therefore,
    we can verify if Eve can win almost-surely or Adam can win with positive probability in the MDP in polynomial time. This proves \cref{lemma:MAcheckNPEasy}. 

\begin{restatable}{theorem}{ZlkDAGMDP}\label{lemma:ZlkDAGMDP}\label{thm:ZlkDAGMDP}
        Given an MDP $\Mc$ where the Muller objective is represented as a Zielonka DAG $\Zc_{C,\Fc}$, deciding whether Adam can satisfy the Muller objective with positive probability (resp. almost-surely) in $\Mc$ can be computed in time $\Oc(|\Mc||\Zc_{C,\Fc}|)$.        
\end{restatable}

Chatterjee showed that for MDPs, deciding if there is a strategy to almost-surely (or positively) satisfy a Muller objective represented by a set of subset of colours $\Fc$ that is union closed or upward closed and succinctly represented as a ``basis condition'' is in $\P$~\cite[Section 4]{Cha07}. Since every upward-closed or union-closed condition can be represented as a simple Zielonka DAG where all its nodes other than the root are leaf nodes, \cref{thm:ZlkDAGMDP} is therefore an improvement. 

We now prove the lower bound for the problem of checking memoryless adversarial resolvability using the result of $\NP$-hardness of checking history determinism~\cite{Pra24a}, and observe, in the appendix, that the reduction in \cite{Pra24a} can also be modified to show the $\NP$-hardness of checking if an automaton is MA. 
\begin{restatable}{lemma}{MAcheckNPhard}\label{lemma:MAcheckNPhard}
    Checking if a given parity automaton is memoryless-adversarially resolvable is $\NP$-hard.
\end{restatable}

\subsection{Memoryless-stochastically resolvable automata and stochastically resolvable automata}

We now discuss the questions of deciding whether an automaton is MR or SR in this subsection. 

\begin{proposition}\label{prop:safetyPTIME}\label{prop:reachweakPSPACE}
    Deciding if a safety automaton is stochastically resolvable is in $\ptime$, and deciding if a reachability or weak automaton is stochastically resolvable is $\PSPACE$-complete. 
\end{proposition}
\begin{proof}
    From \cref{lemma:SR-implies-SD}, an automaton is SR if and only if it is SD-by-pruning. A safety automaton is SD if and only if it is determinisable by pruning if and only if it is HD \cite[Theorem 19]{BL21}, which is decidable in $\ptime$~\cite{BL23quantitative}. 

    Deciding if a weak automata is pre-SD is  $\PSPACE$-complete~\cite[Theorem 3]{AKL21}. The automaton constructed for their lower bound is also a reachability automaton, thus providing our lowerbound. 
    The upper bound follows from our results (\cref{lemma:sdweak-is-mr,lemma:SR-implies-SD}) that a weak (or reachability) automaton is SR if and only if it is pre-semantically deterministic.
\end{proof}
\begin{question}[Resolver-(co)Buchi]
    Given a finite memory resolver~$\Rc$ for a (co)B\"uchi automaton $\Ac$, is $\Rc$ an almost-sure resolver for $\Ac$?
\end{question}
\begin{restatable}{lemma}{lemmaUndecidablecoBuchi}\label{lemma:UndecidablecoBuchi}
        The problem \textsf{Resolver-coBuchi} is undecidable.
\end{restatable}
Our reduction to prove undecidability of the problem of \textsf{Resolver-coBuchi} is by reducing an instance of the problem of checking the emptiness of probabilistic B\"uchi automaton under the positive semantics of acceptance (the word is accepted if there is an accepting run of non-zero probability). This problem was proved undecidable by Baier, Bertrand, and Gr\"o{\ss}er~\cite[Theorem~2]{BBG08}.    

The dual problem of emptiness checking for probabilistic coB\"uchi automata is decidable. Therefore, we show the following lemma using a different reduction.
\begin{restatable}{lemma}{lemmaUndecidableBuchi}\label{lemma:UndecidableBuchi}
    The problem \textsf{Resolver-Buchi} is undecidable.
\end{restatable}

We reduce from the zero-isolation problem for finite probabilistic automata. We define the problem formally in the appendix, but informally, the zero-isolation problems for finite probabilistic automata asks if there are words for which probability of acceptance is positive and reaches arbitrarily close to $0$. This problem was proved undecidable by Gimbert and Oaulhadj~\cite[Theorem 4]{GO10}. 


\section{Discussion}
We introduced several classes of automata that act as intermediary classes between deterministic and nondeterministic automata. These new classes of automata are distinct from similar notions in the literature. We extensively compared the novel classes of MA, MR, and SR automata with HD and SD automata. Here, we briefly discuss a few other related notions such as almost-DBP or good-for-MDP automata.  
In an effort to also reason about history deterministic and semantically deterministic automata, Abu Radi, Kupferman, and Leshkowitz~\cite{AKL21} introduced and studied a class of automata they called almost-determinisable by pruning (almost-DBP). 
While their definition is related to the probability that a run on a randomly generated word in the language is accepting, we study the notion where a randomly generated run by a resolver on every word in the language is accepting.  
To further highlight the difference, almost-DBP B\"uchi auomtata are the same as semantically deterministic B\"uchi automata, whereas SR B\"uchi automata are a different class than SD B\"uchi auomtata.  

Good-for-MDP automata are automata that admit compositionality with MDPs that make them relevant for reinforcement learning and MDP model checking~\cite{HPSS0W20}. We remark that SR automata are good-for-MDP, since an almost-sure resolver for an automaton can be used as a strategy for syntactic satisfaction objective for the product of automaton with any MDP. The converse is not true, however, since good-for-MDP B\"uchi automata recognise all $\omega$-regular languages~\cite[Section 3.2]{HPSS0W20} but SR B\"uchi automata are only as expressive as deterministic B\"uchi automata (\cref{theorem:pih}).

Although we have studied the classes of MA, SR, and MR automata comprehensively, some fundamental questions about these classes are still unknown. For the class of MA automata, the exact probability distributions in the resolver do not matter (\cref{lemma:indifferent2probabilities}). However, the same remains unknown for SR and MR automata. 
\begin{question}
Are almost-sure resolvers that resolve nondeterminism in stochastically resolvable settings indifferent to the exact probabilities distribution?  
\end{question}
Despite the problem of resolver checking for the cases of B\"uchi or coB\"uchi being undecidable, more study is needed to conclude the complexity status of the class membership problem for SR and MR automata. 
\begin{question}
    Is it decidable to check if a given (co)B\"uchi automaton is SR?
\end{question}
Our result of \textsf{Resolver-coBuchi} being undecidable is somewhat surprising when considered together with \cref{theorem:coBuchiHDisSR}, where we showed that every SR coB\"uchi automaton can be efficiently converted to an MA coB\"uchi automaton in polynomial time without any additional states. Thus, it might be reasonable to expect that the above question can be answered affirmatively.

\cref{theorem:coBuchiHDisSR} also implies that HD coB\"uchi automata have language-equivalent MA coB\"uchi automata with at most as many states. We ask whether the same holds for parity automata, and thus, if it is  possible to trade the exponential memory for resolvers in HD parity automata with randomness.  
\begin{question}
    Does every HD parity automaton have a language-equivalent MA parity automaton of the same size?
\end{question}



\bibliography{gfg}

\newcommand{\etalchar}[1]{$^{#1}$}
\begin{thebibliography}{BMM{\etalchar{+}}23}

\bibitem[AK20]{AK20}
Shaull Almagor and Orna Kupferman.
\newblock Good-enough synthesis.
\newblock In {\em {CAV} {(2)}}, volume 12225 of {\em Lecture Notes in Computer Science}, pages 541--563. Springer, 2020.

\bibitem[AK22]{AK22}
Bader {Abu Radi} and Orna Kupferman.
\newblock {Minimization and Canonization of {GFG} Transition-Based Automata}.
\newblock {\em Log. Methods Comput. Sci.}, 18(3), 2022.

\bibitem[AK23]{AK23}
Bader {Abu Radi} and Orna Kupferman.
\newblock {On Semantically-Deterministic Automata}.
\newblock In {\em International Colloquium on Automata, Languages, and Programming, {ICALP} 2023}, volume 261 of {\em LIPIcs}, pages 109:1--109:20. Schloss Dagstuhl - Leibniz-Zentrum f{\"{u}}r Informatik, 2023.

\bibitem[AKL21]{AKL21}
Bader {Abu Radi}, Orna Kupferman, and Ofer Leshkowitz.
\newblock {A Hierarchy of Nondeterminism}.
\newblock In {\em Mathematical Foundations of Computer Science, MFCS 2021}, volume 202 of {\em LIPIcs}, pages 85:1--85:21. Schloss Dagstuhl - Leibniz-Zentrum f{\"{u}}r Informatik, 2021.

\bibitem[BBG08]{BBG08}
Christel Baier, Nathalie Bertrand, and Marcus Gr{\"{o}}{\ss}er.
\newblock {On Decision Problems for Probabilistic B{\"{u}}chi Automata}.
\newblock In {\em FoSSaCS}, volume 4962 of {\em Lecture Notes in Computer Science}, pages 287--301. Springer, 2008.

\bibitem[BGB12]{BGB12}
Christel Baier, Marcus Gr\"{o}sser, and Nathalie Bertrand.
\newblock Probabilistic $\omega$-automata.
\newblock {\em J. ACM}, 59(1), March 2012.

\bibitem[BK18]{BK18}
Marc Bagnol and Denis Kuperberg.
\newblock {B{\"{u}}chi Good-for-Games Automata Are Efficiently Recognizable}.
\newblock In {\em Foundations of Software Technology and Theoretical Computer Science, {FSTTCS} 2018}, volume 122 of {\em LIPIcs}, pages 16:1--16:14. Schloss Dagstuhl - Leibniz-Zentrum f{\"{u}}r Informatik, 2018.

\bibitem[BKS17]{BKS17}
Udi Boker, Orna Kupferman, and Michal Skrzypczak.
\newblock {How Deterministic are Good-For-Games Automata?}
\newblock In {\em {FSTTCS}}, volume~93 of {\em LIPIcs}, pages 18:1--18:14. Schloss Dagstuhl - Leibniz-Zentrum f{\"{u}}r Informatik, 2017.

\bibitem[BL21]{BL21}
Udi Boker and Karoliina Lehtinen.
\newblock {History Determinism vs. Good for Gameness in Quantitative Automata}.
\newblock In {\em Proc.\ of {FSTTCS}}, pages 35:1--35:20, 2021.

\bibitem[BL23a]{BL23quantitative}
Udi Boker and Karoliina Lehtinen.
\newblock {Token Games and History-Deterministic Quantitative-Automata}.
\newblock {\em Logical Methods in Computer Science}, 19(4), 2023.

\bibitem[BL23b]{BL23}
Udi Boker and Karoliina Lehtinen.
\newblock {When a Little Nondeterminism Goes a Long Way: An Introduction to History-Determinism}.
\newblock {\em {ACM} {SIGLOG} News}, 10(1):24--51, 2023.

\bibitem[BMM{\etalchar{+}}23]{BNNSS22}
Tamajit Banerjee, Rupak Majumdar, Kaushik Mallik, Anne{-}Kathrin Schmuck, and Sadegh Soudjani.
\newblock {Fast Symbolic Algorithms for Omega-Regular Games under Strong Transition Fairness}.
\newblock {\em TheoretiCS}, 2, 2023.

\bibitem[Bor09]{Bor09}
{\'E}mile Borel.
\newblock Les probabilit{\'e}s d{\'e}nombrables et leurs applications arithm{\'e}tiques.
\newblock {\em Rendiconti del Circolo Matematico di Palermo (1884-1940)}, 27(1):247--271, Dec 1909.

\bibitem[Can17]{Can17}
Francesco Cantelli.
\newblock Sulla probabilista come limita della frequencza.
\newblock {\em Rend. Accad. Lincei}, 26:39, 1917.

\bibitem[CDGH15]{CDGH15}
Krishnendu Chatterjee, Laurent Doyen, Hugo Gimbert, and Thomas~A. Henzinger.
\newblock {Randomness for free}.
\newblock {\em Inf. Comput.}, 245:3--16, 2015.

\bibitem[CH11]{CH11}
Krishnendu Chatterjee and Monika Henzinger.
\newblock {\em Faster and Dynamic Algorithms For Maximal End-Component Decomposition And Related Graph Problems In Probabilistic Verification}, pages 1318--1336.
\newblock 2011.

\bibitem[Cha07]{Cha07}
Krishnendu Chatterjee.
\newblock Stochastic m{\"{u}}ller games are pspace-complete.
\newblock In {\em {FSTTCS}}, volume 4855 of {\em Lecture Notes in Computer Science}, pages 436--448. Springer, 2007.

\bibitem[CHP07]{CHP07}
Krishnendu Chatterjee, Thomas~A. Henzinger, and Nir Piterman.
\newblock {Generalized Parity Games}.
\newblock In {\em Foundations of Software Science and Computational Structures, 10th International Conference, {FOSSACS} 2007}, volume 4423 of {\em Lecture Notes in Computer Science}, pages 153--167. Springer, 2007.

\bibitem[DJW97]{DJW97}
Stefan Dziembowski, Marcin Jurdzinski, and Igor Walukiewicz.
\newblock {How Much Memory is Needed to Win Infinite Games?}
\newblock In {\em {LICS}}, pages 99--110. {IEEE} Computer Society, 1997.

\bibitem[EKS16]{EKS16}
Javier Esparza, Jan Křetínský, and Salomon Sickert.
\newblock {From LTL to deterministic automata}.
\newblock {\em Formal Methods in System Design}, 49(3):219--271, 2016.

\bibitem[ES22]{ES22}
R{\"{u}}diger Ehlers and Sven Schewe.
\newblock {Natural Colors of Infinite Words}.
\newblock In {\em {FSTTCS}}, volume 250 of {\em LIPIcs}, pages 36:1--36:17. Schloss Dagstuhl - Leibniz-Zentrum f{\"{u}}r Informatik, 2022.

\bibitem[GH10]{GH10}
Hugo Gimbert and Florian Horn.
\newblock Solving simple stochastic tail games.
\newblock In {\em {SODA}}, pages 847--862. {SIAM}, 2010.

\bibitem[GO10]{GO10}
Hugo Gimbert and Youssouf Oualhadj.
\newblock Probabilistic automata on finite words: Decidable and undecidable problems.
\newblock In {\em {ICALP} {(2)}}, volume 6199 of {\em Lecture Notes in Computer Science}, pages 527--538. Springer, 2010.

\bibitem[GOP11]{GOP11}
Hugo Gimbert, Youssouf Oualhadj, and Soumya Paul.
\newblock {Computing Optimal Strategies for Markov Decision Processes with Parity and Positive-Average Conditions}.
\newblock working paper or preprint, January 2011.

\bibitem[HD05]{HD05}
Paul Hunter and Anuj Dawar.
\newblock {Complexity Bounds for Regular Games}.
\newblock In {\em {MFCS}}, volume 3618 of {\em Lecture Notes in Computer Science}, pages 495--506. Springer, 2005.

\bibitem[HK23]{HK23}
Emile Hazard and Denis Kuperberg.
\newblock {Explorable Automata}.
\newblock In {\em {CSL}}, volume 252 of {\em LIPIcs}, pages 24:1--24:18. Schloss Dagstuhl - Leibniz-Zentrum f{\"{u}}r Informatik, 2023.

\bibitem[HP06]{HP06}
Thomas~A. Henzinger and Nir Piterman.
\newblock {Solving Games Without Determinization}.
\newblock In {\em Computer Science Logic, {CSL} 2006}, volume 4207 of {\em Lecture Notes in Computer Science}, pages 395--410. Springer, 2006.

\bibitem[HPS{\etalchar{+}}20]{HPSS0W20}
Ernst~Moritz Hahn, Mateo Perez, Sven Schewe, Fabio Somenzi, Ashutosh Trivedi, and Dominik Wojtczak.
\newblock {Good-for-MDPs Automata for Probabilistic Analysis and Reinforcement Learning}.
\newblock In {\em {TACAS} {(1)}}, volume 12078 of {\em Lecture Notes in Computer Science}, pages 306--323. Springer, 2020.

\bibitem[Kec12]{Kec12}
A.~Kechris.
\newblock {\em Classical Descriptive Set Theory}.
\newblock Graduate Texts in Mathematics. Springer New York, 2012.

\bibitem[KPV06]{KPV06}
Orna Kupferman, Nir Piterman, and Moshe~Y. Vardi.
\newblock Safraless compositional synthesis.
\newblock In {\em {CAV}}, volume 4144 of {\em Lecture Notes in Computer Science}, pages 31--44. Springer, 2006.

\bibitem[KS15]{KS15}
Denis Kuperberg and Michal Skrzypczak.
\newblock {On Determinisation of Good-for-Games Automata}.
\newblock In {\em {ICALP} {(2)}}, volume 9135 of {\em Lecture Notes in Computer Science}, pages 299--310. Springer, 2015.

\bibitem[KV05]{KV05}
Orna Kupferman and Moshe~Y. Vardi.
\newblock Safraless decision procedures.
\newblock In {\em {FOCS}}, pages 531--542. {IEEE} Computer Society, 2005.

\bibitem[Lan69]{Lan69}
L.~H. Landweber.
\newblock Decision problems for $\omega$-automata.
\newblock {\em Mathematical Systems Theory}, 3(4):376--384, 1969.

\bibitem[LP25]{LP25}
Karoliina Lehtinen and Aditya Prakash.
\newblock {The 2-Token Theorem: Recognising History-Deterministic Parity Automata Efficiently}.
\newblock 2025.
\newblock Private communication, to appear in STOC 2025. A publicly available proof for results can be found in [Pra25].

\bibitem[McN66]{McN66}
Robert McNaughton.
\newblock Testing and generating infinite sequences by a finite automaton.
\newblock {\em Information and Control}, 9(5):521--530, 1966.

\bibitem[NW98]{NW98}
Damian Niwinski and Igor Walukiewicz.
\newblock Relating hierarchies of word and tree automata.
\newblock In {\em {STACS}}, pages 320--331, 1998.

\bibitem[Pra24]{Pra24a}
Aditya Prakash.
\newblock {Checking History-Determinism is NP-hard for Parity Automata}.
\newblock In {\em {FoSSaCS {(1)}}}, volume 14574 of {\em Lecture Notes in Computer Science}, pages 212--233. Springer, 2024.

\bibitem[Pra25]{Pra25}
Aditya Prakash.
\newblock {\em History-Deterministic Parity Automata: Games, Complexity, and the 2-Token Theorem}.
\newblock PhD thesis, University of Warwick, 2025.

\bibitem[Saf88]{Saf88}
Shmuel Safra.
\newblock On the complexity of omega-automata.
\newblock In {\em {FOCS}}, pages 319--327. {IEEE} Computer Society, 1988.

\bibitem[SC16]{Shi16}
A.N. Shiryaev and D.M. Chibisov.
\newblock {\em Probability-1}.
\newblock Graduate Texts in Mathematics. Springer New York, 2016.

\bibitem[SS10]{SS10}
E.M. Stein and R.~Shakarchi.
\newblock {\em Complex Analysis}.
\newblock Princeton lectures in analysis. Princeton University Press, 2010.

\bibitem[Tar72]{Tar72}
Robert Tarjan.
\newblock Depth-first search and linear graph algorithms.
\newblock {\em SIAM Journal on Computing}, 1(2):146--160, 1972.

\bibitem[Var85]{Var85}
Moshe~Y. Vardi.
\newblock Automatic verification of probabilistic concurrent finite-state programs.
\newblock In {\em {FOCS}}, pages 327--338. {IEEE} Computer Society, 1985.

\end{thebibliography}
\appendix
\section{Useful definitions and lemmas}
\subsection{Lemmas concerning probability and analysis}\label{app:defs}
We state two important lemma about an infinite sequence of events below, commonly referred to as the Borel-Cantelli lemma and the second Borel-Cantelli Lemma. Note that the second Borel-Cantelli lemma however only holds if the infinite sequence of events are independent. 

\begin{lemma}[Borel-Cantelli lemma \cite{Bor09,Can17, Shi16}]\label{lemma:borellcantelli}
    Let $E_1,E_2\dots,E_n$ be an infinite sequence of events in some probability space. If the sum of probabilities of the events is finite, then the probability that infinitely many of them occur is $0$. That is, $$\text{if }\sum_{n=1}^\infty\Pr[E_n]<\infty \text{ then } \Pr[\limsup_{n\to\infty}E_n] = 0.$$
\end{lemma}
\begin{lemma}[Second Borel-Cantelli lemma~\cite{Bor09,Can17, Shi16}]\label{lemma:secondborellcantelli}
      Let $E_1,E_2\dots,E_n$ be an infinite sequence of \emph{independent} events in some probability space. If the sum of probabilities of the events is infinite, then almost-surely infinitely many of them occur. That is, $$\text{if }\sum_{n=1}^\infty\Pr[E_n] = \infty \text{ then } \Pr[\limsup_{n\to\infty}E_n] = 1.$$  
\end{lemma}
We will also use the following standard result from analysis.
\begin{proposition}[\!\!\protect{\cite[Proposition~3.1]{SS10}}]\label{prop:complexanalysisConverge}
If the infinite sum $\sum_{i=1}^\infty|a_n|<\infty$, then the infinite product $\prod_{i=1}^\infty(1+a_n)$ converges. Moreover, the product converges to 0 if and only if at least one of its factors is 0.
\end{proposition}

\subsection{Partial-observation games}

We deal with 2-player turn-based partial-observation stochastic games in this paper that are played between Eve and Adam. These are given by a tuple $$\Gc=(V_{\eve},V_{\adam},E,\Act_{\eve},\Act_{\adam},\Oc_{\eve},\Oc_{\adam},\delta)$$ with the following components.
\begin{enumerate}
    \item \emph{Arena.} The set of vertices is given by $V=V_{\eve}\uplus V_{\adam}$, with Eve (resp.\ Adam) said to own the vertices in $V_{\eve}$ (resp.\ $V_{\adam}$). Edges $E$ are directed edges that are a subset of the set $(V_{\eve} \times V_{\adam}) \cup (V_{\adam}\times V_{\eve})$.
    \item \emph{Actions.} $\Act_{\eve}$ (resp.\ $\Act_{\adam}$) is a finite set of actions for Eve (resp.\ Adam). We use $\Act=\Act_{\eve} \uplus \Act_{\adam}$ to denote the set of actions for both players.  
    \item \emph{Transition function.} The function $\delta$ is a probabilistic transition function $$\delta:V\times \Act \xrightarrow{} \Distribution(E)$$ that, for each vertex in $V_{\eve}$ (resp.\ $V_{\adam}$) and action in $\Act_{\eve}$ (resp.\ $\Act_{\adam}$), assigns a probability distribution to the set of outgoing edges from that vertex. 
    \item \emph{Observations.} The set $\Oc_{\eve}$ and $\Oc_{\adam}$ are partitions of $E$. These uniquely map each vertex $e\in E$ to its observations for Eve and Adam as $\Oc_{\eve}(e)$ and $\Oc_{\adam}(e)$, respectively.
\end{enumerate}

A \emph{play} of a game $\Gc$ proceeds as follows. Starting with a token at an initial vertex $v$, the player who owns that vertex chooses an action $a$ from their set of actions, and the token is moved along an edge $e=(v,v')$ with probability $\delta(v,a)\circ (e)$. The new position of the token is $v'$, from where the play proceeds similarly, for infinitely many rounds. 

Thus, each play is an infinite path $\rho=e_0 e_2 e_3 \dots$ in the arena. For the play $\rho$, the \emph{observation sequence} for Eve (resp.\ Adam) is given by $\Oc_{\eve}(e_0 e_1 e_2 \dots)$ (resp.\ $\Oc_{\adam}(e_0 e_1 e_2\dots)$). A \emph{finite play} is a finite prefix of an infinite play.

A \emph{strategy} $\sigma$ for Eve is a partial function $$\sigma:\Oc_{\eve} (E^*)\xrightarrow{}\Distribution(\Act_{\eve})$$ from her observation sequence of finite plays ending at an Eve's vertex to her actions. The strategy $\sigma$ for Eve is said to be \emph{pure} if it assigns, to any observation sequence of a finite play ending at an Eve's vertex, each action either the probability 0 or 1. A pure strategy can equivalently be represented by a function $\sigma:\Oc_{\eve} (E^*)\xrightarrow{}\Act_{\eve}$. A strategy of Eve that is not pure is said to be a \emph{random strategy} for Eve. Pure and random strategies for Adam are defined analogously. 

An \emph{objective} for Eve is given by a set $\Theta \subseteq E^{\omega}$ of plays. A Borel objective is a Borel-measurable set in the Cantor topology, or also known as the product topology, on $E^{\omega}$~\cite[Chapter 17]{Kec12}. We will consider $\omega$-regular objectives, that is, objectives that can be given by a nondeterministic B\"uchi automaton that has its alphabet as $E$, which are Borel-measurable~\cite{Lan69}.

For a game $\Gc$ with an $\omega$-regular winning objective $L$, strategy $\sigma$ for Eve, and strategy $\tau$ for Adam, the probability that a play generated when Eve plays according to $\sigma$ and Adam plays according to $\tau$ is in $L$ is well-defined~\cite[Lemma 4.1]{Var85}, and we denote it by $\Pr_{\sigma,\tau}(\Gc)$.

We say that a strategy $\sigma$ for Eve in $\Gc$ is \emph{almost-sure} (resp.\ \emph{positive}) winning if for all strategies $\tau$ of Adam in $\Gc$, we have that $\Pr_{\sigma,\tau}=1$ (resp.\ $\Pr_{\sigma,\tau}>0$). We say that a strategy $\sigma$ for Eve in $\Gc$ is \emph{surely winning} or just \emph{winning} if Eve wins all possible plays produced in $\Gc$ in which Eve is playing according to $\sigma$. Almost-sure winning, positive winning, and surely winning strategies for Adam are defined analogously. 

\subsubsection*{Restrictions of games} 
We will consider the following restrictions of partial-observation games.
\paragraph*{Complete-observation games} This is the case when both the players have complete observation, i.e., $$\Oc_{\eve} = \Oc_{\adam} =\{\{e\}\mid e\in E\}.$$

\paragraph*{Non-stochastic games} This is the case when the transition function $\delta$ is a function $\delta:V\times \Act \xrightarrow{} E$ that assigns an edge to each action and vertex instead of a probability distribution over edges.

\paragraph*{One-player games}
One-player games for are games where there is a unique action available to either Eve or Adam, i.e., either $\Act_{\eve}=\{a\}$ or \ $\Act_{\adam}=\{a\}$. 

We use the above restrictions to define the games we deal with in this paper rigorously. We note that the HD~game (\cref{defn:hd-game}) and the 2-token game (\cref{defn:twotokengame}) on an automaton are examples of complete-observation nonstochastic games. 

\paragraph*{Stochastic-resolvability game}
The stochastic-resolvability (SR) game on $\Ac$ is a nonstochastic\footnote{Yes, the irony of this sentence is not lost on us.} game defined as follows.
\begin{definition}
    For a parity automaton $\Ac=(Q,\Sigma,\Delta,q_0)$, the the SR game on $\Ac$ is a nonstochastic game defined as follows. Eve's vertices are given by $V_{\eve}=Q\times \Sigma$, while Adam's vertices are given by $V_{\adam}=Q$. The game has the following edges.
    \begin{enumerate}
        \item $\{q\xrightarrow{} (q,a) \mid q \in Q, a \in \Sigma\}$ (Adam chooses a letter) 
        \item $\{(q,a) \xrightarrow{} (q') \mid q\xrightarrow{a:c}q' \in \Delta$\}. (Eve chooses a transition on her token)
    \end{enumerate}
    The game starts at the vertex $q_0$. Eve's actions are given by $\Act_{\eve}= \Delta$, while Adam's actions are given by $\Act_{\adam}=\Sigma$. The transition function $\delta$ is defined trivially.

    \noindent Eve has complete observation, that is $\Oc_{\eve}=\{\{e\}\mid e\in E\}$, while Adam has \emph{no observation}, that is, $\Oc_{\adam}=\{\{E\}\}$. 
    
    \noindent The winning condition for Eve is the same as in the HD game, i.e., a play of the SR game on $\Ac$ is winning for if the following condition holds: if the word formed by Adam's choice of letters is in $\Lc(\Ac)$ then the run formed on that word by Eve's choice of transitions is accepting.
\end{definition}

We now prove \cref{lemma:random-is-pure}. 
\randomispure*
\begin{proof}
    If $\Mc$ is a finite-memory strategy using which she almost-surely wins the SR game on $\Ac$ against all Adam's strategies, then it is clear that the same strategy is a almost-sure resolver for $\Ac$. Indeed, for any word~$w$ in $\Lc(\Ac)$, consider the strategy of Adam in the SR game, where he picks letters from $w$ in sequence. Then, Eve's strategy builds an accepting run almost-surely, and therefore, her strategy is an almost-sure resolver for~$\Ac$.

    For the other direction, suppose $\Mc$ is an almost-sure resolver for $\Ac$, and note that $\Mc$ can also be seen as a finite-memory strategy for Eve in the SR game on $\Ac$. We will show that Eve wins any play almost-surely using the strategy $\Mc$ against any Adam's strategy. To see this, consider the one-player partial-observation stochastic game $\Gc$ for Adam where we fix Eve's strategy $\Mc$. We know, since $\Ac$ is SR, that Adam does not have a pure strategy, i.e., Adam does not have a strategy that uses no randomness to win with a positive probability. For one-player partial-observation stochastic games, random strategies are as powerful as pure strategies in games with $\omega$-regular winning conditions~\cite[Theorem 7]{CDGH15}. It follows that Adam has no strategy in $\Gc$ to win with positive probability, as desired.  
\end{proof}

\paragraph*{Markov decision processes} A Markov decision process, or MDP for short, is a one-player complete-observation stochastic game, where we call the player Adam. MDPs can equivalently be represented by the tuple $\Mc = (V_A, V_R, E, \delta)$, with the following components.
\begin{enumerate}
    \item \emph{Arena.} The set of vertices is given by $V = V_A \uplus V_R$, where $V_A$ are vertices owned by the player, and $V_R$ are \emph{stochastic vertices}. The set $E$ consists of directed edges from $V_A$ to $V_R$ or from $V_R$ to $V_A$.
    \item \emph{Stochastic vertices.} The function $\delta$ is a partial function $\delta: E_R \xrightarrow{} \mathbb{U}$, which assigns a probability to the set of outgoing edges $E_R$ that start at a stochastic vertex, such that the sum of probabilities of the outgoing edges from each vertex in $V_R$ is $1$. 
\end{enumerate}

A play of the MDP $\Mc$ starts at a vertex $v_0$, and in round $i$ where the current position is $v_i$:
\begin{enumerate}
    \item if $v_i$ is a vertex in $V_A$ then \emph{Adam} selects an outgoing edge $e_i=(v_i,v_{i+1})$ from $v_i$;
    \item otherwise, if $v_i$ is a vertex in $V_R$, then an outgoing edge $e_i=(v_i,v_{i+1})$ is chosen with probability $\delta(e_i)$. 
\end{enumerate}
The play then goes to $v_{i+1}$, from where round $(i+1)$ begins.

Objectives and almost-sure (resp.\ positive) winning strategies for Adam in MDPs are defined similar to how we defined it for the games above. In \cref{appendixsubsec:np-completeness}, we will deal with MDPs with Muller conditions.

We will now show that resolvers for MA automata are indifferent to probabilities.
\begin{lemma}\label{lemma:indifferent2probabilities}
For any memoryless-stochastically resolvable automaton $\Ac$, if there is a memoryless resolver $\Mc$ that is an almost-sure winning strategy for the HD game over $\Ac$ then any memoryless resolver $\Mc'$ that assigns nonzero probabilities to the same set of transitions of $\Ac$ is also a memoryless resolver for $\Ac$ using which Eve wins the HD game. 
\end{lemma}
To prove \cref{lemma:indifferent2probabilities}, we will use simulation games.

\begin{definition}\label{df:simulation-game}
        The \emph{simulation game} of $\Ac$ by $\Bc$ is a complete-observation non-stochastic game. A play of this game starts with Adam's token and Eve's token in the initial states of $\Bc$ and $\Ac$, respectively, and for infinitely many rounds. In round $i$ of the game when Eve's token is at $q_i$ and Adam's token is at $p_i$: 
    \begin{enumerate}
        \item Adam selects a transition $p_i \xrightarrow{a_i:c} p_{i+1}$ in $\Ac$ along which he moves his token to $p_{i+1}$, on a letter $a_i$;
        \item Eve moves her token along a transition $q_i \xrightarrow{a_i: c'} q_{i+1}$ in $\Bc$ on the same letter.
    \end{enumerate}
The play results in a sequence of Adam's selected transition on $\Ac$, which in turn also forms a word, and the sequence of Eve's selected transitions forms a run on the same word on $\Bc$. Eve wins such a play if Adam's run is rejecting or if Eve's run is accepting. 
\end{definition}
If Eve has a strategy to win the simulation game of $\Ac$ by $\Bc$, then we say that $\Bc$ \emph{simulates} $\Ac$. 

\begin{lemma}\label{lemma:hd-equiv-detsimulation}
     Let $\Ac$ be a nondeterministic parity automaton and $\Dc$ be a deterministic parity automaton that is language equivalent to $\Ac$. Then Eve wins the HD game on $\Ac$ if and only if $\Ac$ simulates $\Dc$.
\end{lemma}
\begin{proof}
    If Eve has a strategy to win the HD game on $\Ac$, then she can use the same strategy to win the simulation game of $\Dc$ by $\Ac$, where she ignores the transitions that Adam's token takes. For the converse, suppose that Eve wins the simulation game of $\Dc$ by $\Ac$ using the strategy $\sigma$. Then, she can use $\sigma$ to play the HD game on $\Ac$, where she picks transitions on her token using $\sigma$ by playing the simulation game against the unique run of $\Dc$ on the word Adam has played in the HD game so far. If Adam's word in the HD game is in the language, then the unique run of $\Dc$ on that word is accepting, and since $\sigma$ is a winning strategy, the run on Eve's token in the HD game is accepting as well.
\end{proof}

We now prove \cref{lemma:indifferent2probabilities} using the proof of \cref{lemma:hd-equiv-detsimulation} above.

\begin{proof}[Proof of \cref{lemma:indifferent2probabilities}]
Let $\Dc$ be a deterministic parity automaton that is language-equivalent to $\Ac$. The simulation game between two parity automata is a Rabin game~\cite[Pages 154-155]{CHP07}. For stochastic Rabin games, the probability distribution of the stochastic nodes does not change the set of vertices from which the player wins almost surely, as long as the set of transitions assigned nonzero probabilities remains the same~\cite[Theorem~1]{BNNSS22}. 

In the proof of \cref{lemma:hd-equiv-detsimulation}, we showed that strategies for Eve in the HD game on $\Ac$ can easily be converted to strategies for Eve in the simulation game of $\Ac$ by $\Dc$ and vice versa. Combining this with the above fact, we reach the desired conclusion.
\end{proof}

\section{Appendix for \cref{sec:succandcomp}}
\subsection{Safety, reachability and weak automata}\label{app:succandcompsafety}\label{app:succandcompWeak}
We prove the theorem that every SD automaton, and therefore every pre-SD automaton, is determinisable by pruning.
\SDsafetyisDBP*
\begin{proof}
    Let $\Sc$ be a SD safety automaton. For each state $q$ in $\Sc$ and letter $a$ in the alphabet of $\Sc$, fix a language-preserving transition on $a$ that is outgoing from $q$, and consider the deterministic automaton $\Dc$ consisting of these transitions. It suffices to show that $\Lc(\Sc) \subseteq \Lc(\Dc)$. Let $w$ be a word in $\Lc(\Sc)$, and let $\rho$ be the unique run of $\Dc$ on $w$. For any finite prefix $u$ of $w$, consider the state $q$ that is reached in $\Dc$ upon reading the word $u$. Because $\Sc$ is SD, $u^{-1}w\in \Lc(\Ac,q)$, and hence, $q$ is not the rejecting sink state. It follows that $\rho$ does not reach the rejecting sink state, and hence is an accepting run. We obtain $w\in \Lc(\Dc)$, as desired.
\end{proof}

\lemmasdweakismr*
\begin{proof}
    Let $\Ac$ be a semantically deterministic weak automaton. We will show that the resolver which selects transitions uniformly at random constructs runs that are almost-surely accepting on words in $\Lc(\Ac)$. Let $w$ be a word in $\Lc(\Ac)$, $\rho_{acc}$ an accepting run of $\Ac$ on $w$, and $\rho$ a run of $\Ac$ on $w$, where transitions are chosen uniformly at random. Then there is a finite prefix $u$ of $w$, such that the run $\rho_{acc}$ after reading $u$, and on the word $w'=u^{-1}w$ only contains accepting transitions. Let $K=n2^n$, where $n$ is the number of states of $\Ac$. The crux of our proof is to show the following claim.
 \begin{claim}\label{claim:claiminsdweakismr}
        There is a positive probability $\epsilon>0$, such that on any infix of $w'$ that has length at least $K$, the segment of the run $\rho$ on that infix contains an accepting transition with probability at least $\epsilon$.
    \end{claim}
    \begin{proof}
        Let $v$ be an infix of $w'$ such that $w'=u'vw''$, and $v$ has length $K$. Suppose that the run $\rho$ after reading $uu'$ is at the state $q$, and $\rho_{acc}$ is at the state $p$. Consider the sequence of states $p_0p_1p_2\dots p_K$ in $\rho$ from $p$ on the word $v$, and the sequence of set of states that can be reached from $q$ on reading the prefixes of $v$:
        $$\{q_0\} \xrightarrow{a_0} Q_1 \xrightarrow{a_1} Q_2 \xrightarrow{a_2} \dots \xrightarrow{a_K} Q_K,$$ where $v=a_1a_2\dots a_K$ and $Q_{l+1}$ is the set consisting of states to which there is a transition from a state in $Q_l$ on the letter $a_l$.  By the pigeonhole principle, there are numbers $i<j$, such that $(p_i,S_i)=(p_j,S_j)$. 
        
        Let $v'$ be the word $v'=a_{i} a_{i+1} \dots a_j$, and consider the word $t = a_1 a_2 \dots a_{i-1} (v')^{\omega}$. Then $t$ is in $\Lc(\Ac,p)$, and since $\Ac$ is SD, $t$ is in $\Lc(\Ac,q)$ as well. Thus, there must be a run from a state in $S_i$ on the word $v'$ that contains an accepting transition. It follows that there is a run  from $q$ on the word $v$ that contains an accepting transition. Let $\epsilon=\frac{1}{d^K}$, where $d$ is the maximum outdegree in $\Ac$. Then, the probability that a random run from $q$ on $v$ contains an accepting transition is at least $\epsilon$, as desired.    
    \end{proof}
We now use \cref{claim:claiminsdweakismr} to prove that any random run from $q$ on $w'=u^{-1}w$ is almost-surely accepting. Let $w'=v_1v_2 v_3\dots$, where each $v_i$ is of length $K=n 2^n$, and $\rho_q$ be the suffix of the run $\rho$ on the word $w'$. Note that $\rho_q$ is rejecting if and only if $\rho_q$ contains finitely many priority $2$ transitions (since $\Ac$ is weak). Thus, we have the following chain of inequalities to show that $\rho_q$ is rejecting with probability $0$. 
\begin{equation*}
    \begin{split}
                &\prob[\rho \text{ is rejecting}] = \prob[\rho_q \text{ is rejecting}]\\
                &=\prob[\bigcup_{N\in\mathbb{N}} \rho_q \text{ does not contain a transition of}\\ 
                &\quad\quad\quad\quad\text{priority 2 after prefix $v_1v_2\dots v_N$}]\\
                &=\lim_{N\to\infty}\prob[\rho_q \text{ does not contain a transition of}\\ 
                &\quad\quad\quad\quad\quad\quad\text{priority 2 after prefix $v_1v_2\dots v_N$}] \\
        &\leq \lim_{N\to\infty} \prod_{n\geq N} (1-\epsilon) = 0\\  
    \end{split}
\end{equation*}
Thus, $\rho$ is almost-surely accepting, as desired.
\end{proof}

\subsection{CoB\"uchi automata}\label{app:succandcompCB}

\lemmaHDcoBuchinotMR*
\begin{proof}
    Consider the HD coB\"uchi automaton shown in \cref{fig:HDcoBuchinotMR}, which we re-illustrate in \cref{fig:HDcoBuchinotMRappendix} for convenience. 
\begin{figure}[ht]
\centering
        \begin{tikzpicture}
        \tikzset{every state/.style = {inner sep=-3pt,minimum size =20}}

    \node[state] (q0) at (0,0) {$q_0$};
    \node[state] (q1) at (0,2) {$q_1$};
    \node[state] (d1) at (2,0) {$d_1$};
    \node[state] (d2) at (2,2) {$d_2$};
    \node[state] (d3) at (4,2) {$d_3$};
    \path[-stealth]
    (-0.5,-0.5) edge (q0)
    (q0) edge [bend left = 15] node [left] {$x$} (q1)
    (q1) edge [red, dashed, bend left = 15] node [right] {$a$} (q0)
    (q0) edge [loop left] node [left] {$x$} (q0)
    (q0) edge [red, dashed, loop below] node [left] {$b$} (q0)
    (q1) edge [loop left] node [left] {$x$} (q1)
    (q1) edge node [above] {$b$} (d2)
    (d2) edge [red, dashed, loop above] node [left] {$b$} (d2)
    (d2) edge [red, dashed, bend left = 15] node [right] {$a$} (d1)
    (d1) edge [red, dashed, bend left = 15] node [left] {$b$} (d2)
    (q0) edge node [above] {$a$} (d1)
    (d1) edge [loop right] node [right] {$x,a$} (d1)
    (d2) edge [bend left = 15] node [above] {$x$} (d3)
    (d3) edge [bend left = 15] node [below] {$b$} (d2)
    (d3) edge [loop right] node [right] {$x$} (d3)
    (d3) edge [red, dashed] node [above] {$a$} (d1)
    
;
    \end{tikzpicture}
\caption{A HD coB\"uchi automaton that is not MR. The rejecting transitions are represented by dashed arrows.}\label{fig:HDcoBuchinotMRappendix}
\end{figure}
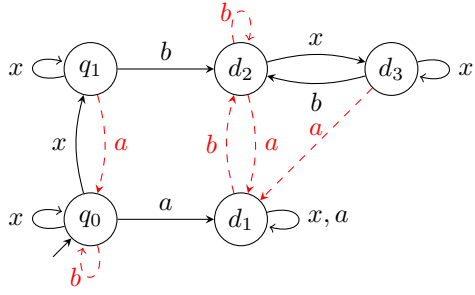 
    The automaton $\Cc$ has nondeterminism on the letter $x$ in the initial state $q_0$. Informally, in the HD game or the SR game, Eve needs to ``guess'' whether the next sequence of letters Adam will give forms a word in $x^* a$ or in $x^+ b$. The automaton $\Cc$
    recognises the language $$L=(x+a+b)^{*} ((x)^{\omega} + (x^* a)^{\omega} + (x^+ b)^{\omega}).$$ 
    \paragraph*{$\Cc$ is HD}  If Eve's token in the HD game reaches the state $d_1,d_2$, or $d_3$, then she wins the HD game from here onwards since her transitions are deterministic. At the start of the HD game on $q_0$, or whenever she is at $q_0$ after reading $a$ or $b$ in the previous round, she decides between staying at $q_0$ until a $a$ or $b$ is seen, or moving to $q_1$ in the first $x$ as follows.
    \begin{itemize}
        \item If the word read so far has a suffix in $x^{*}a$, then she stays in $q_0$ until the next $a$ or~$b$.
        \item If the word read so far has a suffix in $x^{+}b$, then she takes the transition to $q_1$ on~$x$.
        \item Otherwise, she stays in $q_0$ till the next $a$ or $b$.
    \end{itemize}
    Due to the language of $\Cc$ being the set of words which have a suffix in $x^{\omega},(x^{*}a)^{\omega},$ or $(x^+b)^{\omega}$, the above strategy guarantees that Eve's token moves on any word in $L$ in HD game to one of $d_1$ or $d_2$, from where she wins the HD game.
    
    \paragraph*{$\Cc$ is not MR} Note that the automaton $\Cc$ does not accept the same language if any of its transitions are deleted. Therefore, consider a memoryless resolver $\Mc$ for $\Cc$ that takes the self-loop on $x$ on $q_0$ with probability $(1-p)$ and the transition to $q_1$ on $x$ with probability $p$, for some $p$ satisfying $0<p<1$. 
    
    We will show that on the word $w = x a x^2 a x^3 a \dots$, the resolver $\Mc$ constructs a rejecting run with a positive probability. Let $\rho$ be a run on $w$ constructed by $\Mc$. We denote $w$ as $v_1 a v_2 a \dots $, where $v_i=x^i$ for each $i\geq 1$.
    
    If $\rho$ is ever at the state $q_0$ after reading $v_1 a v_2 a \dots v_i$, then $\rho$ is at the state $d_1$ after reading $a$, from where $\rho$ is accepting. Thus, $\rho$ is rejecting if and only if after reading the substring $v_i$, $\rho$ is at $q_1$. The probability that $\rho$, starting from $q_0$, ends on $q_1$ after reading $v_i=x^i$ is $(1-p^i)$. Thus, 
$$
                \prob[\rho \text{ is rejecting}] = \prod_{i=1}^n\left(1-p^{i}\right).
$$
The above quantity is positive due to \cref{prop:complexanalysisConverge}, and thus $\rho$ is not almost-surely accepting, as desired.
\end{proof}

\subsubsection*{Converting SR automata to MA automata}
Next we focus on providing a detailed proof for \cref{theorem:coBuchiHDisSR}.
\theoremcobuchisrtoma*
To prove \cref{theorem:coBuchiHDisSR}, we fix a coB\"uchi automaton $\Ac$ that is stochastically resolvable. We do the following relabelling of priorities on $\Ac$ to get another coB\"uchi automaton $\Cc$, such that a run in $\Ac$ is accepting if and only if that run is accepting in~$\Cc$. 
\paragraph*{Priority-reduction} Consider the graph $G$ consisting of the states of $\Ac$ and the transitions of $\Ac$ that have priority~$0$. Consider the strongly connected components of this graph, and for any priority $0$ transition that is not in any SCC, we change its priority to $1$ in $\Cc$. The rest of the transitions in $\Cc$ have the same priority as in $\Ac$.

\begin{proposition}\label{prop:priority-reduction}
A run in $\Cc$ is accepting if and only if the corresponding run is accepting in $\Ac$.
\end{proposition}
\begin{proof}
    Let $\rho$ be a run. If $\rho$ is a rejecting run in $\Ac$, then it contains infinitely many priority $1$ transitions in $\Ac$. It is clear then that $\rho$ contains infinitely many priority $1$ transitions in $\Cc$ as well. 
    
    Otherwise, if $\rho$ is accepting in $\Ac$ then $\rho$ contains finitely many priority $1$ transitions in $\Ac$. Thus, $\rho$ eventually stays in the same SCC in $G$, and therefore, $\rho$ contains finitely many priority $1$ transitions in~$\Cc$.
\end{proof}
It follows from \cref{prop:priority-reduction} that $\Ac$ and $\Cc$ are language equivalent, and since $\Ac$ is SR, so is $\Cc$: any almost-sure resolver for $\Ac$ is also an almost-sure resolver for $\Cc$.
\begin{corollary}\label{cor:lang-c-is-lang-a}
The automaton $\Cc$ is language equivalent to $\Ac$ and stochastically resolvable.
\end{corollary}

Recall the safety-automaton $\Cc_{\safety}$ that we had defined in \cref{subsec:sac-cobuchi}. The following observation is easy to see.

\begin{proposition}\label{prop:safe-cobuchi-langcomparison}
    For every state $q$ in $\Cc$, $\Lc(\Cc_{\safety},q)\subseteq \Lc(\Cc,q)$.
\end{proposition}
\begin{proof}
    If $\rho$ is an accepting run on $w$ in $(\Cc_{\safety},q)$, then the same run $\rho$ does not contain any priority $1$ transition in $(\Cc,q)$, and therefore, is accepting.
\end{proof}

Recall the definition of SR-covers and SR self-coverage presented in \cref{subsec:sac-buchi}. We next show that $\Bc$ has SR self-coverage.
\lemmaSRhassafeSRcoverage*
\begin{proof}
We fix an almost-sure resolver $\Mc$ for Eve in $\Cc$,  and let $\Pc$ be the probabilistic automaton that is the resolver-product of $\Mc$ and $\Cc$ (see \cref{sec:prelims}). We define $\Pc_{\safety}$ as a safety probabilistic automaton that is the \emph{safe-approximation} of $\Pc$ as follows, similar to how we defined $\Cc_{\safety}$. 
The automaton $\Pc_{\safety}$ has the same states as $\Pc$, and in addition, a rejecting sink state $q_{\bot}$. 

For each transition $\delta=(q,m)\xrightarrow{a:0}(q',m')$ in $\Pc$ of priority $0$ that has probability $p$, we add the transition $(q,m)\xrightarrow{a:0} (q',m')$ in $\Pc_{\safety}$ with the same probability $p$. For each state $(q,m)$ and letter $a$, we add a transition $(q,m)\xrightarrow{a:1}q_{\bot}$ that has probability $p'$, so that the sum of  the probabilities of all outgoing transitions from $(q,m)$ on $a$ is~$1$.

To prove \cref{lemma:coBuchiSRhassafeSRcoverage}, we suppose, towards contradiction, that
there is a state $q$ in $\Cc$, such that for every state $p$ coreachable to $q$ in $\Cc$, $(\Cc_{\safety},p)$ does not SR-cover $(\Cc_{\safety},q)$. In particular, for every state $(p,m)$ in $\Pc$, where $p$ is coreachable to $q$ in $\Cc$, we have $\Lc(\Pc_{\safety},(p,m)) \subsetneq \Lc(\Cc_{\safety},q)$. 

This implies that there is a finite word $u_{(p,m)}$ such that there is a run from $q$ to some state $q'$ on $u_{(p,m)}$ in $\Cc_{\safety}$ that only consists of safe transitions, while any run of $\Pc_{\safety}$ from $(p,m)$ on $u_{(p,m)}$ reaches the rejecting sink state $q_{\bot}$ with a positive probability $\epsilon_{(p,m)}$. Recall that we modified $\Ac$ to obtain $\Cc$ so that each priority $0$ transition occurs in an SCC consisting of only priority $0$ transitions. Thus, there is a word $v_{(p,m)}$ on which there is run from $q'$ to $q$ in $\Cc_{\safety}$ that contains only safe transitions. Consider the finite word $\alpha_{(p,m)}=u_{(p,m)} v_{(p,m)}$. Then there is a run from $q$ to $q$ in $\Cc$ on $\alpha_{(p,m)}$ that contains only priority $0$ transitions, while a run from $(p,m)$ on $\alpha_{(p,m)}$ in $\Pc$ contains a transition of priority~$1$ with probability $\epsilon_{(p,m)}>0$. 

We define the words $\alpha_{(p',m')}$ and the real number $\epsilon_{(p',m')}>0$ similarly, for all states $(p',m')$ in $\Pc$, such that $p'$ is coreachable to $q$ in $\Cc$. Define $\epsilon>0$ as the quantity $$\epsilon=\min\{\epsilon_{(p,m)}\mid (p,q) \in \CR(\Cc) \}.$$  

We will describe a strategy for Adam in the SR game on $\Cc$, using which he wins almost-surely against Eve's strategy $\Mc$. Adam starts by giving a finite word $u_q$, such that there is a run of $\Cc$ from its  initial state to $q$. Then Adam, from this point and at each \emph{reset}, selects a state $(p,m)$ of $\Pc$ uniformly at random, such that $p$ is coreachable to $q$ in $\Cc$ and $m$ is a memory-state in $\Mc$. He then plays the letters of the word $\alpha_{(p,m)}$ in sequence. Adam then \emph{resets} to select another such state $(p',m')$ with $(p',q) \in \CR(\Cc)$ and plays similarly. 

Consider a play of the SR game on $\Cc$ where Eve is playing according to her strategy $\Mc$ and Adam is playing according to the strategy described above. Note that at each reset, Eve's token is at a state $p$ that is weakly coreachable to $q$ and Eve has the memory $m$. Let $\lvert\Pc\rvert$ be the number of states in the probabilistic automaton. Adam picks the state $(p,m)$ at that reset with probability at least $\frac{1}{\lvert \Pc \rvert}$, from which point Eve's run on the word $\alpha_{(p,m)}$ contains a transition of priority $1$ with probability at least $\epsilon$. Thus, between every two consecutive resets, Eve's token takes a priority $1$ transition with probability at least $\frac{\epsilon}{\lvert \Pc \rvert}$. By the second Borel-Cantelli lemma, the run on Eve's token contains infinitely many priority $1$ transitions and hence is rejecting with probability 1. Thus, $\Mc$ is not an almost-sure resolver for Eve (\cref{lemma:random-is-pure}), which is a contradiction.  
\end{proof}

We note that SR-covers is a transitive relation.
\begin{lemma}\label{lemma:sr-cover-transitivity}
If $\Ac_1,\Ac_2,\Ac_3$ are three nondeterministic parity automata such that  $\Ac_1 \succ_{SR} \Ac_2$ and $\Ac_2 \succ_{SR} \Ac_3$, then $\Ac_1 \succ_{SR} \Ac_3$.
\end{lemma}
\begin{proof}
    Observe that if $\Ac_2 \succ_{SR} \Ac_3$, then $\Lc(\Ac_2) \supseteq \Lc(\Ac_3)$. Thus, if Eve has a strategy to construct a run in $\Ac_1$ that is almost-surely accepting on any word in $\Lc(\Ac_2)$, then the same strategy constructs a run that is almost-surely accepting on any word in $\Lc(\Ac_3)$. Thus, $\Ac_1$ SR-covers $\Ac_3$, as desired. 
\end{proof}

Thus, the following result follows from the definition of SD self-coverage and \cref{lemma:sr-cover-transitivity} above.

\lemmacobuchisometingdbp*
\begin{proof}
    Consider the directed graph $H$ whose vertices are states of $\Cc$. We add an edge from $q$ to $p$ in $\Cc$ if $(q,p)$ in $\Cc$ and $(\Cc_{\safety},p)$ SR-covers $(\Cc_{\safety},q)$. Note that if there is a path from $r$ to $s$ in $H$, then $(\Cc_{\safety},s)$ SR-covers $(\Cc_{\safety},r)$ and $(r,s)\in\WCR(\Cc)$. Since $\Cc$ has SR self-coverage, we note that every vertex has outdegree $1$. Thus, for every vertex $q$, there is a vertex $p$, such that there is a path from $q$ to $p$ and a path from $p$ to $p$ in $\Hc$. The conclusion follows.
\end{proof}

Note that if for some state $p \in \Cc$,  $(\Cc_{\safety},p)$ SR-covers $(\Cc_{\safety},p)$, then $(\Cc_{\safety},p)$ is SR, and therefore, $(\Cc_{\safety},p)$ is pre-SD (\cref{lemma:SR-implies-SD}), and therefore determinisable-by-pruning (\cref{lemma:sd-safety-is-dbp}). We thus call a state $p$ of $\Cc$ as \emph{safe-deterministic} if $(\Cc_{\safety},p)$ is SR. 

\paragraph*{Construction of $\Hc$} We will construct an MA automaton $\Hc$ that is language-equivalent to $\Cc$. The states of $\Hc$ consists of states that are safe-deterministic in $\Cc$.

For the transitions of $\Hc$, we start by fixing a uniform determinisation of transitions from every state that is safe-deterministic in $\Cc$ to obtain $\Cc'$, so that for every safe-deterministic state $q$, $\Lc(\Cc_{\safety},q)=\Lc(\Cc'_{\safety},q)$. Such a determinisation exists since $(\Cc_{\safety},q)$ is DBP. 

We add the transitions of $\Cc'$ in $\Hc$. Then, for every state $p$ and letter $a$ in $\Hc$, we add priority $1$ transitions from $p$ on $a$ to every state $q$ that is safe-deterministic in $\Cc$ and weakly coreachable to an $a$-successor of $p$ in~$\Cc$. 

We let the initial state of $\Hc$ be a safe-deterministic state that is weakly coreachable to the initial state of $\Cc$. Note that such a state exists due to \cref{lemma:cobuchi-srselfcoverage-implies-somtingsdbp}. This concludes our description of $\Hc$. 

In the next two lemmas, we show that $\Lc(\Hc)=\Lc(\Cc)$ and that $\Hc$ is MA.

\begin{lemma}\label{lemma:cobuchi-lang-equivalence}
    The automata $\Hc$ and $\Cc$ are language-equivalent.
\end{lemma}
\begin{proof}
    \textit{{$\Lc(\Cc) \subseteq \Lc(\Hc)$}:} Let $w$ be a word in $\Lc(\Cc)$, and $\rho$ an accepting run of $\Cc$ on $w$. Then, there is a decomposition~$w=uw'$, such that $u \neq \epsilon$ and $\rho$ after reading the prefix $u$ does not contain any priority $1$ transition on the suffix~$w'$ of~$w$. 
    Suppose $\rho$ is at the state $q$ after reading $u$. Then there is a safe-deterministic state $p$, such that $p$ is safe-deterministic, $(p,q) \in \CR(\Cc)$, and $(\Cc_{\safety},p)$ SR-covers $(\Cc_{\safety},q)$. Thus, $\Lc(\Cc_{\safety},p)\supseteq \Lc(\Cc_{\safety},q)$. Since $w' \in \Lc(\Cc_{\safety},q)$, it follows that $w' \in \Lc(\Cc_{\safety},p)$.
    
    Observe that $(\Cc_{\safety},p)=(\Hc_{\safety},p)$, and for each state in $\Hc$ and a letter in $\Sigma$, there is at most one priority 0 transition from that state on that letter. Thus, there is a unique run from $p$ on $w'$ in $\Hc$ that contains only priority 0 transitions. Consider the run of $\Hc$ on $w$ that takes arbitrary transitions until reading the penultimate letter of $u$, and then takes the transition to $p$ on the last letter of $u$, and then follows the unique run from $p$ on $w'$ that contains only priority 0 transitions. This is an accepting run of $\Hc$ on $w$, and thus $w \in \Lc(\Hc)$.
    
    \textit{{$\Lc(\Hc) \subseteq \Lc(\Cc)$}:} Let $w$ be a word in $\Lc(\Hc)$, and $\rho$ an accepting run of $\Hc$ on $w$. Then, there is a decomposition of $w$ as $uw'$, such that $\rho$ after reading $u$ does not contain any priority 1 transition on $w'$. Suppose $\rho$ is at the state $p$ after reading $u$. Then, $w' \in \Lc(\Hc_{\safety},p) = \Lc(\Cc_{\safety},p) \subseteq \Lc(\Cc,p)$, where the last equality holds due to \cref{prop:safe-cobuchi-langcomparison}. Recall that $\Cc$ is pre-semantically deterministic (\cref{lemma:SR-implies-SD}), i.e., contains a language-equivalent SD subautomaton $\Cc'$. Let $q$ be a state in $\Cc'$ that is reached after reading the word $u$. We will show the following claim.

    \begin{claim}\label{claim:cobuchilangquiv}
        $\Lc(\Cc,p) \subseteq \Lc(\Cc',q).$
    \end{claim}
            Note that $p$ and $q$ are weakly coreachable in $\Cc$. Thus, there is a sequence of states $p_1,p_2,\dots,p_k$ in $\Cc$ and finite words $u_0, u_1,\dots,u_k$, such that there are runs from the initial state of $\Cc$ to both $p$ and $p_1$ on the word $u_0$, to $p_i$ and $p_{i+1}$ on the word $u_i$ for each $i \in [1,k-1]$, and to $p_k$ and $q$ on $u_k$. We can pick the states $p_1,p_2,\dots,p_k$, such that there is a run from $q_0$ to $p_i$ in $\Cc'$ on the words $u_{i-1}$ and $u_{i}$ for each $i \in [1,k-1]$. Since $\Cc'$ is SD, this implies that $\Lc(\Cc',q)=\Lc(\Cc',p_1)$. Note, due to \cref{lemma:SDautomata}, that $$\Lc(\Cc,p) \subseteq u^{-1}\Lc(\Cc)=u^{-1}\Lc(\Cc')= \Lc(\Cc',q),$$  and thus the proof of the claim follows.
   
Using \cref{claim:cobuchilangquiv}, we note that there is an accepting run $\rho'$ on $w=uw'$ in $\Cc'$ which follows a run to $q$ on the word $u$, and then since $w' \in \Lc(\Cc,p) \subseteq \Lc(\Cc',q)$, follows an accepting run from $q$ on the word $w'$ in $\Cc'$. Since $\Cc'$ is a subautomaton of $\Cc$, $\rho'$ is also an accepting run of $\Cc$ on $w$, as desired.
\end{proof}

We next show that $\Hc$ is MA.
\begin{lemma}\label{lemma:cobuchi-h-is-ma}
    The coB\"uchi automaton $\Hc$ is a memoryless-adversarially resolvable automaton.
\end{lemma}
\begin{proof}
    Consider the following memoryless strategy for Eve in the HD game on $\Hc$, where from the state $q$ on the letter $a$:
    \begin{enumerate}
        \item if there is a priority $0$ transition from $q$ on $a$, then she picks that transition (observe that such a transition is unique);
        \item Otherwise, she picks an outgoing priority $1$ transition from $q$ on $a$ uniformly at random.
    \end{enumerate}
    We claim that this strategy is winning for Eve in the HD game on $\Ac$. To see this, consider a play in which Adam in the HD game is constructing a word letter-by-letter and Eve is building a run according to the above strategy. If Adam produces a word that is not in $\Lc(\Hc)$, then Eve wins trivially. Otherwise, eventually, Adam's word must have a prefix $u$ and a run which is at a state $p$ after reading $u$ and after which the suffix $w'$ that Adam constructs in the rest of the rounds by his letters is in $\Lc(\Hc_{\safety},p)$. Observe that the run from $p$ on $w'$ in $(\Hc_{\safety},p)$ is unique.
    Suppose, Eve's token is at the state $q$ after the word $u$ is read. Consider the run of Eve from $q$ where she picks transitions according to her strategy above, while Adam builds a word such that there is a unique run from $p$ on that word consisting of only priority 0 transitions. For every finite word $v$ that is a prefix for some infinite word in $\Lc(\Hc_{\safety},p)$, let $p_v$ be the unique state to which there is a run from $p$ to $p_v$ consisting of only priority 0 transitions.
    
    Then, after a word $v$ is read and Adam chooses a letter $a$, whenever Eve's token is at a state $q_v$ and Eve has no priority $0$ transition available to her, she takes the transition to $p_{va}$ with probability at least $1-\frac{1}{\lvert \Hc \rvert}$, where $\lvert \Hc \rvert$ is the number of states of $\Hc$. If Eve's token after the word $v'$ is read is at $p_{v'}$, then it is clear that Eve wins the HD game from here on. Thus, 
    \begin{equation*}
        \begin{split}
            &\prob[\text{Eve's run is rejecting}]\\ 
            &= \prob[\text{Eve takes infinitely many priority $1$ transitions and } \\ &\text{never takes the transition to $p_v$ after the word $v$}] \\ 
            &\leq \prod_{n\in\mathbb{N}} (1-\frac{1}{|\Hc|}) = 0. 
        \end{split}
    \end{equation*}
    Thus, Eve's run is almost-surely accepting, as desired.
\end{proof}

We have thus proved so far that every SR automaton has a language-equivalent MA automaton with at most as many states. We next show that we can find such an MA automaton for every input SR automaton efficiently, thus proving \cref{theorem:coBuchiHDisSR}.
\theoremcobuchisrtoma*
\begin{proof}
    Let $\Ac$ be an SR coB\"uchi automaton. The priority-reduction procedure relabels the priorities of transitions of  $\Ac$ to obtain a coB\"uchi automaton $\Cc$, in which every priority $0$ transitions occurs in a strongly connected component consisting of only priority $1$ transitions. This procedure is efficient since SCCs can be computed in linear time~\cite{Tar72}. From \cref{cor:lang-c-is-lang-a}, the automaton $\Cc$ is language-equivalent to $\Ac$ and $\Cc$ is stochastically resolvable. We then find states $p$ in $\Cc$, such that $(\Cc_{\safety},p)$ is HD and find a pure positional strategy from all such states. These are the safe-deterministic states, since SR safety automata are determinisable-by-pruning. Such states and this strategy can be found efficiently \cite[Theorem 4.5]{BL23quantitative}. Then, construction of $\Hc$ we described takes polynomial time, since the relations of coreachability and weak-coreachability can be computed in $\ptime$. This automaton $\Hc$ has as many states as $\Cc$ and hence $\Ac$, is language-equivalent to $\Ac$ (\cref{lemma:cobuchi-lang-equivalence}), and is MA (\cref{lemma:cobuchi-h-is-ma}). This concludes our proof.
\end{proof}
\subsection{B\"uchi automata}\label{succandcompBuchi}
\buchisdnotsr*
\begin{proof}
    Consider the B\"uchi automaton $\Bc$ in \cref{fig:BuchiSDbutnotSR}, which we reillustrate in \cref{fig:BuchiSDbutnotSRAgain} for the reader's convenience.  This automaton $\Bc$ has nondeterminism on the initial state $q_0$, and it recognises the language $$((x \cdot (a+b)\cdot y)^{*}(x\cdot (a+b)\cdot z))^{\omega}.$$ It is easy to verify that $\Bc$ is SD.
    \begin{figure}[ht]
    \centering
        \begin{tikzpicture}[auto]
        \tikzset{every state/.style = {inner sep=-3pt,minimum size =15}}

    \node[state] (s1)  at (0,0) {$q_0$};
    \node[state]  (s2)  at (2,0) {};
    \node[state] (s3) at (1,0.8) {$q_a$};
    \node[state] (s4) at (1,-0.8) {$q_b$};

    \node[state] (f1)  at (-1.2,0) {};
    \node[state]  (f2)  at (-2.2,0.8) {};
    \node[state] (f3) at (-2.2,-0.8) {};

    \path[->]
        (f3) edge node [left] {$x$}  (f2)
        (f1) edge node [yshift=1mm] {$y$} (f3)
        (f2) edge node [xshift=-2mm] {$a,b$} (f1)
        (s2) edge [double,bend left = 8] node {$z$}  (s1)
          (s1) edge  node [below,xshift=2mm,yshift=2mm] {$x$} (s3)
         (s1) edge  node [above,xshift=2mm,yshift=-2mm] {$x$} (s4)
         
         (s3) edge  node [above] {$a$} (s2)
         (s4) edge  node [below] {$b$} (s2)
         
         (s2) edge [double,bend left = 8] node [below] {$z$} (s1)
         (s2) edge [bend right = 8] node [above] {$y$} (s1)
         (s3) edge  node [above] {$b$} (f1)
         (s4) edge node [below] {$a$} (f1)
         (f1) edge node [above,xshift=1mm,yshift=-0.5mm] {$z$} (s1)
;
    \path[->,every node/.style={sloped,anchor=south}]
        ;
    \end{tikzpicture}
    \caption{A semantically deterministic B\"uchi automaton that is not stochastically resolvable. The accepting transitions are double-arrowed, and the initial state is $q_0$.}
    \label{fig:BuchiSDbutnotSRAgain}
\end{figure}
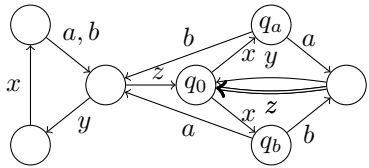 
    
    We will describe a strategy for Adam in the SR game on $\Bc$ using which Adam wins almost-surely in the SR game. This would imply, due to \cref{lemma:random-is-pure}, that $\Bc$ is not SR. 
    
    Note that when Eve's token is at $q_0$ in the SR game, Eve needs to guess whether the next letter is going to be $a$ or $b$. If she guesses incorrectly then her token moves to the left states--- states $l_1,l_2,$ and $l_3$, where she stays until a $z$ is seen. Adam's strategy in the SR game is as follows. Let $Y$ be the regular expression $xay+xby$ and $Z$ be the regular expression $xaz+xbz$. Note that both $Y$ and $Z$ consist of two words. Adam picks a word from the set $YZY^2ZY^3ZY^4Z \dots$ in the SR game, where from each occurrence of $Y$ or $Z$, he picks one of the two words in the regular expression with half probability. We next show that the probability $p_n$ that Eve's token, starting at $q_0$, takes an accepting transition on reading a word chosen randomly from $Y^{n}Z$ is $\frac{1}{2^{n+1}}$. Indeed, note that $1-p_n$ is same as the probability that Eve's token does not reach the left states on the word $Y^n Z$. This is the case only if, whenever Eve's token is at the state $q_a$ (resp.\ $q_b$), Adam picks the letter $a$ (resp.\ $b$). Thus, the probability that Eve's token does not reach a left state on the word $Y^n Z$ is $\frac{1}{2^{n+1}}$, and hence $p_n=\frac{1}{2^{n+1}}$. Thus, in the SR game where Adam picks the word as above,

  $$
         \prob[\text{Eve's run is accepting}] = \sum_{n\geq 1} \frac{1}{2^{n+1}} = \frac{1}{2}.
$$It then follows from the Borel-Cantelli lemma (\cref{lemma:borellcantelli}) that the probability that Eve's token takes infinitely many accepting transitions in the SR game is 0, as desired. 
\end{proof} 

\HDBuchinotMR*
\begin{proof}[Proof of \cref{lemma:HDBuchinotMR}]
    Consider the B\"uchi automaton $\Bc$ shown in \cref{fig:HDBuchinotMR}, which we re-illustrate  below for convenience.
    Let $\Sigma_{\diamond}=\{a,b,c,\diamond\}$ and $\Sigma=\{a,b,c\}$. Then the B\"uchi automaton $\Bc$ recognises the language  $\left[(L_1+L_2)^*(L_1L_1+L_2L_2)\right]^\omega$, where $L_1 = {\Sigma_\diamond}^*  c^+\diamond $ and $L_2 =  
    {\Sigma_\diamond}^* a \Sigma^* b^+\diamond
    $.  Equivalently, it accepts words in $(L_1+L_2)^\omega$ that are however not in $(L_1+L_2)^*(L_1L_2)^\omega$.  
    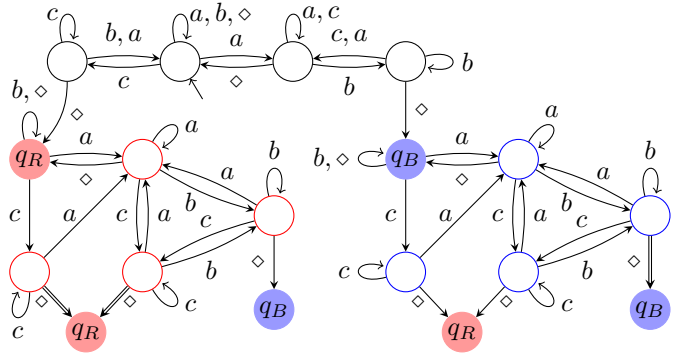
\begin{figure}[ht]
\centering
        \begin{tikzpicture}
        \tikzset{every state/.style = {inner sep=-3pt,minimum size =15}}

    \node[state] (q0) at (-1,1.5) {};
    \node[state] (q1) at (-2.5,1.5) {};
    \node[state] (q2) at (2,1.5) {};
    \node[state] (q3) at (0.5,1.5) {};
    \path[->] (-0.7,1) edge (q0);
    
    \node[state, fill=blue!40, blue!40] (r0) at (2,0.2) {$q_B$};
    \node[state, blue] (r1) at (2,-1.3) {};
    \node[state, blue] (r2) at (3.5,-1.3) {};
    \node[state, blue] (r3) at (3.5,0.2) {};    
    \node[state, blue] (r4) at (5.25,-0.55) {};    

    \node[state, fill=red!40, red!40] (l0) at (-3,0.2) {$q_R$};
    \node[state, red] (l1) at (-3,-1.3) {};
    \node[state, red] (l2) at (-1.5,-1.3) {};
    \node[state, red] (l3) at (-1.5,0.2) {};    
    \node[state, red] (l4) at (0.25,-0.55) {};    

    \node[state, fill=red!40, red!40] (l5) at (-2.25,-2.1) {$q_R$};    
    \node[state, fill=blue!40, blue!40] (l6) at (0.25,-1.8) {$q_B$};  
    
    \node[state, fill=red!40, red!40] (r5) at (2.75,-2.1) {$q_R$};    
    \node[state, fill=blue!40, blue!40] (r6) at (5.25,-1.8) {$q_B$};    
    \node (l52) at (-2.25,-2.1) {$q_R$};    
    \node (l62) at (0.25,-1.8) {$q_B$};  
    
    \node (r51) at (2.75,-2.1) {$q_R$};    
    \node (r61) at (5.25,-1.8) {$q_B$};    
    \node (l01) at (-3,0.2) {$q_R$};
    \node (r01) at (2,0.2) {$q_B$};
    
    \path[-stealth]
    (q0) edge [loop above] node [right] {$a,b,\diamond$} (q0)
    (q0) edge [bend left = 8] node [above] {$a$} (q3)
    (q3) edge [bend left = 8] node [below] {$\diamond$} (q0)
    (q0) edge [bend left = 8] node [below] {$c$} (q1)
    (q1) edge [bend left = 8] node [above] {$b,a$} (q0)
    (q3) edge [bend right = 8] node [below] {$b$} (q2)
    (q2) edge [bend right = 8] node [above] {$c,a$} (q3)
    (q3) edge [loop above] node [right] {$a,c$} (q3)
    (q1) edge [loop above] node [left] {$c$} (q1)
    (q2) edge [loop right] node [right] {$b$} (q2)

    (l0) edge [loop above] node [above] {$b,\diamond$} (l0)
    (l0) edge [bend left = 8] node [above] {$a$} (l3)
    (l3) edge [bend left = 8] node [below] {$\diamond$} (l0)
    (l1) edge  node [below, left] {$a$} (l3)
    (l0) edge  node [left] {$c$} (l1)
    (l3) edge [bend right = 8] node [left] {$c$} (l2)
    (l2) edge [bend right = 8] node [right] {$a$} (l3)
    (l3) edge [bend right = 8] node [below=3pt,left] {$b$} (l4)
    (l4) edge [bend right = 8] node [above=3pt,right] {$a$} (l3)
    (l4) edge [bend right = 8] node [above=0.5pt] {$c$} (l2)
    (l2) edge [bend right = 8] node [below=0.5pt] {$b$} (l4)
    
    (l3) edge [in=30,out=60,loop] node [right] {$a$} (l3)
    (l1) edge [in=240,out=270,loop] node [below] {$c$} (l1)
    (l2) edge [in=-30,out=-60,loop]  node [right] {$c$} (l2)
    (l4) edge [loop above] node [above] {$b$} (l4)

    (r0) edge [in=165,out=195,loop] node [left] {$b,\diamond$} (r0)
    (r0) edge [bend left = 8] node [above] {$a$} (r3)
    (r3) edge [bend left = 8] node [below] {$\diamond$} (r0)
    (r1) edge  node [below, left] {$a$} (r3)
    (r0) edge  node [left] {$c$} (r1)
    (r3) edge [bend right = 8] node [left] {$c$} (r2)
    (r2) edge [bend right = 8] node [right] {$a$} (r3)
    (r3) edge [bend right = 8] node [below=3pt,left] {$b$} (r4)
    (r4) edge [bend right = 8] node [above=3pt, right] {$a$} (r3)
    (r2) edge [bend right = 8] node [below=0.5pt] {$b$} (r4)
    (r4) edge [bend right = 8] node [above=0.5pt] {$c$} (r2)

    (r3) edge [in=30,out=60,loop]  node [above] {$a$} (r3)
    (r1) edge [loop left] node [left] {$c$} (r1)
    (r2) edge [in=-30,out=-60,loop] node [right] {$c$} (r2)
    (r4) edge [loop above] node [above] {$b$} (r4)

    (q1) edge [in=50,out=-90] node [right] {$\diamond$} (l0)
    (q2) edge node [right] {$\diamond$} (r0)

    (l1) edge[ double] node [left] {$\diamond$} (l5)
    (l2) edge[ double] node [right] {$\diamond$} (l5)
    (l4) edge node [left] {$\diamond$} (l6)

    (r1) edge node [left] {$\diamond$} (r5)
    (r2) edge node [right] {$\diamond$} (r5)
    (r4) edge [double] node [left] {$\diamond$} (r6)
;
    \end{tikzpicture}
\caption{A HD B\"uchi automaton that is not MR. The accepting transitions are represented by double arrows. All red-filled states ($q_R$) are identified as the same state, and all blue-filled states ($q_B$) are identified as the same state.}\label{fig:HDBuchinotMRAgain}
\end{figure}
    This automaton only has runs on words of the form $(L_1+L_2)^\omega$. When viewed as a finite-state automaton restricted to red (resp.\ red) states where the B\"uchi transitions are accepting transitions and $q_R$ (resp.\ $q_B$) is the initial state, this automaton accepts words in $L_1$ (resp.\ $L_2$). 
    
    For a run to contain the accepting transitions infinitely often, observe that it must visit the state $q_R$ or $q_B$ after reading some prefix in $(L_1+L_2)^*L_1$ and $(L_1+L_2)^*L_2$, respectively. Furthermore, observe that runs of words in $L_1$ and $L_2$ that start from the states $q_R$ or $q_B$, respectively, visit an accepting transition and then end at $q_R$ and $q_B$, respectively. 

\paragraph*{Eve's strategy in the HD game}
The only state with nondeterminism in automaton $\Ac$ is on the state $q_0$ on the letter $a$, from where  Eve can either choose to keep her token in $q_0$ using the transition that is the self loop: $\delta_1$, or she can move her token along the $a$-transition that goes right: $\delta_2$.  
Intuitively, Eve needs to guess at $q_0$ whether the word being input from now is going to be $L_1$ or $L_2$. If she guesses incorrectly, then her token ends up at the starting state, and she can guess again. If the resolver guesses correctly, Eve's token goes to the deterministic part of the game, i.e., the red or blue state, from where she wins the HD game. 

 We claim that the following Eve's strategy in the HD game is winning, where when she is at the initial state $q$ and Adam gives the letter $a$: if the longest prefix of the input word so far is in the language $(L_1 + L_2)^*L_1$ (rather than $(L_1 + L_2)^*L_2$), then she chooses the transition $\delta_1$ on her. If the longest prefix of the input word so far is in the language $(L_1 + L_2)^*L_2$ instead, then she chooses $\delta_2$ instead. 

Observe that for any word in the language, there are infinitely many prefixes which are either in the language $(L_1 + L_2)^*L_1L_1$ or $(L_1 + L_2)^*L_2L_2$. Consider the first time that the prefix is in $(L_1 + L_2)^*L_1L_1$ (the case for the prefix in the language $(L_1 + L_2)^*L_2L_2$ is similar). Let this prefix be denoted by $w \cdot u\cdot v$, such that $w\in (L_1 + L_2)^*$ and $u$ and $v$ are both words in $L_1$. Although there are many decompositions possible, we find one which ensures that the length of $u$ and $v$ are the shortest. 

Suppose Eve's token takes the transition $\delta_1$ after reading $w$. Then on reading a word in $L_1L_1$, her token would visit an accepting state, and therefore a deterministic component. Otherwise, is she chose the transition $\delta_2$ on reading $w$, then after reading the word $w\cdot u$, the longest prefix of $w\cdot u $ is in the language $(L_1 + L_2)^*L_1$, and therefore the she would chose the transition $\delta_1$. Using this transition and continuing to read a word in $L_1$, the run on her token reaches the deterministic part of the automaton, from where Eve wins the HD game. 

\paragraph*{No memoryless stochastic resolver}
We show that any resolver $\Mc_p$ that assigns with probability $0\leq p \leq 1$, the transition $\delta_1$ and with probability $1-p$, the transition $\delta_2$,
is not an almost-sure resolver. Note that if either of $\delta_1$ or $\delta_2$ is removed from the automaton $\Bc$, then the language changes, and hence if $p=0$ or $p=1$, $\Mc_p$ is not an almost-sure resolver. We therefore suppose that $0<p<1$.

We will construct a word in the language that is accepted with probability $<1$ by the resolver-product $\Mc_p \circ \Bc$. Consider the word $w = ac\diamond a^2c\diamond a^3c\diamond\dots \diamond a^{n}c\diamond\dots$, which is in the language since $w$ is in $(L_1)^\omega$. 
A run on $w$ is accepted if and only if some finite prefix of the run on the word visits the red state. For the run to visit a red state, the transition $\delta_1$ should be chosen by the resolver at every step wh $a^kc\#$ for some $k$. 

The probability that a run of the word $a^kc\diamond$ constructed using $\Mc_p$ starting from $q_0$ and ends at the red state $q_r$ is $p^k$, since the probability of $\delta_1$ being chosen at every step on a word $a^k$ is $p^k$. 
Therefore, the probability that a run constructed using $\Mc_p$ on the word $a^kc\diamond$ starting from $q_0$ \emph{does not} ends in the red state is $1-p^k$.

The probability that on the finite word $w_n = ac\diamond a^2c\diamond a^3c\diamond\dots a^{n}c\diamond$, a run constructed using the resolver $\Mc_p$,  \emph{does not} visit the red state even once is the probability that it does not visit the red state for any substring $a^{i}c\diamond$ which is 
\begin{align*}
\Pr[\text{a run on $w_n$ constructed using }\Mc_p&\\\text{ does not visit a red state}] &= \prod_{i=1}^n\left(1-p^{i}\right)    
\end{align*}

For the infinite word $w = ac\diamond a^2c\diamond a^3c\diamond\dots a^{n}c\diamond\dots$, again using the resolver $\Mc_p$,
\begin{align*}
    \Pr[\text{a run on $w$ using resolver }\Mc_p & \\\text{ does not } \text{visit a red state}] & = \prod_{i=1}^\infty\left(1-p^{i}\right)
\end{align*} 

Since the value $\sum_i^\infty |-p^i| <\infty$ for $0<p<1$ and each $p^i$ is positive, we obtain that $\prod_{i=1}^\infty\left(1-p^{i}\right)$ converges to a positive value (see \cref{prop:complexanalysisConverge}). We can argue further that since each of the elements in the product is strictly smaller than $1-p$, but strictly larger than $0$, therefore $0<\prod_{i=1}^\infty\left(1-p^{i}\right)<1-p$.

This shows that the word $w$ is not accepted with probability~$1$ using $\Mc_p$ as a resolver for any $0<p<1$ and therefore this automaton is not MR.
\end{proof}

\section{Appendix for \cref{sec:pih}}
\subsection{Parity index hierarchy is strict for SR parity automata}
We will give a complete proof of \cref{theorem:pih} below, which shows that the parity-index hierarchy on stochastically resolvable automata is strict.
\pih*
We begin by observing that SR $[i,j]$-parity automata are as expressive as MR-parity $[i,j]$  automata.

\begin{lemma}
    Every SR $[i,j]$-parity  automaton can be converted to an MR $[i,j]$-parity  automaton.
\end{lemma}
\begin{proof}
    Let $\Ac$ be an SR automaton, and let $\Mc$ be an almost-sure resolver for $\Ac$. Let $\Pc$ be the resolver-product of $\Mc$ and $\Ac$, and $\Nc$ be the underlying nondeterministic automaton of $\Pc$. Then  the memoryless resolver for $\Nc$ that picks transitions according to $\Pc$ is an almost-sure resolver, since $\Lc(\Pc) = \Lc(\Ac) = \Lc(\Nc)$.
\end{proof}

Thus, it suffices to show that MR $[i,j]$-parity automata are as expressive as deterministic $[i,j]$-parity  automata. 
Observe also that any deterministic $[i,j]$-parity  automaton is trivially an MR $[i,j]$-parity automaton. For the other direction, let $L$ be an $\omega$-regular language that is not recognised by any deterministic $[i,j]$-parity automaton. We will show that the priorities of any automaton $\Mc$ recognising $L$ include all the priorities interval $[2m+i+1,2m+j+1]$ for some integer $m$, which proves \cref{theorem:pih}. We will do this in the following two steps. 

First, in \cref{lemma:SR-and-parity-languages}, we will show that for if an MR automaton $\Nc$ recognises the \emph{$[i+1,j+1]$-parity language $L_{[i+1,j+1]}$}, which is the set of words in $[i+1,j+1]^{\omega}$ such that the highest priority occurring infinitely often is even, then the priorities of $\Nc$ include all the priorities in the interval $[2m+i+1,2m+j+1]$ for some $m$. 

Then, in \cref{lemma:index-hierarchy-reduction-to-parity-language}, we show that if $\Mc$ is an MR automaton recognising $L$---which, recall, is an $\omega$-regular language that is not recognised by any deterministic $[i,j]$-parity automaton---then there is an automaton $\Nc$ recognising the $[i+1,j+1]$-parity language whose priorities are a subset of the priorities of $\Mc$. 

Combining these two results proves \cref{theorem:pih}. We show these two lemmas next. 

\begin{lemma}\label{lemma:SR-and-parity-languages}
   For every two natural numbers $\alpha,\beta$ with $\alpha<\beta$, any MR automaton $\Nc$ recognising the $[\alpha,\beta]$-parity language must contain all the priorities in $[2m+\alpha,2m+\beta]$ for some integer $m$.
\end{lemma}
\begin{proof}
    Let $\Nc$ be an MR automaton that recognises the $[\alpha,\beta]$-parity language $L_{[\alpha,\beta]}$. Let $\Mc$ be a memoryless almost-sure resolver for $\Nc$, and $\Pc = \Mc \circ \Nc$ be the probabilistic automaton that is the resolver product of $\Mc$ and $\Nc$. We assume, without loss of generality, that $\Mc$ chooses every transition of $\Nc$ on some word with positive probability. 
    
 We start by noting that $\Nc$ contains a priority of the same parity as $\alpha$. Indeed, consider the word $\alpha^{\omega}$. If $\alpha$ is odd (resp. even), then the word $\alpha^{\omega}$ is rejected (resp. accepted) by $\Nc$, and hence $\Nc$ must contain an odd (resp. even) priority, and thus, so must $\Pc$. Thus, let $\eta$ be the lowest priority in $\Pc$ that has the same parity as $\alpha$. We shall induct on the following statement, which implies the proof of \cref{lemma:SR-and-parity-languages}.

    \begin{claim}[IH(k)]\label{claim:Ihk}
        For each natural number $k$ with $\alpha+k\leq \beta$, there is a finite word $u_k \in ([\alpha,\alpha+k])^{*}$ and a positive real $0<\theta_k\leq 1$, such that from all states $q$ in $\Pc$, a run from $q$ on $u_k$ in $\Pc$, with probability at least $\theta_k$, contains a priority that is at least $\eta+k$ and that has the same parity as $\alpha+k$.
    \end{claim} 
    
    We prove the claim by induction. For the base case where $k=0$, we distinguish between the cases of when $\alpha$ is odd and when $\alpha$ is even.
    
    \paragraph*{Base case: $\alpha$ is even} Consider the word $\alpha^{\omega}$, which is in $L_{[\alpha,\beta]}$ and thus accepted by $\Pc$ almost-surely, i.e., a random run of $\Pc$ on $\alpha^{\omega}$ is accepting with probability $1$. We assume, without loss of generality, that every state $q$ in $\Pc$ can be reached from the initial state of $\Pc$. Then the automaton $(\Pc,q)$, must accept the word $\alpha^{\omega}$ almost-surely. This implies that for each state $q$ there is be an $n_q \in \mathbb{N}$, such that there is a path from $q$ on $\alpha^{n_q}$ that sees an even priority, which by minimality of $\eta$ (recall that $\eta$ is the smallest priority in $\Pc$ that has the same parity as $\alpha$), is at least $\eta$. Thus, there is a positive probability $\theta_q$ such that a random run of $\alpha^{n_q}$ on $q$ sees an even priority at least $\eta$. Taking $u_0$ to be $\alpha^{n_0}$, where  $n_0$ is the maximum of $n_q$'s for all states $q$, we get that a random run from each state $q$ on $u_0$ sees an even priority at least $\eta$ with probability at least $\theta$, where $\theta =\min \{\theta_q \mid q \in Q\}$.

    \paragraph*{Base case: $\alpha$ is odd} Consider the word $\alpha^{\omega} \notin L_{[\alpha,\beta]}$ that is rejecting in $\Pc$ in $\Nc$. Since $\Lc(\Pc) = \Lc(\Nc)$, all possible runs of $\Pc$ on $\alpha^{\omega}$ is rejecting. Therefore, all runs from each state $q$ on $\alpha^{\omega}$ in $\Pc$ must contain at least one odd priority. This implies that for each state $q$, there is a finite word $u_q=\alpha^{n_q}$ on which there is a run from $q$ in $\Pc$ that contains an odd priority. Thus, there is a positive probability $\theta_q$ of seeing an odd priority at least as large as $\eta$ from each state $q$ on $\alpha^{n_q}$. Taking $u_0$ to be $\alpha^{n_0}$, where  $n_0$ is the maximum of $n_q$'s for all states $q$, we get that a random run from each state $q$ on $u_0$ sees an odd priority at least $\eta$ with probability that is at least the minimum of $\theta_q$'s for all states $q$. 

    This completes our base case for when $k=0$. 
    
    Let us now prove the induction step. Assume IH($k$) holds for some natural number $k<\beta-\alpha$. That is, there is a finite word $u_k \in [\alpha,\alpha+k]^{*}$ and a positive real $\theta_k$, such that from any state $q$, the probability that a run from $q$ on $u_k$ sees a priority that is at least $\eta+k$ and of the same parity as $\eta+k$ is at least $\theta_k$. We now show that this implies IH($(k+1)$).

    We once again distinguish between the cases $(\alpha+k+1)$ is even or odd.
    
    \paragraph*{Induction step: if $(\alpha+k+1)$ is even} 
    
    Fix a state $q$ and consider the finite word $v=u_k \cdot (\alpha+k+1)$. Since $u_k$ is a prefix of $v$, from every state $p$, the probability that an odd priority at least $\eta+k$ occurs in a run from $p$ on $v$ in $\Pc$ is at least $\theta_k$. The word $v^{\omega}$ is accepted almost surely from $q$ in $\Pc$, however. We shall utilise the above two remarks to show that an even priority at least $(\eta+k+1)$ occurs with positive probability on a run from $q$ on the word $v^n$ from some $n$ in $\Pc$.
    
    More concretely, let $\zeta^q_n$ denote the probability that a run of $\Pc$ from $q$ on the word $v^n$ contains an even priority that is at least $\eta+(k+1)$. The sequence $\zeta^q_0,\zeta^q_1,\zeta^q_2,\cdots$ then is nondecreasing, and since it is bounded above by 1, converges to a value $\zeta^q$. Below, we show that $\zeta^q=1$, which implies that $\zeta^q_n>0$ for some finite~$n$. This follows from the computations below.

    If $\rho$ is a run of $\Pc$ from $q$ on the word  $v^{\omega}$, then 
    \begin{equation*}
    \begin{split}
        \prob[\rho \text{ is accepting}]
        &= \prob[\max(\inf(\rho)) \text{ is even}] \\
        &= \prob[\max(\inf(\rho)) \text{ is even} \\ 
        &\quad\quad\qquad\text{ and at most $(\eta+k-1)$}] \\
        &+ \prob[\max(\inf(\rho)) \text{ is even} \\
        &\quad\quad\qquad\text{ and at least $(\eta+k+1)$}] \\
    \end{split} 
    \end{equation*}
    We first show that that $$\prob[\max(\inf(\rho)) \text{ is even and at most $(\eta+k-1)$}]=0.$$ To do so, we use the fact that from any state in $\Pc$, the probability that an odd priority at least $(\eta+k)$ occurs on a run from $\Pc$ on $v$ is at least $\theta_k$. Thus the probability that a priority at most $(\eta+k-1)$ occurs on a run in $\Pc$ from any state on $v$ is at most $(1-\theta_k)$.

\begin{equation*}
    \begin{split}
        &\prob[\max(\inf(\rho)) \text{ is even and at most $(\eta+k-1)$}] \\
        &\leq  \prob[\max(\inf(\rho)) \text{ is at most $(\eta+k-1)$}]  \\ 
        &=\prob[\bigcup_{N \in \mathbb{N}} \text{(No priority that is at least $\eta+k$ }\\
        &\quad\quad\quad\quad\text{occurs after reading the $N^{th}$ $v$ in $\rho$)} ] \\
        &=\lim_{N\to\infty}\prob[\text{(No priority that is at least  $\eta+k$ }\\
        &\quad\quad\quad\quad\quad\quad\text{occurs after reading the $N^{th}$ $v$ in $\rho$)}] \\
        &\leq \lim_{N\to\infty} \prod_{n\geq N} (1-\theta_k) = 0
    \end{split}
\end{equation*}
Thus, we have, 
\begin{equation*}
    \begin{split}
        \prob[\rho \text{ is accepting}]
        &= \prob[\max(\inf(\rho)) \text{ is even }\\
        &\quad\quad\quad\quad\quad\text{and at most $(\eta+k-1)$}] \\
        &+ \prob[\max(\inf(\rho)) \text{ is even }
        \\&\quad\quad\quad\quad\quad\text{and at least $(\eta+k+1)$}] \\
        &= \prob[\max(\inf(\rho)) \text{ is even }\\
        &\quad\quad\quad\quad\quad\text{and at least $(\eta+k+1)$}] \\
        &\leq \prob[\rho \text{ contains a transition with an} 
        \\ &\qquad\text{even priority at least $(\eta+k+1)$}] \\
        &=\lim_{n\to \infty} \zeta^q_n = \zeta^q\\
    \end{split} 
\end{equation*}
    
Since $v^{\omega}$ is almost-surely accepted, $\rho$ is almost-surely an accepting run, and hence we get $1 \leq \zeta^q$, which gives us $\zeta^q$ is $1$ and hence $\zeta^q_n>0$ for some $n$. Let $n$ be large enough so that $\zeta^p_n>0$ for all states $p$. Then, any run of $\Pc$ on the word $u_{k+1}=v^n$ from any state contains an even priority at least $(\eta+k+1)$ with probability $\theta_{k+1}=\min (\zeta^p_n)_{p \in Q}$, as desired.

\paragraph*{Induction step: if $(\alpha+k+1)$ is odd} 

Consider the finite word $v=u_k \cdot (\alpha+k+1)$. Since $u_k$ is a prefix of $v$, from every state $q$, the probability that an even priority at least $\eta+k$ is seen on a run from $q$ on $v$ in $\Pc$ is at least $\theta_k$. In particular, from each state $q$, there is a run from $q$ on $v$ in $\Pc$, such that an even priority at least $\eta+k$ occurs.  

Fix a state $q$, and consider the following run from $q$ on $v^{\omega}$ in $\Pc$ in which the highest even priority occurring infinitely often is at least $(\eta+k)$. There is a run $\rho_1$ in $\Nc$ from $q$ to some state $q_1$ on $v$ such that the highest even priority occurring in $\rho_1$ is at least $(\eta+k)$. Similarly, there is a run $\rho_2$ from $q_1$ to $q_2$ such that the highest even priority occurring in $\rho_1$ is at least $(\eta+k)$. Extending this to get $\rho_3,\rho_4,\dots$ similarly, and then concatenating $\rho_1 \cdot \rho_2 \cdot \rho_3 \dots$ to get $\rho$, we see that $\rho$ is a run from $q$ on $v^{\omega}$ in which the highest even priority occurring infinitely often is at least $(\eta+k)$. Since the word $v^{\omega}$ is rejected by $\Nc$, this means that there is an odd priority at least $(\eta+k+1)$ that occurs in $\rho$ infinitely often. Let $n_q$ be the minimum natural number so that there is a run from $q$ on $v^{n_q}$ that contains a transition of an odd priority at least $(\eta+k+1)$. Thus, the probability that a run of $\Pc$ from $q$ on $v^{n_q}$ that contains a transition of an odd priority at least $(\eta+k+1)$ is positive, say $\theta_q$. Taking $u_k$ to be $v^{n}$, where  $n$ is the maximum of $n_q$'s for all states $q$, we get that a random run from each state $q$ on $u_k$ sees an odd priority at least $(\eta+k+1)$ with probability that is at least the minimum of $\theta_q$'s for all states $q$, as desired.  

This completes our proof of \cref{claim:Ihk}, and also of \cref{lemma:SR-and-parity-languages}.
\end{proof}

A consequence of the lemma above is that an MR $[i,j]$-parity automata cannot recognise the $[i+1,j+1]$-parity language.  The next lemma allows us to reduce  \cref{theorem:pih} to parity languages, from where the conclusion follows.

\begin{lemma}\label{lemma:index-hierarchy-reduction-to-parity-language}
    Let $L$ be an $\omega$-regular language, such that $L$ cannot be recognised by any deterministic $[i,j]$-parity automata. If automaton $\Ac$ is an MR automaton recognising $L$, then there is an MR automaton $\Nc$ recognising the $[i+1,j+1]$-parity language  whose priorities are a subset of the priorities of $\Ac$. 
\end{lemma}
\begin{proof}
     Let $\Dc$ be a deterministic parity automaton that recognises $L$, and $\Pc$ be a probabilistic parity automaton that is obtained by taking the resolver-product of $\Ac$ with a memoryless resolver that ensures all accepted words in $\Dc$ are accepted almost-surely in $\Pc$.
     
     Since $L$ cannot be recognised by a deterministic $[i,j]$ automaton, we know due to the flower lemma of Niwi\'nski and Walukiewicz that $\Dc$ contains an $[i+1,j+1]$-\emph{flower}~\cite[Lemma 14]{NW98}. That is, there exists a state $p$ in $\Dc$, an integer $\ell$, and finite words $u_{k}$ for each $k$ in the interval $[i+1,j+1]$, such that the unique run from $p$ on $u_k$ ends at $p$ and the highest priority seen during this run is $2\ell+k$. Let $v$ be a finite word such that there is a run from the initial state of $\Dc$ to $p$ on the word $v$.

     We build a nondeterministic automaton $\Nc$ using the automaton $\Ac$, and the finite words $v,u_{i+1},u_{i+2},\cdots,$ $u_{j+1}$. The states of the automaton $\Nc$ are same as the states of $\Pc$. The initial state of $\Nc$ is a state that can be reached in $\Pc$ from its initial state on the word $v$ with positive probability. The transitions of $\Nc$ are over the alphabet $[i+1,j+1]$, and we have a transition from $q$ to $q'$ on the letter $k$ with priority $\pi$ if there is a run from $q$ to $q'$ in $\Pc$ on the word $u_k$, where the largest priority that run contains is $\pi$. 
     
     We construct an almost sure resolver $\Pc_\Nc$ for $\Nc$, which, by construction is a memoryless resolver, utilising the probabilistic automaton $\Pc$. Concretely, the almost-sure resolver $\Pc_\Nc$, for a state $q$ and letter $k\in [i+1,j+1]$, chooses the transition to $q'$ of priority $\pi$ with probability $\zeta$, if the probability of reaching $q'$ from $q$ on the word $u_k$ such that the highest priority occurring in the run is $\pi$ is $\zeta$. Abusing the notation slightly, we let $\Pc_\Nc$ denote the probabilistic automaton that is obtained by taking the resolver product of $\Pc_\Nc$ with~$\Nc$.

     We will show that $\Nc$ recognises the $[i+1,j+1]$-parity language, and $\Pc_\Nc$ recognises all words in the $[i+1,j+1]$-parity language with probability $1$. This follows from the following chain of equivalences. 
     
     The word $w'=k_0 k_1 k_2 \cdots$ with $k_i \in [i+1,j+1]$ is in the $[i+1,j+1]$-parity language \emph{if~and~only~if}  $w=u_{k_0}\cdot u_{k_1} \cdot u_{k_2} \cdots $ is accepted by $(\Dc,p)$ \emph{if and only if} the word $v.w$ is accepted by $\Ac$ and accepted almost surely by $\Pc$ \emph{if and only if} the word $w'=k_0 k_1 k_2 \cdots$ is accepted by $\Nc$ and almost surely by $\Pc_\Nc$. This concludes our proof. 
\end{proof}

We note that \cref{lemma:index-hierarchy-reduction-to-parity-language,lemma:SR-and-parity-languages} together prove \cref{theorem:pih}. 
\section{Appendix for \cref{sec:complexity}}
\subsection{Recognising Memoryless adversarially resolvable automata}\label{appendixsubsec:np-completeness}

\paragraph*{Token games}

For a parity automaton $\Ac$, the 2-token game on $\Ac$ is defined similarly to the HD game on $\Ac$ (\cref{defn:hd-game}), with Adam building a word letter-by-letter and Eve building a run on her token on Adam's word transition-by-transition, but additionally, Adam is also building two runs on two of his tokens transition-by-transition. The winning condition for Eve is that if any of the runs of Adam's tokens are accepting, then the run on her token is also accepting.

\begin{definition}\label{defn:twotokengame}
Given a nondeterministic parity automaton $\Ac=(Q,\Sigma,\Delta,q_0)$, the $2$-token game on $\Ac$ is a two-player non-stochastic complete-observation game between Eve and Adam that starts with an Eve's token and Adam's two distinguishable tokens at $q_0$, and proceeds in infinitely many rounds. In round $i$ when Eve's token is at $q_i$ and Adam's tokens are at $p^1_i$ and $p^2_i$:
\begin{enumerate}
    \item Adam selects a letter $a_i\in\Sigma$;
    \item Eve selects a transition $q_i\xrightarrow{a_i:c_i}q_{i+1}$ and moves her token to $q_{i+1}$;
    \item Adam selects transitions $p^1_i\xrightarrow{a_i:c^1_i}p^1_{i+1}$ and $p^2_i\xrightarrow{a_i:c^2_i}p^2_{i+1}$ on each of his tokens and moves his first and second token to $p^1_{i+1}$ and $p^2_{i+1}$, respectively.
\end{enumerate}
In the limit of a play of the $2$-token game, Eve constructs a run on her token, and Adam constructs a run on each of his two tokens, all on the same word. We say that Eve wins such a play if the run on her token is accepting or both the runs on Adam's tokens are rejecting. We say Eve \emph{wins the 2-token game} on $\Ac$ if she has a strategy to win the 2-token game.
\end{definition}

The rest of the section of the appendix will focus on the proof of \cref{lemma:ZlkDAGMDP}, which is restated below.
\ZlkDAGMDP*
\begin{proof}
        We begin by stating a lemma proved by Gimbert and Horn~\cite{GH10} for any prefix-independent objective---a property that Muller objectives satisfy.
        \begin{theorem}[\!\!\protect{\cite[Theorem~3.2]{GH10}}]\label{lemma:lastminute}\label{lemma:tailobjectives}
            If a player, say Adam, wins with a positive probability from some vertex of a complete-information two-player stochastic games, then he wins almost-surely from at least one vertex.
        \end{theorem}
        The following is a corollary of \cref{lemma:lastminute}, described in the work of Gimbert, Oualhadj, and Paul. 
        \begin{corollary}[\!\!\protect{\cite[Corollary~1]{GOP11}}]\label{cor:tailobjectives}
            If the set of vertices from which Adam wins almost-surely in an MDP with prefix-independent objectives can be computed in polynomial time, then there is also a polynomial time algorithm to compute the set of vertices from which Adam wins with positive probability. 
        \end{corollary}
       We note that the additional time taken to compute the vertices that wins with positive probability is, in fact, linear. 
        Because of \cref{cor:tailobjectives}, we will focus only on providing an $\Oh(|\Mc||\Zc_{C,\Fc}|)$-time algorithm to compute the almost-sure winning region for Adam. 
        
        Consider an MDP $\Mc$, with a Muller objective that is specified using a colouring function  $\col \colon E \rightarrow C$ over the edges of the MDP. 
        Let the Zielonka DAG representing the winning condition be $\Zc_{C,\Fc}$. Recall that vertices of the DAG are elements of the subset of colours, that is, if $\mu_X\in V(\Zc_{C,\Fc})$, then $X\subseteq C$. We can assume that the DAG $\Zc_{C,\Fc}$ has exactly one source node, $\nu_C$. The source node of a Zielonka DAG is either labelled by an accepting subset, that is, a subset  in $\Fc$, or by a rejecting subset, that is not in $\Fc$. 
        \paragraph*{Maximal End Components} An end component
        of an MDP $\Mc = (V_A,V_R, E,\delta)$ is a subset of the vertices $V$ of the arena such that 
        \begin{itemize}
            \item the subgraph induced by $M$ in the graph $(V,E)$ is a strongly connected component, and
            \item for every $v\in V_R$ and for edges $v\rightarrow w\in E$, we have $w\in M$.
        \end{itemize}
        A Maximal End Component (MEC) is an end component that is maximal with respect to inclusion order. 
        
        Every vertex of an MDP belongs to at most one MEC. From every vertex in the MEC $M$, Adam can ensure that all vertices in $M$ are visited infinitely often, and the play starting from $M$ stays within $M$ with probability $1$. A MEC can be computed in polynomial time and even sub-quadratic time~\cite{CH11}.
        
        Consider the subroutine that takes as input an MDP $\Mc$ with a colouring function $\col$ from vertices $V$ to a finite set of colours $C$, along with a Zielonka DAG $\Zc_{C,\Fc}$ that uniquely represents the winning condition. 

        For a subset $X$ of colours, we write $\Mc|_X$ to denote the sub-MDP of $\Mc$ restricted to vertices in $\Mc$ from which the player Adam can ensure that all colours seen are within $X$ with probability $1$.  
        For an end-component $M$, we write $\col(M)$ to represent the set of all edge colours for all edges between two vertices in $M$. Therefore, $\col(M) = \{c\mid c= \col(u,v)\text{ and }u,v\in M\}$

        We refer to the reverse topological order of the nodes in a DAG to refer to the ordering $\nu_{X_1},\nu_{X_2},\dots,\nu_{X_d}$ of the nodes such that there is no edge $\nu_{X_i}\rightarrow\nu_{X_j}$, where $j<i$, that is, no node has an edge to a node occurring earlier than this in the order. 
        \begin{algorithm}
            \begin{algorithmic}\caption{A bottom up Zielonka-DAG algorithm to compute almost-sure winning sets in MDPs with Muller Objective}\label{algo:AlmostSureZielonkaDAG}
                \Procedure{AlmSureReach}{$\Mc, T$}
                    \State \Return{set of vertices in $\Mc$ from which player can almost surely reach $T$}
                \EndProcedure
                \Procedure{AlmostSureMuller}{$\Mc,\Dc$}
                        \For{nodes $\nu_X$ in reverse topological order in $\Zc_{C,\Fc}$}
                            \If{$X$ is an accepting set of colours}
                                \For{ $M_i\in M_1,\dots,M_k$, where each $M_i$ a maximal end components of $\Mc|_X$}
                                    \If{$\col(M_i)\cap \Bar{Y_j} \neq \emptyset$ for all $Y_j$ such that $\nu_X\rightarrow\nu_{Y_j}$}
                                    \State $\Win_{\nu_X}\gets\Win\cup M_i$
                                    \State $\Win\gets\textsc{AlmSureReach}_\Mc(\Win_{\nu_X})$
                                    \EndIf
                                \EndFor
                            \EndIf
                        \EndFor
                    \Return $\Win$
                \EndProcedure
            \end{algorithmic}
            \end{algorithm}
    \paragraph*{Correctness}
    We show that the procedure $\textsc{AlmostSureMuller}(\Mc,\Dc)$ returns exactly the set of vertices of the MDP from which Adam can satisfy the Muller objective with probability~$1$. 
    Let $W$ represent the set of vertices from which there is a strategy for Adam to ensure the Muller objective is satisfied with probability $1$. We show that $W = \Win$, where $\Win$ is the set returned by \cref{algo:AlmostSureZielonkaDAG}.  
    
    $W\subseteq\Win$. 
    We show that any vertex that is almost-surely winning is added to $\Win$ by \cref{algo:AlmostSureZielonkaDAG}. We say that an end components $M$ of the MDP is \emph{$\Fc$-winning} if $\col(M)$ is an accepting Muller set. 
    We use the following lemma to prove $W\subseteq\Win$.
    \begin{lemma}[\!\!\protect{\cite[Lemma~9]{Cha07}}]\label{lemma:Cha07Union}
        The vertices from which the player can achieve the Muller objectives with probability $1$ (the set of vertices $W$) are exactly the set of vertices that  can reach the union of all $\Fc$-winning end components with probability $1$.
    \end{lemma}
    
    From \cref{lemma:Cha07Union}, it is enough to show that for any end component $M$ such that all vertices almost-surely satisfies the Muller objective, $M\subseteq \Win$. Indeed, since $\Win$ that is returned by the algorithm is closed under almost-sure attractors, if all such $\Fc$-winning end components are in $\Win$, all almost-sure winning vertices are in $\Win$.
    
    Consider an end component $\Mc$ that is $\Fc$-winning. Then there is a node $\nu_X$ in the DAG that is the earliest (in reverse topological order) node such that $X$ is an accepting subset of colours and $X\supseteq \col(M)$ and $\col(M)\nsubseteq Y_j$ for any $\nu_{Y_j}$ that has an edge from $\nu_X$. Such a $\nu_X$ exists because $\col(M)$ is an accepting set, and since $\col(E)$ is the root of the DAG. The set would just be the smallest superset of $\col(M)$ among the labels of the nodes of the DAG. 
    
    In the sequel, we show the stronger statement that $\Fc$-winning end component $M$ satisfies $M\subseteq \Win_{\nu_X}$. 
    In $\Mc|_X$, the $\Fc$-winning $M$ is also an end component, since it contains winning vertices in $X$. Moreover, this end component $M$ is contained in some maximal end component $M_i$. Since $\col(M_i)\supseteq \col(M)$, and we assumed $\nu_X$ was the earliest node, it follows from the structure of the DAG that $X$ is a minimal subset of accepting colours containing $\col(M)$. Therefore, it must also be the case that $\col(M)\nsubseteq Y_i$ for any outgoing $\nu_{Y_i}$, such that $\col(M)\nsubseteq Y_i$. 
    Therefore, $M_i$, which is a superset of $M$, is added to $\Win_{\nu_X}$, and therefore $\Win$.
    

    \textit{$\Win\subseteq W$.} We provide a winning strategy for every vertex in $W$ that almost-surely achieves the Muller objective.  Since $\Win$ is monotonically increasing, that is, vertices are added and not removed, we consider the first time that a vertex $v$ was added.

    We argue that the algorithm maintains the invariance that from any vertex in $\Win$, there is a strategy such that the Muller objective is satisfied with probability $1$ using a strategy that only visits the vertices in $\Win$.  Trivially this is true when $\Win = \emptyset$. 

    When a vertex $v$ is added to the set, it is because either $v\in M_i$ for some MEC of $\Mc|_X$ or because $v$ is almost surely reachable to $\Win_{\nu_X}$.
    If it was the former, then consider the strategy of the player that visits all vertices in $M_i$ infinitely often, while avoiding all vertices whose colour is not in $X$.  Such a strategy exists since $\Mc|_X$ is defined as the sub MDP from which the player could avoid any colour not in $X$, and further $M_i$ is a maximal end component ensuring that by randomising between all the edges in $M_i$, the set of colours seen infinitely often is exactly $\col(M_i)$ which is a winning subset of colours. 
    If the vertex $v$ was added because it had an almost sure reachability strategy to $\Win$, then from $v$, the player follows that strategy to reach the set $\Win_{\nu_X}$, and later follows the existing winning strategy for vertices defined for the vertices $\Win_{\nu_X}\subseteq \Win$. 
    \paragraph*{Runtime} The algorithm runs in time that is polynomial in the size $m$ of the MDP and the size $d$ of the DAG $\Zc_{C,\Fc}$. More precisely, we show that the runtime is at most $\Oh( m\cdot d)$.
    This is because topologically sorting the nodes of the DAG $\Dc$ takes time at most $\Oh(d)$ using Tarjan's algorithm~\cite{Tar72}. 
    Later, for each node $\nu_X$, the algorithm computes an MDP that avoids colours not in $X$, and then computes the maximal end components of such MDPs. This step takes time that is linear in the size of the MDP~\cite{CH11}. Finally, computing the almost sure reachable set of vertices can also be done in linear time in the number of states in the MDP. Therefore, the time taken is at most $\Oh(m\cdot d)$
\end{proof}
\MAcheckNPhard*
\begin{proof}
    We will use the proof by Prakash that showed NP-hardness for the problem of deciding history-determinism~\cite{Pra24a}. There, the reduction is from the problem of solving \emph{good 2-D parity games}, which we define first. 

    A \emph{2-dimensional parity game} $\Gc$ is a complete-observation non-stochastic game, where each edge $e = (u,v)$ in $\Gc$ is labelled with two priorities $(f_e,s_e)$. We call $f_e$ and $s_e$ as the first and second priority of $e$, respectively and write the edge $e$ as $e=u\xrightarrow{f_e,s_e}v$. We say that an infinite play in $\Gc$ satisfies the first (resp.\ second) parity condition if the largest first (resp.\ second) priority occurring infinitely often amongst the edges of that play is even. An infinite play $\rho$ in $\Gc$ is winning for Eve if the following condition holds: \emph{if $\rho$ satisfies the first parity condition, then $\rho$ also satisfies the second parity condition}. Otherwise, Adam wins that play. We say Eve wins $\Gc$ if she has a sure winning strategy, i.e., any play constructed when Eve is playing according to that strategy is winning.

    We say that a 2-dimensional parity game $\Gc$ is \emph{good} if all plays that satisfy the first parity condition also satisfy the second parity condition. The problem of deciding if Eve wins good 2-D parity games is $\NP$-complete~\cite[Lemma 13]{Pra24a}.  

    We reduce the problem of deciding if Eve wins a good 2-D parity game to deciding if a parity automaton is MA, by showing that the reduction of \cite{Pra24a}, starting from a good 2-D parity game $\Gc$, constructs $\Hc$, such that $\Hc$ contains a deterministic subautomaton $\Dc$ and:
    \begin{enumerate}
        \item if Eve wins $\Gc$, then $\Hc$ simulates $\Dc$ and $\Hc$ is determinisable-by-pruning, and therefore, $\Hc$ is MA.
        \item if Adam wins $\Gc$, then $\Hc$ does not simulate $\Dc$, $\Hc$ is not HD, and hence $\Hc$ is not MA.
    \end{enumerate}
    \paragraph*{The reduction}
    Let $\Gc$ be a 2-D parity game. We next describe the construction of $\Hc$ as above. The automata $\Hc$ will contain a deterministic subautomaton $\Dc$ and is over the alphabet $\Sigma = E\cup \$$. Informally, the priorities of transitions of $\Dc$ and $\Hc$ will capture the first and second priorities of edges in $\Gc$, respectively. 
    
    For each Adam vertex $v$, we have the states $v_D$ in $\Dc$ and $v_H$ in $\Hc$. For each edge $e = v\xrightarrow{f_e:s_e} u$ outgoing from $v$, we have the transitions $v_D \xrightarrow{e:f_e} u_{\$}$ and $v_H \xrightarrow{e:s_e} u_H$ on the letter $e$ (see \cref{fig:Adam-vertex}).
    


    Adam choosing an outgoing edge from $v$ in $\Gc$ is captured by a round of the simulation game where Adam's token is in $v_D$ and Eve's token is in $v_H$. This way, Adam, by choice of his letter in the simulation game, captures the choice of outgoing edges from his vertex. 
       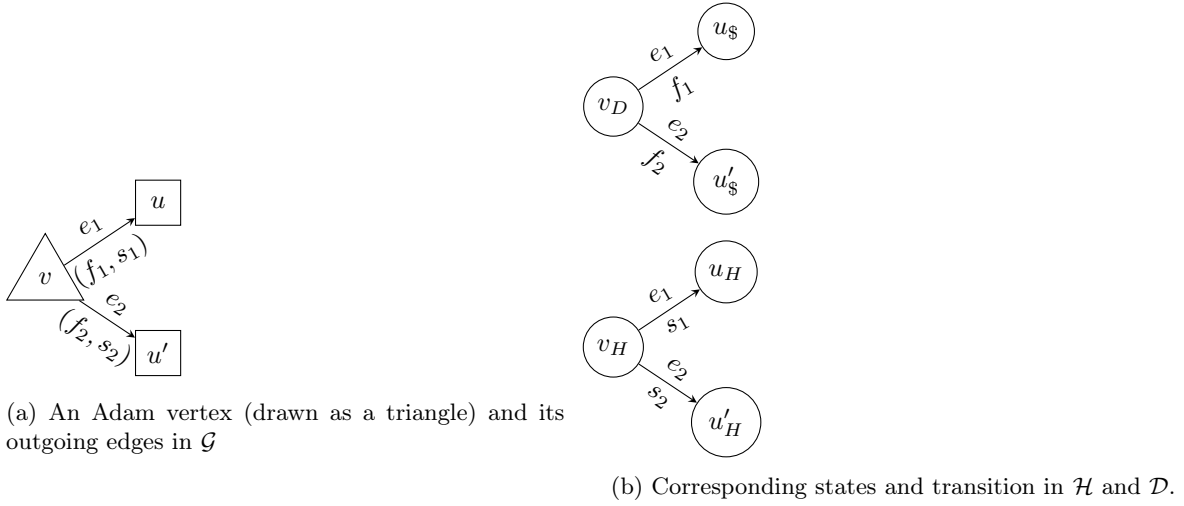
\begin{figure}
           \centering
    \begin{subfigure}[b]{0.45\linewidth}
    \begin{tikzpicture}
        \usetikzlibrary{graphs,quotes}
        \usetikzlibrary{backgrounds,scopes,fit,shapes.geometric}
         \node (v) [regular polygon, regular polygon sides = 3, minimum size = 0.25 cm, draw]    at (0,0)   {$v$};
         \node (u1) [rectangle, minimum size = 0.6 cm, draw] at (1.5,1) {$u$};
         \node (u2) [rectangle, minimum size = 0.6 cm, draw] at (1.5,-1) {$u'$};
         \path[-stealth]
         (v) edge node [above, rotate = 30] {$e_1$} node [below, rotate = 30] {$(f_1,s_1)$} (u1)
         (v) edge node [above, rotate = 327] {$e_2$} node [below, rotate = 327] {$(f_2,s_2)$} (u2)
         ;
    \end{tikzpicture}
    \caption{An Adam vertex (drawn as a triangle) and its outgoing edges in $\Gc$}\label{fig:Adam-vertexa}
    \end{subfigure}~
    \begin{subfigure}[t]{0.5\linewidth}
    \begin{tikzpicture}
        \usetikzlibrary{graphs,quotes}
        \usetikzlibrary{backgrounds,scopes,fit,shapes.geometric}    
        \node (vD) [circle, minimum size = 0.5 cm, draw] at (0,0) {$v_D$};
         \node (u1dollar) [circle, minimum size = 0.5 cm, draw] at (1.5,1) {$u_\$$};
         \node (u2dollar) [circle, minimum size = 0.5 cm, draw] at (1.5,-1) {$u'_{\$}$};
    
         \node (vH) [circle, minimum size = 0.5 cm, draw] at (0,-3.2) {$v_H$};
         \node (u1H) [circle, minimum size = 0.5cm, draw] at (1.5,-2.2) {$u_H$};
         \node (u2H) [circle, minimum size = 0.5cm, draw] at (1.5,-4.2) {$u'_H$};

         \path[-stealth]
         (vD) edge node [above, rotate = 30] {$e_1$}  node [below, rotate = 30] {$f_1$} (u1dollar)
         (vD) edge node [above, rotate = 325] {$e_2$} node [below, rotate = 325] {$f_2$} (u2dollar)
         (vH) edge node [above, rotate = 30] {$e_1$} node [below, rotate = 30] {$s_1$} (u1H)
         (vH) edge node [above, rotate = 325] {$e_2$} node [below, rotate = 325] {$s_2$} (u2H)
        ;
    \end{tikzpicture}
    \caption{Corresponding states and transition in $\Hc$ and $\Dc$.}\label{fig:Adam-vertexb}
    \end{subfigure}
    \caption{Transitions corresponding to Adam's vertices. The letters are displayed on top, and the priorities are below each edge.}\label{fig:Adam-vertex}
    \end{figure}
    For each Eve vertex $v$ and outgoing edge $e = v\xrightarrow{f,s}u$ from that vertex, we will have the states $v_{\$}$, $v_D$, and $u_D$ in $\Dc$, and the states $v_H$, $(v_H,e)$, and $(u_H)$ in $\Hc$. In $\Dc$, we will have the transitions $v_\$ \xrightarrow{\$:0} v_D, v_D \xrightarrow{e:f}: (u_D)$, and in $\Hc$, we will have the transitions $v_H \xrightarrow{\$:0} (v_H,e), (v_H,e) \xrightarrow{e:s} u_H$. Additionally, for every outgoing edge $e'=v\xrightarrow{f',s'}u'$ that is different from $e$, we add the transition $(v_H,e) \xrightarrow{e':s'} u'_D$ (see \cref{fig:Eve-vertex}).  
      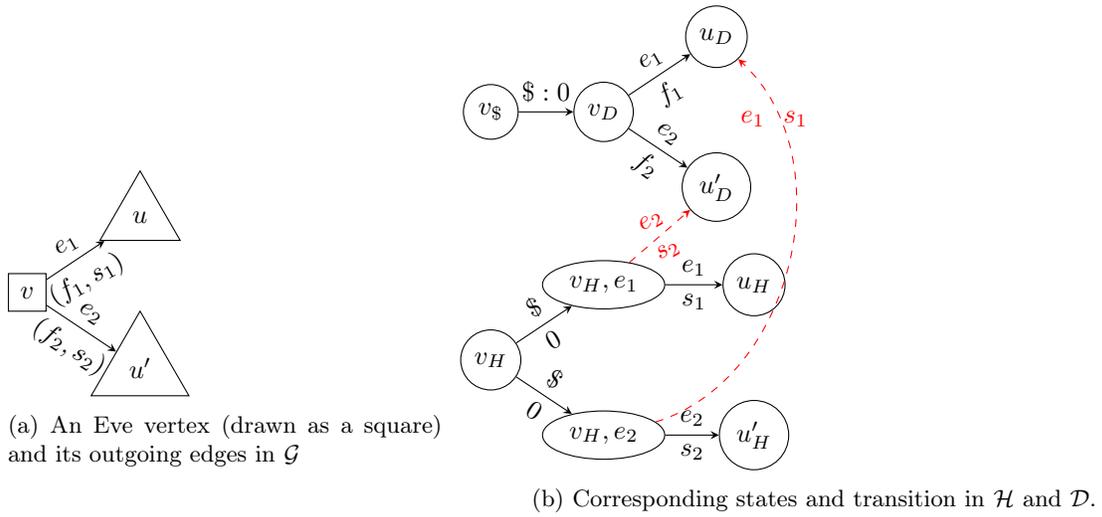
\begin{figure}
      \centering
\begin{subfigure}[b]{0.35\linewidth}
    \begin{tikzpicture}
    \usetikzlibrary{graphs,quotes}
    \usetikzlibrary{backgrounds,scopes,fit,shapes.geometric}
    \node (v) [rectangle, minimum size = 0.5 cm, draw] at (0,0) {$v$};
     \node (u1) [regular polygon, regular polygon sides = 3, minimum size =0.3cm, draw]    at (1.5,1)   {$u$};
     \node (u2) [regular polygon, regular polygon sides = 3, minimum size = 0.3cm, draw] at (1.5,-1) {$u'$};
      \path[-stealth]
     (v) edge node [above, rotate = 30] {$e_1$} node [below, rotate = 30] {$(f_1,s_1)$} (u1)
     (v) edge node [above, rotate = 327] {$e_2$} node [below, rotate = 327] {$(f_2,s_2)$} (u2)
     ;
    \end{tikzpicture}
    \caption{An Eve vertex (drawn as a square) and its outgoing edges in $\Gc$}\label{fig:Eve-vertexa}
\end{subfigure}~
\begin{subfigure}[t]{0.57\linewidth}
    \begin{tikzpicture}
        \usetikzlibrary{graphs,quotes}
    \usetikzlibrary{backgrounds,scopes,fit,shapes.geometric}
     \node (vdollar) [circle, minimum size = 0.5 cm, draw] at (0,0) {$v_\$$};
     \node (vD) [circle, minimum size = 0.5 cm, draw] at (1.5,0) {$v_D$};
     \node (u1D) [circle, minimum size = 0.5 cm, draw] at (3,1) {$u_D$};
    \node (u2D) [circle, minimum size = 0.5 cm, draw] at (3,-1) {$u'_D$};     

     \node (vH) [circle, minimum size = 0.5 cm, draw] at (0,-3.3) {$v_H$};
     \node (e1) [ellipse, minimum size = 0.5cm, draw] at (1.5,-2.3) {$v_H,e_1$};
     \node (e2) [ellipse, minimum size = 0.5cm, draw] at (1.5,-4.3) {$v_H,e_2$};
     \node (u1H) [circle, minimum size = 0.5 cm, draw] at (3.5,-2.3) {$u_H$};
     \node (u2H) [circle, minimum size = 0.5 cm, draw] at (3.5,-4.3) {$u'_H$};
    
     \path[-stealth]
     (vdollar) edge node [above] {$\$:0$} (vD)
     (vH) edge node [above, rotate = 30] {$\$$} node [below, rotate = 30] {$0$} (e1)
    (vH) edge node [above, rotate = 327] {$\$$} node [below, rotate = 327] {$0$} (e2) 
     
     (vD) edge node [above, rotate = 30] {$e_1$}  node [below, rotate = 30] {$f_1$} (u1D)
     (vD) edge node [above, rotate = 325] {$e_2$} node [below, rotate = 325] {$f_2$} (u2D)
    
    (e1) edge node [above] {$e_1$} node [below] {$s_1$} (u1H)
     (e2) edge node [above] {$e_2$} node [below] {$s_2$} (u2H)

    (e1) edge [red, dashed] node [above, rotate = 30] {$e_2$} node [below, rotate = 30] {$s_2$} (u2D)
    (e2) edge [red, dashed, bend right = 60] node [left,yshift=19mm,xshift=-2mm] {$e_1$} node [right,yshift=19mm,xshift=-2mm] {$s_1$} (u1D)
    ;
\end{tikzpicture}
\caption{Corresponding states and transition in $\Hc$ and $\Dc$.}\label{fig:Eve-vertexb}
\end{subfigure}    
\caption{Transitions corresponding to Eve's vertices. The letters are displayed on top, and the priorities are below each edge}\label{fig:Eve-vertex}
\end{figure}

    Eve's choosing of an outgoing edge from $v$ in $\Gc$ is captured by two rounds of the simulation game where at the start, Adam's token is in $v_{\$}$ while Eve's token is in $v_H$. Then, Adam selects the letter $\$$ and his token takes the transition to $v_D$. Then Eve can move her token to any state of the form $(v_H,e)$, where $e$ is an outgoing edge from $v$ in $\Gc$. Thus, Eve, by choice of her transition on $\$$ from $v_H$, captures the choice of outgoing edges available from her vertex. When Adam's token is at the state $v_D$ and Eve's token is at the state $(v_H,e)$, Adam can either replay the edge $e$ as the next letter, from where Adam's token goes to $u_H$ and $v_H$. Or, Adam can choose another edge $e'=(v,u')$ that is outgoing from $v$ as the letter: this causes both Eve's and Adam's token to be in the same state $u'_D$, from where Eve trivially wins the simulation game. We call these moves of Adam as \emph{corrupted moves}, and they cause Eve's token to take the dashed transitions in~\cref{fig:Eve-vertex}. 

    We suppose that the initial state of $\Gc$ is an Adam vertex $v$, and we let the initial state of $\Hc$ and $\Dc$ be $v_H$ and $v_D$, respectively. This concludes our description of $\Hc$ and $\Dc$.

    \paragraph*{Correctness of the reduction}
    Note that any play of the simulation game of $\Dc$ by $\Hc$ where Adam does not take any corrupted moves will correspond to a play of the game $\Gc$. The priorities of Adam's token and Eve's token corresponds to the first and second priorities of the edges in $\Gc$, respectively. Thus, Eve wins $\Gc$ if and only if $\Hc$ simulates $\Dc$. For a more rigorous proof, we refer the reader to \cite{Pra24a}.  

   Note that since $\Gc$ is good, any play satisfies the first parity condition also satisfies the second parity condition: this implies that $\Lc(\Dc) \subseteq \Lc(\Hc)$. Since $\Dc$ is a subautomaton of $\Hc$, we obtain that $\Lc(\Dc) = \Lc(\Hc)$. Thus, $\Hc$ simulates $\Dc$ if and only if $\Hc$ is HD if and only if Eve wins $\Gc$ (\cref{lemma:hd-equiv-detsimulation}). 
   
   If Eve wins a 2-dimensional parity game, then Eve has a positional pure winning strategy~\cite[Page 158]{CHP07}. Thus, if Eve wins $\Gc$, then she has a positional strategy given by $\sigma:V_{\eve} \xrightarrow{} E$. This yields a positional strategy for Eve in the HD game on $\Hc$, where from the vertex $v_H$ for $v \in V_{\eve}$ and on the letter $\$$, she chooses the transition  to $(v_H,(v_H,\sigma(v_H)))$. Thus, $\Hc$ is determinisable-by-pruning and hence MA.

   Otherwise, if Eve loses $\Gc$, then $\Hc$ is not HD and hence not MA. This shows the correctness of the above reduction. Thus, the problem of deciding if an automaton is MA is NP-hard.
\end{proof}

\subsection{Recognising SR and MA automata and verifying resolvers}

\lemmaUndecidablecoBuchi*
\begin{proof}
We reduce from the problem of checking the emptiness of a B\"uchi probabilistic automata under the probabilistic semantics.

\begin{question}[\textsf{PBA-Emptiness}]
    Given a probabilistic B\"uchi automaton, is there a word $w$ on which the probability of a run being accepting is non-zero?
\end{question}

\paragraph*{Construction}
Consider a probabilistic B\"uchi automaton $\Bc = (Q_\Bc,\Sigma,\Delta_\Bc,s)$ over the alphabet $\Sigma$. We will construct a non-deterministic coB\"uchi automaton $\Cc$ similar to the one in \cref{fig:ReductioncoBuchi}. We give a formal definition of $\Cc = (Q_\Cc,\Sigma', \Delta_\Cc,s)$ below where the
\begin{itemize}
    \item alphabet is $\Sigma' = \Sigma\cup \{\$,a,b\}$,
    \item states $Q_\Cc = Q_\Bc\uplus \{s, s_1, s_2, s_{\text{fin}}\}$,
    \item initial state is $s$
    \item transitions are $\Delta_{\Cc} = \Delta_{\Bc}' \uplus\{
    s\xrightarrow{\$: 0}s_1,
    s\xrightarrow{\$: 0}s_2,
    s_1\xrightarrow{a: 0}q_0,
    s_2\xrightarrow{b: 0}q_0, 
    s_1\xrightarrow{b: 0}s_{\text{fin}}, 
    s_2\xrightarrow{b: 0}s_{\text{fin}}\} $ where $\Delta_{\Bc}' = \{q \xrightarrow{a:0}q'\mid q \xrightarrow{a:1}q'\in \Delta_{\Bc} \}\cup \{q \xrightarrow{a:1}q'\mid q \xrightarrow{a:2}q'\in \Delta_{\Bc} \}$.
\end{itemize}
Consider the following memoryless (finite memory) resolver $\Rc$ which outputs with probability $1/2$ one of the two transitions $s\xrightarrow{\$: 0}s_1$ and $s\xrightarrow{\$: 0}s_2$ out of the initial state $s$. 
The memoryless resolver further outputs transitions with the exact probability distributions as the probabilistic automaton $\Bc$ from all the states in $Q_\Bc$. At states $s_1,s_2,s_{\text{fin}}$ from where the transitions are deterministic, the resolver chooses the unique transition available on that letter with probability $1$. 
\begin{figure}
    \centering
    \begin{tikzpicture}[shorten >=1pt, node distance=1.5cm, on grid, auto]
     \tikzstyle{state}=[circle, draw, minimum size=15pt, inner sep=1pt]
  \node[state, initial,initial text=] (q0) {$s$};
  \node[state] (q1) [above right=of q0] {$s_1$};
  \node[state] (q2) [below right=of q0] {$s_2$};
  \node[state] (q3) [right=of q1] {$q_0$};
  \node[state] (q4) [right=of q2] {$s_\textsf{fin}$};

  \path[->]
    (q0) edge node[sloped]  {$\$\colon \frac{1}{2}$} (q1)
    (q0) edge node[sloped,swap]  {$\$\colon \frac{1}{2}$} (q2) 
    (q1) edge node[sloped] {$a$} (q3)
    (q1) edge node[above, xshift=-3mm, yshift = 4mm] {$b$} (q4)
    (q2) edge node[sloped,swap] {$a$} (q4)
    (q2) edge node[below , xshift=-3mm, yshift = -4mm] {$b$} (q3)
    (q4) edge[loop right] node{$\Sigma$} (q4);

    \draw (2,0) rectangle (4.5,1.5);
      \node (q1) at (3.5,0.7) {$\Bar{\mathcal{B}}$};
\end{tikzpicture}
    \caption{Constructing a coB\"uchi $\Cc$ automaton where $\Bar{\Bc}$ is the dual of the probabilistic B\"uchi automaton $\Bc$}
    \label{fig:ReductioncoBuchi}
\end{figure}
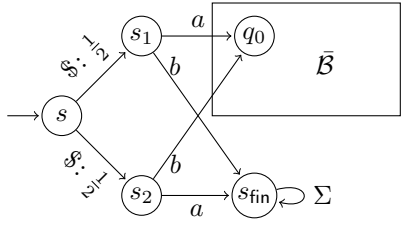


We will now show that the probabilistic automaton $\Bc$ accepts a word  $w\in \Sigma^\omega$ with probability $>0$ if and only if 
$\Rc$ is not an almost-sure resolver. 

\paragraph*{$\implies$} Suppose the PBA $\Bc$ accepts a word $w$ with probability $p>0$. 
Consider the word $ \$ a w$. On this word, the probability of the word being accepted is 
\begin{equation*} 
\begin{split}
&\Pr[\text{reaching }s_1\text{ on } \$]\cdot \Pr[\Rc\text{ resolving } aw \text{ from }s_1] \\
        +  
    &\Pr[\text{reaching }s_2\text{ on } \$] \cdot \Pr[\Rc\text{ resolving } aw \text{ from }s_2] \\
    &= \frac{1}{2}\cdot\Pr[\Rc\text{ resolving } w \text{ from }q_0]\\
    &+ \frac{1}{2}\cdot\Pr[\Rc\text{ resolving } w \text{ from }s_{\text{fin}}]
\end{split}
\end{equation*}
From our construction, the resolver is such that the probability of $w$ having an accepting run from $q_0$ is $1-p$. Therefore, on $w$, the probability of acceptance by the resolver is $$\frac{1}{2} + \frac{1-p}{2} = 1-\frac{p}{2}<1.$$
\paragraph*{$\impliedby$} 
Suppose the resolver is \emph{not} an almost-sure resolver. This implies that there is a word $w$ in the language of the automaton $\Bc$ for which the probability of the run produced by a resolver is $p'<1$. Since the language is $\$(a+b)\Sigma^\omega$, the first two letters of the word should be $\$$ followed by $a$ or $b$. Without loss of generality, we assume the word is $\$aw$. 
The probability of a run chosen according to the resolver ends in state $q_0$ after reading $\$ a$ is $\frac{1}{2}$, and the probability that the run ends at state $s_{\text{fin}}$ after $\$ a$ is $\frac{1}{2}$. Let the probability of the resolver resulting in an accepting run from $q_0$ be $p$. Since all words in $\Sigma^\omega$ are accepting from $s_{\text{fin}}$, it must be that the probability $p'$ of the resolver producing an accepting run on $\$aw$ is $$\frac{p}{2} + \frac{1}{2} = \frac{1+p}{2}.$$
Since $p'<1$, we know $1+\frac{p}{2}<1$ and therefore $p<1$. Observe that from our construction, the probability $1-p$ is also the probability of the probabilistic B\"uchi automaton $\Bc$ accepting the word $w$. Therefore, there is a word $w$, such that the probability that $w$ is accepted with probability $1-p$, which is positive. 
\end{proof}
\lemmaUndecidableBuchi*

To prove \cref{lemma:UndecidableBuchi}, we will use an undecidability result from probabilistic automaton over finite words. A probabilistic finite automaton, or PFA for short, has similar syntax to an NFA, except that each transition is assigned a probability $0<p<1$, such that for each state and letter, the sum of probabilities of outgoing transitions from that state on that letter is $1$. The semantics of how runs work in PFA can be defined similarly to how we defined for probabilistic parity automata in \cref{sec:prelims}. We represent a PFA by a tuple $\Pc=(Q,\Sigma,\Delta,F,\rho,q_0)$, where $Q$ is the set of states, $\Sigma$ is the alphabet, $\Delta \subseteq Q \times \Sigma \times Q$ is the set of transitions, $F$ is the set of accepting states,  $\rho:\Delta \xrightarrow{} \unitinterval$ is the probability function, and $q_0$ is the initial state.

For a PFA $\Pc$ and a finite word $u$, we denote the probability that a random run of $\Pc$ on $u$ is accepting by $\prob_\Pc (u)$. 

\begin{proof}[Proof of \cref{lemma:UndecidableBuchi}]
As mentioned, we reduce from the zero-isolation problem.

\begin{question}[Zero-Isolation Problem] 
    Given a probabilistic automaton  $\Pc$ over finite words, is there a value $c>0$ such that for all words, the if the word has a non-zero probability of acceptance, then it is accepted by $\Pc$ with probability larger than $c$, that is, for all $w\in\Sigma^*$, is it true that $\prob_\Pc(w) =0 $ or $\prob_\Pc(w) > c$.
\end{question}   
The zero-isolation problem is undecidable~\cite[Theorem 4]{GO10}. 
We detail the reduction of the construction below and later prove the correctness afterwards. 
\paragraph*{Construction}
Let $\Pc$ be an PFA. We will construct a B\"uchi automaton $\Bc$ and a resolver $\Rc$ for $\Bc$ such that $\Rc$ is an almost-sure resolver of $\Bc$ if and only if $\Pc$ isolates zero. 

Let the set of all words accepted by $\Pc$ with non-zero probability be written as $\Lc_{>0}(\Pc)$. We will construct a B\"uchi automaton that accepts the language $\Lc_{>0}(\Pc)$ repeated with a symbol $\$$ as a separator between  two words from the language $\Lc_{>0}(\Pc)$, that is,  $$\Lc(\Bc) = \{w_1\$w_2\$\dots \$w_k\$\dots \mid w_i\in L_{>0}(\Pc)\}.$$
We will construct $\Bc$ with the same set of states as the finite probabilistic automaton $\Pc$. 
The transitions of $\Bc$ are exactly the transitions from the probabilistic automata $\Pc$, where any probabilistic transition with non-zero probability are made into non-deterministic choices for $\Bc$. 
Finally, we add more transitions from all states of $\Pc$ to the starting state $q_0$ of $\Pc$ on the letter $\$$.
We need to assign priorities to the transitions of $\Bc$ to complete our construction of $\Bc$. We assign priority $2$ to the transitions on the letter $\$$ that go from a final state of $\Pc$ to $q_0$. We assign a priority of $1$ to all other transitions on $\Bc$.
The resolver $\Rc$ is a memoryless resolver that at state $q$ on letter $a$ uses the same probability distribution as dictated by the transition from $q$ on $a$ in $\Pc$.

Let $\Pc = \tpl{Q,\Sigma,\Delta,\rho,q_0}$. We construct a B\"uchi automaton $\Bc = (Q,\Sigma\uplus\{\$\},\Delta',q_0)$, where $\Delta' = \{q\xrightarrow{\$:1}q_0\mid q\notin F\}\cup \{q\xrightarrow{\$:2}q_0\mid q\in F\} \cup  \{q\xrightarrow{a:1}q_0\mid q\xrightarrow{a(p)}q_0\in\Delta, p>0\}$. An illustration is provided in \cref{fig:PBAfromPFA}. 
\begin{figure}
    \centering
    \begin{tikzpicture}[shorten >=1pt, node distance=2cm, on grid, auto]
   \node[state, initial, initial text=] (q_0) {$q_0$};
   \node[state, accepting, right=of q_0] (q_1) {$q_{f}^1$};
   \node[state, accepting, below=of q_1] (q_2) {$q_f^t$};
   \node[state, below=of q_0, yshift =6mm] (q_3) {$q_2$};
   \node (q_4) at (2.2,1)  {$F$};

   \path[->]
    (q_1) edge [bend right=70,double] node [above] {$\$$} (q_0)
    (q_2) edge [bend right=10, double] node [above] {$\$$} (q_0)
    (q_3) edge [bend left=30] node [left] {$\$$} (q_0);
    
   \draw[thick, rounded corners] 
        ($(q_0.north west)+(-0.3,0.3)$) rectangle ($(q_2.south east)+(0.2,-0.3)$);
 
   \draw[dashed] 
        ($(q_0.north east)+(1,0.3)$) -- ($(q_0.south east)+(1,-2.3)$);
\end{tikzpicture}
    \caption{Converting PFA $\Pc$ to a B\"uchi automaton $\Bc$ by adding transitions on $\$$. Other transitions in $\Bc$ include all transitions of $\Pc$ that have non-zero probability. }
    \label{fig:PBAfromPFA}
\end{figure}
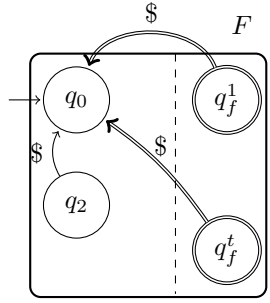
Further, consider the resolver $\Rc$ that chooses the non-deterministic transitions $q\xrightarrow{\$:1}q_0$ with probability $\rho(q\xrightarrow{a}q_0)$ that is the same as the probability assigned to the transition by the automaton $\Pc$. Observe that language of $\Bc$ is words of the word $w_1 \$ w_2 \$ w_3 \$ \dots$, where each $w_i \in \Sigma^{*}$ and infinitely many $w_i$ are in $L_{>0}(\Pc)$.

$$\Lc(\Bc) = ( (\Sigma^{*} \$) (\Lc_{>0} (\Pc))\$ )^{\omega}.$$

We show that the resolver $\Rc$ is an almost-sure resolver  if and only if there is a value $c$ such that $\prob_\Pc(w)>c$ for all words $w$. 
\paragraph*{$\implies$} Suppose $\Pc$ is a probabilistic automaton that does not isolate $0$. Then, for all  $n\in\Nb$ there is a finite word $w_n$, such that $0<\prob_\Pc(w_n)<1/{2^n}$. Consider the infinite word $w_1\$w_2\$w_3\$\dots\$ w_n\$\dots$. This word is in $\Lc(\Bc)$. However, we argue that the resolver $\Rc$ only produces an accepting run with probability $0$. Indeed, the probability of the resolver visiting a transition with priority 2 on the segment of the word $w_i\$$ is exactly the same as the probability of acceptance of $w_i$ in $\Pc$, which is a value between $0$ and $1/2^n$. 
Observe that $\sum_i^\infty \prob_\Pc(w_i)<\sum_i 1/{2^i} < 1$.
By the Borel-Cantelli lemma (\cref{lemma:borellcantelli}), if the infinite sum of the probabilities of events over an infinite sequence of events is finite, then the probability that infinitely many of the events occur is $0$. Therefore, the probability of a run on the word $w_1\$w_2\$w_3\$\dots\$ w_n\$\dots$ containing infinitely many accepting transitions is $0$.
\paragraph*{$\impliedby$}
Suppose $\Pc$ is a probabilistic automaton that isolates $0$, that is, there is a $c>0$, such that for all $w\in\Sigma^*$, either $\prob_\Ac(w) =0 $ or $\prob_\Ac(w) > c$. 
Consider any word in the language of $\Bc$ which is  $((\Sigma^*\$)^*R\cdot\$)^\omega$. This implies it consists of infinitely many substrings (disjoint) of the form $\$w_i\$$ where $\prob_\Pc(w_i) > 0$. Since $\Pc$ has an isolated $0$, we can further claim that $\prob_\Pc(w_i) > c$. Therefore for each $i$, the probability a run resolved using resolver $\Rc$ visiting an accepting state on reading $\$w_i\$$ is at least $c$. From the second Borel-Cantelli lemma (\cref{lemma:secondborellcantelli}), we know if the sum of probabilities of an infinite sequence of independent events is infinite, then the probability of infinitely many of the events occurring is $1$.
For each $i$, the events $E_i$ where the run obtained from the resolver $\Rc$ on reading $w_i$ contains an accepting transition are all independent events. 
Furthermore, the probability of each event $E_i$ is $\prob_\Pc(w_i)$ and since $\sum_i \prob_\Pc(w_i) > \infty$, we can conclude that the resolver produces an accepting run with probability $1$.
\end{proof}

\end{document}